\documentclass[conference,10pt,final]{IEEEtran}
\usepackage{cite}
\ifCLASSINFOpdf
    \usepackage[pdftex]{graphicx}
\else
\fi
\usepackage{amsmath}
\usepackage{amssymb}
\usepackage{amsthm}
\usepackage{amsfonts}
\usepackage{mathtools}

\usepackage{array}
 \usepackage[caption=false,font=footnotesize]{subfig}
\theoremstyle{plain}
\newtheorem{theorem}{Theorem}
\newtheorem{lemma}[theorem]{Lemma}
\newtheorem{corollary}[theorem]{Corollary}
\newtheorem{proposition}[theorem]{Proposition}
\theoremstyle{definition}
\newtheorem{definition}{Definition}
\newtheorem{example}{Example}

\theoremstyle{remark}
\newtheorem*{remark}{Remark}

\theoremstyle{remark}
\newtheorem*{notation}{Notation}

\newcommand{\exec}{\rho}
\DeclarePairedDelimiter{\abs}{\lvert}{\rvert}
\newcommand{\integers}{\mathbb{Z}}
\newcommand{\Rats}{\mathbb{Q}}
\newcommand{\Reals}{\mathbb{R}}
\newcommand{\Nats}{\mathbb{N}}

\newcommand{\Lap}[1]{\mathsf{Lap}{(#1)}}

\newcommand{\euler}{e}
\newcommand{\eulerv}[1]{\ensuremath{\euler^{#1}}}

\newcommand{\Prob}{\mathsf{Prob}}

\newcommand{\st}  {\mathbin{|}}
\newcommand{\set}[1]{\{#1\}}

\newcommand{\card}[1]{\mathbin{|}#1 \mathbin{|}}
\newcommand{\true}{\mathsf{true}}
\newcommand{\false}{\mathsf{false}}
\newcommand{\cA}{\mathcal{A}}

\newcommand{\rv}{\mathsf{r}}

\newcommand{\prbfn}[1] {\Prob[#1]}

\newcommand{\tuple}[1] {\langle #1 \rangle}
\newcommand{\s}[1] {\mathsf{#1}}

\newcommand{\pto} {\hookrightarrow}
\newcommand{\rmv}[1] {}
\newcommand{\dom} {\mathsf{dom}}

\newcommand{\sgn}{\s{sgn}}

\newcommand{\diptext} {DiP}
\newcommand{\dipautos} {{\diptext} automaton}
\newcommand{\dipautop} {{\diptext} automata}
\newcommand{\dipa} {{\diptext}A}
\newcommand{\qinit} {q_{\mathsf{init}}}
\newcommand{\inalph} {\Sigma}
\newcommand{\outalph} {\Gamma}
\newcommand{\parf} {P}
\newcommand{\vars} {X}
\newcommand{\svar} {\mathsf{insample}}
\newcommand{\rvar} {\mathsf{x}}
\newcommand{\states}{Q}
\newcommand{\instates}{Q_{\mathsf{in}}}
\newcommand{\epsstates}{Q_{\mathsf{non}}}

\newcommand{\cnds} {C}
\newcommand{\getest} {\svar \geq \rvar}
\newcommand{\lttest} {\svar < \rvar}
\newcommand{\defaut} {(\states, \inalph, \outalph, \qinit, \vars, \parf, \transf)}
\newcommand{\emptystr}{\tau}

\newcommand{\pathprob}[1] {\mathsf{Pr}[#1]}
\newcommand{\ith}[2][i]{#2[#1]}
\newcommand{\Detcond} {Determinism}
\newcommand{\detcond} {determinism}
\newcommand{\Outcond} {Output Distinction}
\newcommand{\outcond} {output distinction}
\newcommand{\Initcond} {Initialization}
\newcommand{\initcond} {initialization}
\newcommand{\Noninpcond} {Non-input transition}

\newcommand{\trns}[1]  {\xrightarrow{#1}}
\newcommand{\defexec} {q_0 \trns{a_0,o_0} q_1 \trns{a_1,o_1} q_2 \cdots q_{n-1} \trns{a_{n-1},o_{n-1}} q_n}
\newcommand{\len}[1] {\card{#1}}
\newcommand{\fstst} {\mathsf{first}}
\newcommand{\lstst} {\mathsf{last}}
\newcommand{\tl} {\mathsf{tail}}
\newcommand{\inseq} {\mathsf{inseq}}
\newcommand{\outseq} {\mathsf{outseq}}
\newcommand{\trname} {\mathsf{trans}}
\newcommand{\stname} {\mathsf{state}}
\newcommand{\grdname} {\mathsf{guard}}

\newcommand{\svtauto} {\cA_{\s{SVT}}}
\newcommand{\sortauto} {\cA_{\s{sort}}}
\newcommand{\svtpauto} {\cA_{\s{SVT*}}}
\newcommand{\numspauto} {\cA_{\s{NumSp}}}

\newcommand{\transf}{\delta}

\newcommand{\lcycle}{\ensuremath{\mathsf{L}}-cycle}
\newcommand{\gcycle}{\ensuremath{\mathsf{G}}-cycle}
\newcommand{\alpath}{\ensuremath{\mathsf{AL}}-path}
\newcommand{\agpath}{\ensuremath{\mathsf{AG}}-path}
\newcommand{\critical}{{leaking}}

\newcommand{\criticalpath}{{leaking path}}
\newcommand{\criticalcycle}{{leaking cycle}}
\newcommand{\criticalpair}{{leaking pair}}

\newcommand{\Criticalcycle}{{Leaking cycle}}
\newcommand{\Criticalpair}{{Leaking pair}}

\newcommand{\criticaltransition}{{critical transition}}
\newcommand{\criticalintransition}{{critical input transition}}
\newcommand{\criticalnotransition}{{critical non-input transition}}

\newcommand{\violatingc}{{disclosing cycle}}
\newcommand{\violatingp}{{privacy violating path}}
\newcommand{\Violatingc}{{Disclosing cycle}}
\newcommand{\Violatingp}{{Privacy violating path}}

\newcommand{\execl}[1] {q_0 \trns{a_0,o_0} q_1 \trns{a_1,o_1}  q_2 \cdots q_{#1-1} \trns{a_{#1-1},o_{#1-1}} q_{#1}}

\newcommand{\execlb}[1] {q_0 \trns{b_0,o_0} q_1 \trns{b_1,o_1}  q_2 \cdots q_{#1-1} \trns{b_{#1-1},o_{#1-1}} q_{#1}}

\newcommand{\abst}{\ensuremath{\mathsf{abstract}}}
\newcommand{\execsf}[2] {q_{#1} \trns{a_{#1},o_{#1}} q_{#1+1} \trns{a_{#1+1},o_{#1+1}} q_{#1+2} \cdots q_{#2-1} \trns{a_{#2-1},o_{#2-1}} q_{#2}}

\newcommand{\eabsexecl}[1] {q_0 \sigma_0 q_1 \sigma_1  \cdots q_{#1-1}\sigma_{#1-1} q_{#1}}

\newcommand{\eabsexecsf}[2] {q_{#1} \sigma_{#1} q_{#1+1} \sigma_{#1+1}  \cdots q_{#2-1} \sigma_{#2-1}q_{#2}}

\newcommand{\absexec}{\eta}
\newcommand{\inalphaseq}{\alpha}

\newcommand{\inbetaseq}{\beta}
\newcommand{\inalpha}{a}
\newcommand{\inbeta}{b}
\newcommand{\outgammaseq}{\gamma}

\newcommand{\weight}[1]{\mathsf{wt}(#1)}
\newcommand{\cost}[1]{\mathsf{cost}(#1)}

\newcommand{\cG}{\mathcal{G}}

\makeatletter
\newcommand{\removelatexerror}{\let\@latex@error\@gobble}
\makeatother

\newcommand{\blue}[1]{{#1}}

\usepackage{tikz}
\usetikzlibrary{shapes}
\usetikzlibrary{shapes.multipart}
\usetikzlibrary{automata}
\usetikzlibrary{positioning}
\usetikzlibrary{arrows}
\def\mystrut{\vrule height .2cm depth .2cm width 0pt}
\tikzset{
         node distance=3.3cm,
         mynonstate/.style={
           rectangle split,
           rounded corners=0.4cm,
           rectangle split parts=2,
           draw,
           semithick,
           minimum size=1.4cm
         },
         myinstate/.style={
           circle split,
           draw,
           semithick,
         },
         myedgelabel/.style={
           rectangle split,
           rectangle split parts=2
         },
         initial text={},
         every edge/.style={
           draw,
           ->,>=stealth',
           auto,
           semithick
         }
}
\usepackage{xcolor}
\usepackage{algorithm2e}
\usepackage{breqn}

\newcommand*{\AppendixTrue}{}

\begin{document}
%\title{Deciding Differential Privacy for Programs with Unbounded Inputs}
\title{On Linear Time Decidability of Differential Privacy for Programs with Unbounded Inputs}

\author{\IEEEauthorblockN{Rohit Chadha}
\IEEEauthorblockA{Email: chadhar@missouri.edu}
\and
\IEEEauthorblockN{A. Prasad Sistla}
\IEEEauthorblockA{Email: sistla@uic.edu}
\and
\IEEEauthorblockN{Mahesh Viswanathan}
\IEEEauthorblockA{
Email: vmahesh@uiuc.edu}}


%Telephone: (800) 555--1212, Fax: (888) 555--1212}


\IEEEoverridecommandlockouts
\IEEEpubid{\makebox[\columnwidth]{978-1-6654-4895-6/21/\$31.00~
\copyright2021 IEEE \hfill} \hspace{\columnsep}\makebox[\columnwidth]{ }}

\maketitle
 \IEEEpeerreviewmaketitle

\begin{abstract}
We introduce an automata model for describing interesting classes of differential privacy mechanisms/algorithms that include known mechanisms from the literature. These automata can model algorithms whose inputs can be an unbounded sequence of real-valued query answers. We consider the problem of checking whether there exists a constant $d$ such that the algorithm described by these automata are $d\epsilon$-differentially private for all positive values of the privacy budget parameter $\epsilon$. We show that this problem can be decided in time linear in the automaton's size by identifying a necessary and sufficient condition on the underlying graph of the automaton. This paper's results are the first decidability results known for algorithms with an unbounded number of query answers taking values from the set of reals.
\end{abstract}


% \fi
%\IEEEpeerreviewmaketitle

\section{Introduction}
% !TEX root =  ms.tex
Differential privacy~\cite{dmns06,DR14} is a technique developed to preserve individuals' privacy while performing statistical computations on databases containing private information. The differential privacy framework trades accuracy for privacy. In the framework, a differential privacy mechanism mediates data exchange between the database and data analyst. When the mechanism returns the answer to an analyst's query, it introduces \emph{random} noise in the query result before forwarding it to the analyst. The mechanism is parameterized by a \emph{privacy budget parameter $\epsilon$}, and the noise added depends on this parameter. The privacy guarantees are also stated in terms of $\epsilon$ --- a mechanism is said to be \emph{$d\epsilon$-differentially private} if the probability of observing a given output on two adjacent databases differ only up-to a factor of $\euler^{d\epsilon}$, where $d>0$ is a constant and $\euler$ is the Euler's constant.  Setting $\epsilon$ allows the database manager to choose the trade-off between accuracy and privacy. Intuitively, smaller values of $\epsilon$ imply improved privacy guarantees but at the cost of increased inaccuracy in the observed output.

Designing correct differential privacy mechanisms is subtle and error-prone, and even relatively minor tweaks to correct mechanisms can lead to loss of privacy as evidenced by the Sparse Vector Technique (SVT)~\cite{DNRRV09,lyu2016understanding}. This difficulty has generated interest in formally verifying the privacy claims of differential privacy mechanisms. Verifying differential privacy is challenging for several reasons. First, the behavior of a privacy mechanism changes with $\epsilon$ as the random noise employed by the mechanism is parameterized by $\epsilon$. The privacy guarantees are usually required to hold for all $\epsilon>0$ to allow a manager to choose the trade-off between privacy and accuracy. Thus, the verification problem is inherently parametric. Secondly, the random noise employed by a mechanism typically samples from the continuous (or discrete) Laplace distribution.  Thus, verification involves the analysis of an infinite-state stochastic model, even when inputs are constrained to come from a finite set. Finally, the mechanisms may need to process a potentially unbounded sequence of query answers, each of which may take any real value. Verification of differential privacy is known to be undecidable even when the mechanisms operate on a bounded sequence of query answers, each of which takes value from a finite domain~\cite{BartheCJS020}.

Three major directions of research seek to circumvent this challenge. The first direction aims to develop automated and semi-automated techniques to construct privacy proofs~\cite{RP10,GHHNP13,BKOZ13,BGGHS16,BFGGHS16,ZK17,AGHK18,AH18,WDWKZ19,CheckDP}. These techniques are not guaranteed to be complete and may fail to construct a proof even if the mechanism is differentially private. The second line of investigation develops automated techniques to search for privacy violations~\cite{DingWWZK18,BichselGDTV18} and searches amongst a bounded sequence of inputs. The third direction explores decision procedures for verifying differential privacy~\cite{BartheCJS020}. To circumvent the undecidability result, \cite{BartheCJS020} considers mechanisms that sample from Laplacians only a bounded number of times and process (only) a bounded sequence of query answers, each of which is finite valued. Outputs of these mechanisms are also constrained to take values from a finite domain. The decision procedure developed in~\cite{BartheCJS020} converts the problem of checking differential privacy to checking the validity of first-order formulas in the theory of Reals with the exponential function. While the decidability of validity for the theory of Reals with exponential function is a longstanding open problem, formulas obtained in~\cite{BartheCJS020} fall into the decidable fragment identified by~\cite{mccallum2012deciding}. Unfortunately since it relies on the decision procedure for real arithmetic, the verification algorithm has very high complexity.

\paragraph*{\textbf{Contributions}} In this paper, we present the first decision procedure for checking differential privacy for mechanisms that process an \emph{unbounded} sequence of inputs, each of which may be real valued. Further, the mechanisms may also output real values in addition to values from a finite domain. In order to obtain decidability, we make two choices. First, we restrict mechanisms to those that can be modeled by a particular automata class, which we call {\dipautop}. Several mechanisms proposed in the literature, such as SVT and its variants~\cite{DNRRV09,lyu2016understanding} and NumericSparse~\cite{DR14} can be modeled by {\dipautop}. Our decision procedure is sound and complete for mechanisms modeled by such automata, and remarkably, runs in time linear in the size of the automaton. Second, we consider the following verification problem. Instead of asking whether a mechanism is $d \epsilon$ differentially private for a given constant $d>0$ and for all $\epsilon>0$, we ask whether there exists a constant $d$ such that the mechanism is $d \epsilon$ differentially private for all $\epsilon>0$. While the verification problem considered in this paper may appear to be less useful, note that a database manager can choose a lower $\epsilon$ to account for a higher $d$ if the mechanism turns out to be differentially private. The relationship between the computation difficulty of checking $d\epsilon$-differential privacy for a given $d$ and checking if there is some $d$ such that a mechanism is $d\epsilon$-differentially private is unclear. For example, the decidability results in~\cite{BartheCJS020} do not extend to the verification problem we consider in this paper.

We briefly describe the {\dipautop} model introduced in this paper to model differential privacy mechanisms. A {\dipautos} (\dipa) $\cA$ takes arbitrarily long sequences of real-valued query results. Control states of $\cA$ are classified into input and non-input states. The automaton also has a single variable $\rvar$ in which it can store a real value. When the automaton is in an input state, it reads an input value and generates a value, $\svar$, using a Laplace distribution, and compares $\svar$ with the stored value of $\rvar.$ It changes state depending on the result of comparison and outputs a value during the state transition. During the transition, it may also store the sampled value $\svar$ in $\rvar.$ When the automaton is in a non-input state, it does not read an input, but generates $\svar$ using constant parameters and  resets $\rvar$ by storing $\svar$ in $\rvar$ and transitions to a new control state. The state transition's output may be either a discrete value from a finite domain or a real value. The real value could be sampled value $\svar$, or freshly sampled value $\svar'$. The mean and scaling factor of the Laplace distributions used for generating the sampled values $\svar$ and $\svar'$ are determined by \blue{the budget parameter $\epsilon$ and} by constants that depend only on the state. Additionally, for input states, the input value is added to the mean. 

%is, respectively given, by $\mathsf{input}+\mu_q$ and $d_q\epsilon$ for transitions from input states $q$, and are given by $\mu_q$ and $d_q\epsilon$, respectively, for transitions from non-input states; here $\mathsf{input}$ is the input value, $d_q,\mu_q$ are constants that depend on the state. For $\svar',$ the mean and the scaling factor are given by  $\mathsf{input}+\mu'_q$ and $d'_q\epsilon$ for constants $\mu'_q$ and $d'_q$ that depend on $q,$

Surprisingly, we show that the problem of checking whether a privacy mechanism, specified by a {\dipa} $\cA$, is $d\epsilon$-differentially private, for some constant $d>0$ and all $\epsilon>0$, can be reduced to checking some syntactic graph-theoretic conditions on the finite graph \lq\lq underlying\rq\rq $\cA.$ These syntactic conditions are stated as the absence of certain kinds of cycles and paths (See Definition~\ref{def:well-formed} on Page~\pageref{def:well-formed}). These conditions can be checked in time linear in the graph's size by constructing the graph of strongly connected components of the \lq\lq underlying\rq\rq control flow graph. These conditions are independent of the scaling factors and means associated with sampling, and hence, differential privacy does not need to be re-proved if these parameters change.

Furthermore, if the privacy mechanism under consideration is differentially private, we can efficiently compute a constant $d$ using the graph of strongly connected components, such that the mechanism is $d\epsilon$-differentially private for all values of $\epsilon>0.$ The computed $d$ depends on the scaling parameters of states in $\cA$ used when sampling. The computation of the constant $d$ is once again linear, assuming constant time addition and comparison of numbers. We also observe that $d$ computed by our algorithm for SVT and NumericSparse match those known in literature.  
%We show that a privacy mechanism, specified by a \dipa~ automaton, is $a\epsilon$-differentially private, for some constant $a>0$, and for all $epsilon>0$, iff the finite graph associated with the automaton has no reachable {\it leaking} cycle and has no reachable {\it leaking} pair of cycles. Thus, checking differential privacy is reduced to that of checking  existence of a reachable leaking cycle and a reachable leaking pair of cycles. This property can be checked using standard graph algorithms in time linear in the size of the graph. Further more, if the privacy mechanism under consideration is differentially private, we can efficiently compute a constant $a$ using the graph of strongly connected components, such that the mechanism is $a\epsilon$-differentially private for all values of $\epsilon>0.$ 

%We also extend the above automaton model, so that it can output either the generated sampled input value during a transition, or it can output a freshly sampled input value. We extend the previous syntactic condition on the graph associated with the automaton, with additional constraints that can be checked in time linear in the size of the graph using standard graph algorithms. The extended automaton can output values from a finite domain as well as values that are real sampled values as described above. We prove that mechanisms described by the extended automata are differentially private iff the associated graph satisfies the extended constraints. The well know NumericSparse algorithm can be modeled using the extended automata. 
%   In both models of the automata,
The proof that the given syntactic graph conditions are necessary and sufficient for differential privacy is highly non-trivial. To the best of our knowledge, these results are the first results giving efficient algorithms for checking differential privacy of interesting classes of mechanisms that process input query sequences of unbounded length, where the query values are real-valued, and the outputs may take real values.  
   
\paragraph*{\textbf{Organization}} The rest of the paper is organized as follows. Section~\ref{sec:prelims} introduces basic notation and the setup of differential privacy. Our model of {\dipautop} is introduced in Section~\ref{sec:dipauto}. The main results characterizing when a {\dipautop} is differentially private are presented in Section~\ref{sec:decidability}. \ifdefined\AppendixTrue
Because of their length, proofs of our main theorem are deferred to the Appendix. 
\fi
Related work is discussed in Section~\ref{sec:related}. Finally we present our conclusions (Section~\ref{sec:conclusions}).
\blue{\ifdefined\AppendixTrue
An extended abstract of this paper appeared in the 36th Annual IEEE Symposium on Logic in Computer Science (LICS 2021)~\cite{ChadhaSV21}. This version consists of proofs omitted in~\cite{ChadhaSV21}.
\else
The omitted proofs can be located in~\cite{ChadhaSV21b} FIX THE REFERENCE.
\fi}

% !TEX root =  main.tex

\section{Preliminaries}
\label{sec:prelims}

\paragraph*{Sequences}  For a set $\inalph$, $\inalph^*$ denotes the set of all finite sequences/strings over $\inalph$. We shall use $\emptystr$ to denote the empty sequence/string over $\inalph$. For two sequences/strings $\rho,\sigma \in \inalph^*$, we use their juxtaposition $\rho\sigma$ to indicate the sequence/string obtained by concatenating them in order. Consider $\sigma = a_0a_1\cdots a_{n-1} \in \inalph^*$ (where $a_i \in \inalph$). We use $\len{\sigma}$ to denote it's length $n$ and use $\ith{\sigma}$ to denote its $i$th symbol $a_i$.

\paragraph*{Sets and functions} Let $\Nats, \integers,\Rats,\Rats^{\geq 0}, \Reals, \Reals^{>0}$ denote the set of natural numbers, integers, rational numbers, non-negative rationals, real numbers and positive real numbers, respectively. In addition, $\Reals_\infty$ will denote the set $\Reals \cup \set{-\infty,\infty}$, where $-\infty$ is the smallest and $\infty$ is the largest element in $\Reals_\infty$. For a real number $x \in \Reals$, $\card{x}$ denotes its absolute value, and $\sgn(x)$ denotes the \emph{sign} function, i.e., $\sgn(x) = 0$ if $x=0$, $\sgn(x) = -1$ if $x < 0$ and $\sgn(x) = 1$ if $x > 0.$%and for any $m\leq \card{\sigma}$,
%we let $\sigma| m$ denote the prefix of $\sigma$ of length $m$, and
%$\sigma||m$ denote the  suffix  of
%$\sigma$ starting immediately after the prefix $\sigma|m.$
For any partial function $f\::A \pto B$, where $A,B$ are some sets, we let $\dom(f)$ be the set of $x\in A$ such that $f(x)$ is defined.

%We denote the set of real numbers, rational numbers, and integers  by $\Reals,\Rats$, and $\integers$ respectively.  

\paragraph*{Laplace Distribution}
Differential privacy mechanisms often add noise by sampling values from the \emph{Laplace distribution}. The distribution, denoted $\Lap{k, \mu}$, is parameterized by two values --- $k \geq 0$ which called the scaling parameter, and $\mu$ which is the mean. The probability density function of $\Lap{k, \mu}$, denoted $f_{k,\mu}$, is given by
\[
f_{k,\mu}(x) = \frac{k}{2}\eulerv{-k\card{x-\mu}}.
\]
Therefore, for a random variable $X \sim \Lap{k,\mu}$ and $c \in \Reals$, we have
\[
\prbfn{X \leq c} = \frac{1}{2}\left[1 + \sgn(c-\mu)(1 - \eulerv{-k\card{c-\mu}})\right].
\]
Finally observe that for any $\mu_1,\mu_2 \geq 0$, $\Lap{k,\mu_1+\mu_2}$ and $\Lap{k,\mu_1}+\mu_2$ are identically distributed.

% \begin{figure}
% \begin{center}
% \includegraphics[width=8cm]{dip.png}
% \end{center}
% \caption{Schematic diagram of Differential Privacy.}
% \label{fig:dip-scheme}
% \end{figure}

\paragraph*{Differential Privacy} 
\blue{Differential privacy~\cite{dmns06} is a framework that enables statistical analysis of databases containing sensitive, personal information of individuals while ensuring that individuals in the database are not adversely affected by the results of the analysis. In the differential privacy framework, %(schematically shown in Fig.~\ref{fig:dip-scheme}), 
a randomized algorithm, $M$, called the \emph{differential privacy mechanism} %(block $M$ in Fig.~\ref{fig:dip-scheme}), 
mediates the interaction between a (possibly dishonest) data analyst asking queries and a database $D$ responding with answers. 
%A query $q$
%is a deterministic (computable) function, and the database's response is denoted by $q(D)$. 
Queries are  deterministic functions and typically include aggregate questions about the data, like the mean, median, standard deviation of fields in the database. In response to such a sequence of queries, the differential privacy mechanism $M$ will respond with a series of answers, whose value is computed using the actual answers and random sampling, resulting in \lq \lq noisy \rq \rq answers. 
% In response to such a query $q$ (or sequence of queries), the differential privacy mechanism will respond with an answer $M(q(D))$, whose value is computed using some ``noisy'' version of the answers $q(D)$ as well noisy versions of answers to the earlier queries. 
Thus, the differential privacy mechanism provides privacy at the cost of accuracy. Typically, the differential privacy mechanism's  noisy  response depends on a \emph{privacy budget} $\epsilon > 0$.} %Intuitively, smaller values of $\epsilon$ imply better privacy guarantees but more inaccuracy in the observed output.

\blue{The crucial definition of differential privacy captures the privacy guarantees of individuals in the database $D$. For an individual $i$ in $D$, let $D \setminus \set{i}$ denote the database where $i$'s information has been removed
A secure mechanism $M$ ensures that for any individual $i$ in $D$, and any sequence of possible outputs $\overline{o}$, the probability that $M$ outputs $\overline{o}$ on a sequence of queries is approximately the same whether the interaction is with the database $D$ or with $D \setminus \set{i}.$}
%, i.e., the database where $i$'s information has been removed.
% A secure mechanism $M$ ensures that for any individual $i$ in $D$, and any sequence of possible outputs $\overline{o}$, the probability that the mechanism outputs $\overline{o}$ on a sequence of queries is approximately the same whether the interaction is with the database $D$ or with ``$D \setminus \set{i}$'', i.e., the database where $i$'s information has been removed. 
To capture this definition formally, we need to characterize the inputs on which $M$ is required to behave similarly.
%As shown schematically in Fig.~\ref{fig:dip-scheme}, 
Inputs to a differential privacy mechanism could be seen as answers to a sequence of queries asked by the data analyst. If queries are aggregate queries, then answers to $q$ on $D$ and $D \setminus \set{i}$, for individual $i$, are likely to be away by at most $1$. This intuition leads to an %commonly
\blue{often-}used definition of \emph{adjacency}\blue{, such as in SVT~\cite{DNRRV09,lyu2016understanding,DR14} and NumericSparse~\cite{DR14},} that characterizes pairs of inputs on which the differential privacy mechanism $M$ is expected to behave similarly. 
\begin{definition}
\label{def:adjacency}
Two sequences $\rho,\sigma \in \Reals^*$ are said to be \emph{adjacent} if $\len{\rho} = \len{\sigma}$ and for \blue{each} $i \leq \len{\rho}$, $\card{\ith{\rho} - \ith{\sigma}} \leq 1$.
\end{definition}

Having defined adjacency between inputs, we are ready to formally define the notion of privacy. In response, to a sequence of inputs, a differential privacy mechanism produces a sequence of outputs from the set (say) $\outalph$. Since a differential privacy mechanism $M$ is a randomized algorithm, it will induce a probability distribution on $\outalph^*$.

\begin{definition}[$\epsilon$-differential privacy]
\label{def:diff-priv}
A randomized algorithm $M$ that gets as input a sequence of real numbers and produces an output in $\outalph^*$ is said to be \emph{$\epsilon$-differentially private} if for all measurable sets $S \subseteq \outalph^*$ and adjacent $\rho,\sigma \in \Reals^*$ (Definition~\ref{def:adjacency}),
\[
\prbfn{M(\rho) \in S} \leq \eulerv{\epsilon}\, \prbfn{M(\sigma) \in S}.
\]
In the above equation, $\euler$ is the Euler constant.
\end{definition}

\begin{example}
\label{ex:diff-privacy}
Let us look at a couple of classical differential privacy mechanisms from the literature. These will serve as running examples to motivate our definitions and highlight our results.

\RestyleAlgo{boxed} 
\begin{algorithm}
%\removelatexerror
\DontPrintSemicolon
\SetAlgoLined

\KwIn{$q[1:N]$}
\KwOut{$out[1:N]$}
\;
$\rv_T \gets \Lap{ \frac{\epsilon}{2} , T}$\;
\For{$i\gets 1$ \KwTo $N$}
{
    $\rv\gets \Lap{\frac \epsilon {4} , q[i]}$\;
    \uIf{$\rv \geq \rv_T$}{
      $out[i] \gets \top$\;
      exit
      }
    \Else{
      $out[i] \gets \bot$}
}
\caption{SVT  algorithm}
\label{fig:SVT}
\end{algorithm}

Sparse Vector Technique (SVT)~\cite{DNRRV09,lyu2016understanding} is an algorithm to answer the following question in a privacy preserving manner: Given a sequence of query answers $q[1:N]$ and threshold $T$, find the first index $i$ such that $q[i] \geq T$. The algorithm is shown as Algorithm~\ref{fig:SVT}. It starts by sampling a value from the Laplace distribution with mean $T$, and stores this ``noisy threshold'' in the variable $\rv_T$. After that the algorithm reads query answer $q[i]$, perturbs it by sampling from the Laplace distribution with mean $q[i]$ to get $\rv$, and compares this ``noisy query'' $\rv$ with the ``noisy threshold'' $\rv_T$. If $\rv < \rv_T$ then the algorithm outputs $\bot$ and continues by reading the next query. On the other hand, if $\rv \geq \rv_T$ then the algorithm outputs $\top$ and stops. This algorithm is known to be $\epsilon$-differential private. It is worth observing that SVT is parameterized by $\epsilon$; each value of $\epsilon$ gives us a new algorithm which is $\epsilon$-differentially private for that particular value of $\epsilon$.

\RestyleAlgo{boxed} 
\begin{algorithm}
%\removelatexerror
\DontPrintSemicolon
\SetAlgoLined

\KwIn{$q[1:N]$}
\KwOut{$out[1:N]$}
\;
$\rv_T \gets \Lap{ \frac{4\epsilon}{9} , T}$\;
\For{$i\gets 1$ \KwTo $N$}
{
    $\rv\gets \Lap{\frac{2\epsilon}{9} , q[i]}$\;
    \uIf{$\rv \geq \rv_T$}{
      $out[i] \gets \Lap{\frac{\epsilon}{9}, q[i]}$\;
      exit
      }
    \Else{
      $out[i] \gets \bot$}
}
\caption{Numeric Sparse  algorithm}
\label{fig:NumSp}
\end{algorithm}

Consider Algorithm~\ref{fig:NumSp} which shows a differential privacy mechanism called Numeric Sparse~\cite{DR14}. The problem solved by this algorithm is very similar to the one solved by SVT (Algorithm~\ref{fig:SVT}) --- given a sequence of query answers $q[1:N]$ and threshold $T$, find the first index $i$ such that $q[i] \geq T$ \emph{and output $q[i]$}. Algorithm~\ref{fig:NumSp} is similar to Algorithm~\ref{fig:SVT}. The only difference is that instead of outputting $\top$ when $\rv \geq \rv_T$, it outputs a perturbed value of $q[i]$. This algorithm is also known to be $\epsilon$-differentially private for each possible assignment of value to $\epsilon$.
\end{example}

%We conclude this section with a technical lemma that characterizes the probability of two samples from Laplace distributions being ordered.
%\begin{lemma}
%\label{lem:problessequal}
%Suppose $X_i$, for $i=1,2$, are  random variables with $X_i \sim \Lap{k_i,\mu_i}$. Then $\prbfn{X_1\leq X_2}$  is given as follows. When $k_1 \neq k_2$
%\begin{dmath*}
%\prbfn{X_1 \leq X_2} = \frac{1}{2}\left[1 + \sgn(\mu_2-\mu_1)\left(1- \frac{k_2^2}{2(k_2^2-k_1^2)}\eulerv{-k_1\card{\mu_2-\mu_1}} + \frac{k_1^2}{2(k_2^2-k_1^2)}\eulerv{-k_2\card{\mu_2-\mu_1}}\right)\right].
%\end{dmath*}
%On the other hand, when $k_1=k_2=k$
%\begin{dmath*}
%\prbfn{X_1 \leq X_2} = \frac{1}{2}\left[1+\sgn(\mu_2-\mu_1)\left(1 - \eulerv{-k\card{\mu_2-\mu_1}}(1+\frac{k}{2}\card{\mu_2-\mu_1})\right)\right]
%\end{dmath*}
%\end{lemma}

%\input{Section3_dipaut}
\section{\dipautop}
\label{sec:dipauto}

{\diptext} ({\bf \textsf{Di}}fferentially {\bf \textsf{P}}rivate) automata ({\dipa} for short) are a simple model to describe some differential privacy mechanisms known in the literature. Some of the features we hope to capture are those highlighted by Algorithms~\ref{fig:SVT} and~\ref{fig:NumSp}. Recall that the input to a differential privacy mechanism is a sequence of real numbers that correspond to answers to queries. The differential privacy mechanism is a randomized algorithm that processes this input, samples values from distributions like Laplace, and produces a sequence of values as output. These outputs could include real numbers (Algorithm~\ref{fig:NumSp}). Further, as observed in Example~\ref{ex:diff-privacy}, the behavior of the mechanism depends on the privacy budget $\epsilon$. {\dipautop} are a formal model that have these features.

\subsection{Syntax}
\blue{A {\dipa} is a \emph{parametric} automaton with finitely many control states and three real-valued variables $\svar,\svar'$ and $\rvar$. While the variables $\svar$ and $\svar'$ are freshly sampled in each step, the variable $\rvar$ can store real values to be used in later steps.}
% A {\dipa} is a \emph{parametric} automaton with finitely many control states and \emph{one} storage location, where real values can be stored. The unique, single storage location of the automaton will be referred to by the variable $\rvar$. 
The value of the parameter $\epsilon$ (the privacy budget) influences the distribution from which reals values are sampled during an execution. The input to such an automaton is a finite sequence of real numbers. In each step the automaton does the following.
\begin{enumerate}
\item It samples two values, called $\svar$ and $\svar'$, drawn from the distributions $\Lap{d\epsilon, \mu}$ and $\Lap{d'\epsilon,\mu'}$, respectively. The scaling factors $d,d'$ and means $\mu, \mu'$ of these distributions depend on the current state.
\item Depending on the current state, the automaton will either read a real number from the input, or not read anything from the input. If an input value $a$ is read, then $\svar$ and $\svar'$ are updated by adding $a$ to them.
\item The transition results in changing the control state and outputting a value. The value output could either be a symbol from a finite set (like $\bot/\top$ in Algorithm~\ref{fig:SVT}) or one of the two real numbers $\svar$ and $\svar'$ that are sampled in this step (like in Algorithm~\ref{fig:NumSp}). If an input value is read then the transition could be guarded by the result of comparing the sampled value $\svar$ and the stored value $\rvar$. It is possible that for certain values of $\rvar$ and $\svar$, no transition is enabled from the current state. In such a case, the computation ends.
\item Finally, the automaton may choose to store the sampled value $\svar$ in $\rvar$.
\end{enumerate}
The above intuition is captured by the formal definition of {\dipa} below and its semantics described later in this section.

\begin{definition}[{\dipa}]
\label{def:dipa}
Let $\cnds$ be the set of \emph{guard conditions} $\set{\true,\getest,\lttest}$. A \emph{\dipautos} $\cA = \defaut$ where
\begin{itemize}
\item $\states$ is a finite set of states partitioned into two sets: the set of input states $\instates$ and the set of non-input states $\epsstates$,
\item $\inalph = \Reals$ is the input alphabet,
\item $\outalph$ is a finite output alphabet,
\item $\qinit \in \states$ is the initial state,
\item $\vars = \set{\rvar,\svar,\svar'}$ is the set of variables,
\item $\parf : \states \to \Rats^{\geq 0} \times \Rats \times \Rats^{\geq 0} \times \Rats$ is the parameter function that assigns to each state a 4-tuple $(d,\mu,d',\mu')$, where $\svar$ is sampled from $\Lap{d\epsilon,\mu}$ and $\svar'$ is sampled from $\Lap{d'\epsilon,\mu'}$,
\item and $\transf: (\states \times \cnds) \pto (\states \times (\outalph \cup \set{\svar,\svar'}) \times \set{\true,\false})$ is the transition (partial) function that given a current state and result of comparing $\rvar$ with $\svar$, determines the next state, the output, and whether $\rvar$ should be updated to store $\svar$. The output could either be a symbol from $\outalph$ or the values $\svar$ and $\svar'$ that were sampled.
\end{itemize}
The transition function $\transf$ of a {\dipa} will satisfy the following four conditions.

\vspace*{0.05in}
\noindent
{\bf \Detcond:} For any state $q \in \states$, if $\transf(q,\true)$ is defined then $\transf(q,\getest)$ and $\transf(q,\lttest)$ are undefined.

\vspace*{0.05in}
\noindent
{\bf \Outcond:} For any state $q \in \states$, if $\transf(q,\getest)$ is defined to be $(q_1,o_1,b_1)$ and $\transf(q,\lttest)$ is defined to be $(q_2,o_2,b_2)$ then $o_1 \neq o_2$, i.e., distinct transitions from a state have different outputs. Further at least one out of $o_1$ and $o_2$ belongs to $\outalph$, i.e., both transitions cannot output real values.

\vspace*{0.05in}
\noindent
{\bf \Initcond:} The initial state $\qinit$ has only one outgoing transition of the form $\transf(\qinit,\true) = (q,o,\true)$ where $q$ is a state and $o$ is an output symbol. In other words, the guard of the first transition is always $\true$ and the first value sampled is stored in $\rvar$.

\vspace*{0.05in}
\noindent
{\bf \Noninpcond:} From any $q \in \epsstates$, if $\transf(q,c)$ is defined, then $c = \true$; that is, there is at most one transition from a non-input state which is always enabled.
\end{definition}

It is useful to classify transitions of a {\dipa} into different types. Consider a transition $\transf(q,c) = (q',o,b)$. If $q \in \instates$ then it is an \emph{input transition} and if $q \in \epsstates$ then it is a \emph{non-input transition}. If $b = \true$ then the transition will set $\rvar = \svar$, and hence it is called an \emph{assignment transition}. On the other hand, if $b = \false$, the transition will be said to be a \emph{non-assignment} transition. A \emph{pure assignment} transition is an assignment transition with $c = \true$. The {\initcond} condition says that the (only) transition out of the initial state of a {\dipa} is a pure assignment transition.

\begin{example}
\label{ex:autos}
The differential privacy mechanisms in Example~\ref{ex:diff-privacy} can be modeled as {\dipautop}. These are shown in Fig.~\ref{fig:svt-auto} and~\ref{fig:numsp-auto}. When drawing {\dipa}s in this paper, we will follow these conventions. Input states will be represented as circles, while non-input states with be shown as rectangles. The name of each state is written above the line, while the scaling factor $d$ and mean $\mu$ of the distribution used to sample $\svar$ is written below the line. The parameters $d'$ and $\mu'$ for sampling $\svar'$ are not shown in the figures, but are mentioned in the caption and text when they are important; they are relevant only when $\svar'$ is output on a transition. Edges will be labeled with the guard of the transition, followed by the output, and a Boolean to indicate whether the transition is an assignment transition. 

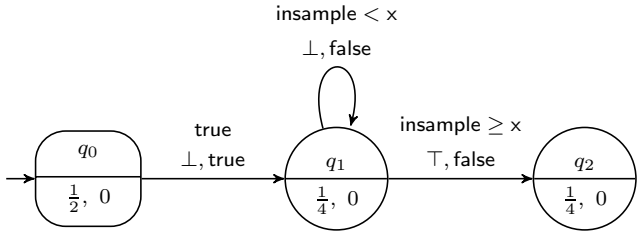
\begin{figure}
\begin{center}
\begin{tikzpicture}
\footnotesize
\node[mynonstate, initial] (q0) {\mystrut $q_0$ \nodepart{two} \mystrut $\frac{1}{2},\ 0$};
\node[myinstate, right of=q0] (q1) {$q_1$ \nodepart{lower} $\frac{1}{4},\ 0$};
\node[myinstate, right of=q1] (q2) {$q_2$ \nodepart{lower} $\frac{1}{4},\ 0$};
\draw (q0) edge node[myedgelabel] {$\true$ \nodepart{two} $\bot, \true$} (q1);
\draw (q1) edge[loop above] node[myedgelabel] {$\lttest$ \nodepart{two} $\bot, \false$} (q1)
                 edge node[myedgelabel] {$\getest$ \nodepart{two} $\top, \false$} (q2);
\end{tikzpicture}
\end{center}
\caption{{\dipa} $\svtauto$ modeling Algorithm~\ref{fig:SVT}. Threshold for the algorithm is $0$ (mean for sampling $\svar$ in state $q_0$).}
\label{fig:svt-auto}
\end{figure}

The SVT algorithm (Algorithm~\ref{fig:SVT}) can be modeled as a {\dipa} $\svtauto$ shown in Fig.~\ref{fig:svt-auto}. Since $\svtauto$ does not output $\svar'$ in any transition, the parameters used for sampling $\svar'$ are not relevant. In this representation of SVT, the threshold used for comparison in the algorithm is hard-coded in the automaton as the mean parameter of the initial state $q_0$. In fact, without loss of generality we can take this to be $0$ as shown in Fig.~\ref{fig:svt-auto}. The initial state $q_0$ of the automaton is a non-input state with $d = \frac{1}{2}$ and $\mu = 0$ (the threshold for the algorithm). From $q_0$, the algorithm samples a value that corresponds to the perturbed threshold and stores this in variable $\rvar$. In state $q_1$, in each step it reads a query value (input), perturbs it by sampling, and compares this with the perturbed threshold stored in variable $\rvar$. If the sampled value is less that $\rvar$ it stays in $q_1$, outputs $\bot$ and leaves $\rvar$ unchanged. On the other hand, if $\getest$ then it outputs $\top$, and transitions to a terminal state $q_2$.

$\svtauto$ can be used to illustrate our classification of transitions. The transition from $q_0$ to $q_1$ is the only non-input transition and the only assignment transition in the automaton; all other transitions are non-assignment, input transitions. In addition, the transition from $q_0$ to $q_1$ is also a pure assignment transition, since the guard is $\true$.

\begin{figure}
\begin{center}
\begin{tikzpicture}
\footnotesize
\node[mynonstate, initial] (q0) {\mystrut $q_0$ \nodepart{two} \mystrut $\frac{4}{9},\ 0$};
\node[myinstate, right of=q0] (q1) {$q_1$ \nodepart{lower} $\frac{2}{9},\ 0$};
\node[myinstate, right of=q1] (q2) {$q_2$ \nodepart{lower} $\frac{2}{9},\ 0$};
\draw (q0) edge node[myedgelabel] {$\true$ \nodepart{two} $\bot, \true$} (q1);
\draw (q1) edge[loop above] node[myedgelabel] {$\lttest$ \nodepart{two} $\bot, \false$} (q1)
                 edge node[myedgelabel] {$\getest$ \nodepart{two} $\svar', \false$} (q2);
\end{tikzpicture}
\end{center}
\caption{{\dipa} $\numspauto$ modeling Algorithm~\ref{fig:NumSp}. The threshold is taken to be $0$. Label of each state below the line shows the parameters for sampling $\svar$. Parameters for sampling $\svar'$ are not shown in the figure; they are $\frac{1}{9}$ (scaling factor) and $0$ (mean) in every state.}
\label{fig:numsp-auto}
\end{figure}
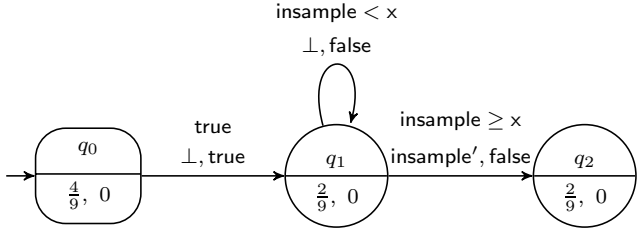

Automaton $\numspauto$ modeling Numeric Sparse (Algorithm~\ref{fig:NumSp}) is shown in Fig.~\ref{fig:numsp-auto}. As in the case of $\svtauto$ (Fig.~\ref{fig:svt-auto}), the threshold is hard-coded in the automaton and is taken to be $0$ (without loss of generality). Parameters used to sample $\svar'$ are not shown in diagram depicting $\numspauto$. We take those to be just be $\frac{1}{9}$ (scaling factor) and $0$ (mean) in every state; in fact, these parameters for $\svar'$ are only important for state $q_1$. The automaton is very similar to $\svtauto$ (Fig.~\ref{fig:svt-auto}) with the only differences being the parameters used when sampling in each state, and the fact that $\svar'$ is output on the transition from $q_1$ to $q_2$ instead of $\top$.
\end{example}

\subsection{Paths and executions}

A {\dipa} $\cA$ defines a probability measure on the \emph{executions} or \emph{paths} of $\cA$ (henceforth just called a path). Informally, a path is just a sequence of transitions taken by the automaton. Observe that the condition of {\outcond} ensures that knowing the current state and output, determines which transition is taken. The input read determines the value of $\svar$ and $\svar'$, and therefore, to define the probability of a path, we need to know the inputs read as well. Finally, on transitions where either $\svar$ or $\svar'$ are output, to define a meaningful measure space, we need to associate an interval $(v,w)$ in which the output value lies. Because of these reasons, we define a path to be one that describes the sequence of (control) states the automaton goes through and the sequence of inputs read and outputs produced. 

Before defining a path formally, it is useful to introduce the following notation. For a pair of states $p,q \in \states$, $a \in \inalph \cup \set{\emptystr}$ and $o \in \outalph \cup (\set{\svar,\svar'} \times \Reals_\infty \times \Reals_\infty)$, we say $p \trns{a,o} q$ if $a = \emptystr$ whenever $p \in \epsstates$ and $a \in \inalph$ whenever $p \in \instates$, and one of the following two conditions holds.
\begin{itemize}
\item If $o \in \outalph$ then there is a guard $c \in \cnds$ and Boolean $b \in \set{\true,\false}$ such that $\transf(p,c) = (q,o,b)$.
\item If $o$ is of the form $(y,v,w)$ where $y \in \set{\svar,\svar'}$ and $v,w \in \Reals_\infty$ then there is a guard $c \in \cnds$ and Boolean $b \in \set{\true,\false}$ such that $\transf(p,c) = (q,y,b)$. Intuitively, an ``output'' of the form $(\svar,v,w)$ (or $(\svar',v,w)$) indicates that the value of $\svar$ ($\svar'$) was output in the transition and the result was a number in the interval $(v,w)$.
\end{itemize}
The \emph{unique} transition, or rather the quintuple $(p,c,q,o',b)$, that witnesses $p \trns{a,o} q$ will be denoted by $\trname(p \trns{a,o} q)$.

\begin{definition}[Path]
\label{def:exec}
Let $\cA = \defaut$ be a {\dipa}. An \emph{execution} or \emph{path} $\rho$ of $\cA$ is a sequence of the form
\[
\rho = \defexec
\]
where $q_i \in \states$ for $0 \leq i \leq n$, $a_j \in \inalph \cup \set{\emptystr}$ and $o_j \in \outalph \cup (\set{\svar,\svar'} \times \Reals_\infty \times \Reals_\infty)$ for $0 \leq j < n$. In addition, we require that $q_j \trns{a_j,o_j} q_{j+1}$ for all $0 \leq j < n$.

Such a path $\rho$ is said to be from state $q_0$ ($\fstst(\rho)$) to state $q_n$ ($\lstst(\rho)$). Its \emph{length} (denoted $\len{\rho}$) is the number of transitions, namely, $n$. If the starting state and ending state of a path are the same (i.e., $q_0 = q_n$) and $\len{\rho} > 0$ then $\rho$ is said to be a \emph{cycle}.
\end{definition}

It will be convenient to introduce some notation associated with paths.
\begin{notation}
Let us consider a path
\[ \rho = \defexec \]
of length $n$. If $\len{\rho} > 0$, then the \emph{tail} of $\rho$, denoted $\tl(\rho)$, is the path of length $n-1$ given by
\[
\tl(\rho) = q_1 \trns{a_1,o_1} q_2 \cdots q_{n-1} \trns{a_{n-1},o_{n-1}} q_n.
\]
The $i$th state of the path is $\stname(\ith{\rho}) = q_i$ and the $i$th transition is $\trname(\ith{\rho}) = \trname(q_i \trns{a_i,o_i} q_{i+1})$. The guard of the $i$th transition is $\grdname(\ith{\rho}) = c$, where $\trname(\ith{\rho}) = (q_i,c,q_{i+1},o',b)$.

Finally, it will be useful to introduce notation for the sequence of inputs read and outputs produced in a path. The output produced will be an element of $(\outalph \cup (\Reals_\infty \times \Reals_\infty))^*$ that ignores the variable name that was output when a real value is output. For $o \in \outalph$, define $\tuple{o} = o$, and for $o$ of the form $(y,v,w)$ where $y \in \set{\svar,\svar'}$ and $v,w \in \Reals_\infty$, define $\tuple{o} = (v,w)$.
\[
\begin{array}{l}
\inseq(\rho) = a_0a_1\cdots a_{n-1}\\
\outseq(\rho) = \tuple{o_0}\tuple{o_1}\cdots \tuple{o_{n-1}}
\end{array}
\]
Two paths $\rho_1$ and $\rho_2$ will be said to be \emph{equivalent} if they only differ in the sequence of inputs read. In other words, equivalent paths are of the same length, go through the same states, and produce the same outputs (and hence take the same transitions).

\blue{Thanks to output distinction, two paths are equivalent if and only if they have the same output sequences. Thus,
paths are uniquely determined by input and output sequences. Finally, modifying the values input in a path yields an equivalent path. 
\begin{proposition}
\label{prop:factspath}
Let $\rho_1$ and $\rho_2$ be two two paths of a {\dipa} $\cA.$
\begin{itemize}
\item $\rho_1$ and $\rho_2$ are equivalent if and only if $\outseq(\rho_1)=\outseq(\rho_2).$
\item If $\inseq(\rho_1)=\inseq(\rho_2)$ and $\outseq(\rho_1)=\outseq(\rho_2)$ then $\rho_1=\rho_2.$
\item For any sequence of reals $\overline a \in \Sigma^*$ such that $\len {\overline{a}}=\len{\inseq (\rho_1)}$, there is a path $\rho_3$ equivalent to $\rho_1$ such that $\inseq(\rho_3)=\overline{a}.$
\end{itemize}
\end{proposition}}

\end{notation}

\subsection{Path probabilities}
We will now formally define what the probability of each path is. Recall that in each step, the automaton samples two values from Laplace distributions, and if the transition is from an input state, it adds the read input value to the sampled values and compares the result with the value stored in $\rvar$. The step also outputs a value, and if the value output is one of the two sampled values, the path requires it to belong to the interval that labels the transition. The probability of such a transition thus is the probability of drawing a sample that satisfies the guard of the transition and (if the output is a real value) producing a number that lies in the interval in the output label. This intuition is formalized in a precise definition. 

Let us fix a path
\[
\rho = \defexec
\]
of {\dipa} $\cA = \defaut$. Recall that the parameters to the Laplace distribution in each step depend on the privacy budget $\epsilon$. In addition, the value stored in the variable $\rvar$ at the start of $\rho$ influences the behavior of $\cA$. Thus, the probability of path $\rho$ depends on both the value for $\epsilon$ and the value of $\rvar$ at the start of $\rho$; we will denote this probability as $\pathprob{\epsilon,x,\rho}$, where $x$ is the initial value of $\rvar$. We define this inductively on $\len{\rho}$. For any $\epsilon$ and any path $\rho$ with $\len{\rho} = 0$, $\pathprob{\epsilon,x,\rho} = 1$. 

For a path $\rho$ of length $> 0$, let $(q_0,c,q_1,o_0,b) = \trname(q_0 \trns{a_0,o_0} q_1)$ be the $0$th transition of $\rho$. Let $\parf(q_0) = (d,\mu,d',\mu')$ and let $\tuple{a_0} = a_0$ if $a_0 \in \Reals$ and $\tuple{a_0} = 0$ if $a_0 = \emptystr$. We will define constants $\ell$ and $u$ as follows. If $o_0 \in \outalph$ then $\ell = -\infty$ and $u = \infty$. Otherwise, $o_0$ is of the form $(y,v,w)$ where $y \in \set{\svar,\svar'}$, and then we take $\ell = v$ and $u = w$. We assume that any integral of the form $\int_e^f g(y)dy = 0$ when $e > f$. Finally, when $o_0$ is of the form $(y,v,w)$ where $y \in \set{\svar,\svar'}$ (i.e., $o_0 \not\in \outalph$), define 
\[
\begin{array}{l}
k = \int_v^w \frac{d\epsilon}{2}\eulerv{-d\epsilon\card{z-\mu-\tuple{a_0}}}dz\\
k' = \int_v^w \frac{d'\epsilon}{2}\eulerv{-d'\epsilon\card{z-\mu'-\tuple{a_0}}}dz
\end{array}
\]

The function $\pathprob{\cdot}$ is defined based on what $c$ and $b$ are. Let us fix $\nu = \mu+\tuple{a_0}$. We begin by considering the case when the $0$th transition of $\rho$ is a non-assignment transition, i.e., when $b = \false$.
\begin{itemize}
\item {\bf Case $c = \true$:} If $o_0 \in \outalph$ then $\pathprob{\epsilon,x,\rho} = \pathprob{\epsilon,x,\tl(\rho)}$. If $o_0 = (\svar,v,w)$ then $\pathprob{\epsilon,x,\rho} = k\pathprob{\epsilon,x,\tl(\rho)}$ and if $o_0 = (\svar',v,w)$ then $\pathprob{\epsilon,x,\rho} = k'\pathprob{\epsilon,x,\tl(\rho)}$
\item {\bf Case $c = \getest$:} If $o_0$ is of the form $(\svar',v,w)$ (i.e., $\svar'$ is output) then 
\[
\pathprob{\epsilon,x,\rho} = k'\left(\int_x^\infty \frac{d\epsilon}{2}\eulerv{-d\epsilon\card{z-\nu}}dz \right)\pathprob{\epsilon,x,\tl(\rho)}.
\]
Otherwise, taking $\ell' = \max(x,\ell)$,
\[
\pathprob{\epsilon,x,\rho} = \left(\int_{\ell'}^u \frac{d\epsilon}{2}\eulerv{-d\epsilon\card{z-\nu}}dz \right)\pathprob{\epsilon,x,\tl(\rho)}.
\]
\item {\bf Case $c = \lttest$:} If $o_0$ is of the form $(\svar',v,w)$ (i.e., $\svar'$ is output) then
\[
\pathprob{\epsilon,x,\rho} = k'\left(\int_{-\infty}^x \frac{d\epsilon}{2}\eulerv{-d\epsilon\card{z-\nu}}dz \right)\pathprob{\epsilon,x,\tl(\rho)}.
\]
Otherwise, taking $u' = \min(x,u)$,
\[
\pathprob{\epsilon,x,\rho} = \left(\int_{\ell}^{u'} \frac{d\epsilon}{2}\eulerv{-d\epsilon\card{z-\nu}}dz \right)\pathprob{\epsilon,x,\tl(\rho)}.
\]
\end{itemize}
Next, when the $0$th transition of $\rho$ is an assignment transition, i.e., $b = \true$, $\pathprob{\cdot}$ is defined as follows.
\begin{itemize}
\item {\bf Case $c = \true$:} If $o_0$ is of the form $(\svar',v,w)$ (i.e., $\svar'$ is output) then 
\[
\pathprob{\epsilon,x,\rho} = k'\int_{-\infty}^\infty \left(\frac{d\epsilon}{2}\eulerv{-d\epsilon\card{z-\nu}}\right) \pathprob{\epsilon,z,\tl(\rho)} dz.
\]
Otherwise,
\[
\pathprob{\epsilon,x,\rho} = \int_{\ell}^u \left(\frac{d\epsilon}{2}\eulerv{-d\epsilon\card{z-\nu}}\right) \pathprob{\epsilon,z,\tl(\rho)} dz.
\]
\item {\bf Case $c = \getest$:} If $o_0$ is of the form $(\svar',v,w)$ (i.e., $\svar'$ is output) then 
\[
\pathprob{\epsilon,x,\rho} = k'\int_{x}^\infty \left(\frac{d\epsilon}{2}\eulerv{-d\epsilon\card{z-\nu}}\right) \pathprob{\epsilon,z,\tl(\rho)} dz.
\]
Otherwise, taking $\ell' = \max(x,\ell)$,
\[
\pathprob{\epsilon,x,\rho} = \int_{\ell'}^u \left(\frac{d\epsilon}{2}\eulerv{-d\epsilon\card{z-\nu}}\right) \pathprob{\epsilon,z,\tl(\rho)} dz.
\]
\item {\bf Case $c = \lttest$:} If $o_0$ is of the form $(\svar',v,w)$ (i.e., $\svar'$ is output) then 
\[
\pathprob{\epsilon,x,\rho} = k'\int_{-\infty}^x \left(\frac{d\epsilon}{2}\eulerv{-d\epsilon\card{z-\nu}}\right) \pathprob{\epsilon,z,\tl(\rho)} dz.
\]
Otherwise, taking $u' = \min(u,x)$,
\[
\pathprob{\epsilon,x,\rho} = \int_{\ell}^{u'} \left(\frac{d\epsilon}{2}\eulerv{-d\epsilon\card{z-\nu}}\right) \pathprob{\epsilon,z,\tl(\rho)} dz.
\]
\end{itemize}
We will abuse notation and use $\pathprob{\cdot}$ to also refer to $\pathprob{x,\rho} = \lambda \epsilon.\ \pathprob{\epsilon,x,\rho}$. Notice that when $\rho$ starts from $\qinit$, because of the {\initcond} condition of {\dipa}, the value of $\pathprob{\cdot}$ does not depend on the initial value of $\rvar$. For such paths, we may drop the initial value of $\rvar$ from the argument list of $\pathprob{\cdot}$ to reduce notational overhead. Even though we plan to use the same function name, the number of arguments to $\pathprob{\cdot}$ will disambiguate what we mean.

\begin{example}
\label{ex:execsem}
Let use consider the {\dipa} $\svtauto$ shown in Fig.~\ref{fig:svt-auto}. A couple of example paths of the automaton are the following.
\[
\begin{array}{l}
\rho_1 = q_0 \trns{\emptystr,\bot} q_1 \trns{0,\bot} q_1 \trns{1,\top} q_2\\
\rho_2 = q_0 \trns{\emptystr,\bot} q_1 \trns{1,\bot} q_1 \trns{1,\top} q_2
\end{array}
\]
Paths $\rho_1$ and $\rho_2$ only differ in the inputs they read: $\inseq(\rho_1) = \tau\cdot 0\cdot 1 = 01$, while $\inseq(\rho_2) = 11$. Thus, $\rho_1$ and $\rho_2$ are equivalent paths. Notice that $\rho_1$ and $\rho_2$ are adjacent (Definition~\ref{def:adjacency}). The outputs produced in these executions is given by $\outseq(\rho_1) = \outseq(\rho_2) = \bot\bot\top$.

Let us now consider $\pathprob{\epsilon,0,\rho_1}$. Since the transition out of $q_0$ is a pure assignment transition, the initial value of $\rvar$ (namely $0$ in this example) does not influence the value of $\pathprob{\epsilon,0,\rho_1}$. Let $X_T,X_1,X_2$ be random variables where $X_T \sim \Lap{\frac{\epsilon}{2},0}$, $X_1 \sim \Lap{\frac{\epsilon}{4},0} + 0$, and $X_2 \sim \Lap{\frac{\epsilon}{4},0}+1$. We can see that 
\[
\pathprob{\epsilon,0,\rho_1} = \prbfn{X_1 < X_T\ \wedge\ X_2 \geq X_T}.
\]
Based on how the random variables are distributed, this can be calculated to be  
\[
\prbfn{X_1 < X_T\: \wedge\: X_2 \geq X_T} = \frac{24 \euler^{\frac{3\epsilon}{4}} - 1 +
 8 \eulerv{\frac{\epsilon}{4}} - 21 \eulerv{\frac{\epsilon}{2}}
 }{48 \euler^{\frac{3\epsilon}{4}}}.
\]

The calculation of $\pathprob{\epsilon,0,\rho_2}$ is similar. Let $X_1'$ be the random variable with $X_1' \sim \Lap{\frac{\epsilon}{4},0}+1$. Then the desired probability is same as $\prbfn{X_1' < X_T\ \wedge X_2 \geq X_T}$. This can be calculated to be
\[
\begin{array}{rl}
\pathprob{\epsilon,0,\rho_2} & = \prbfn{X_1' < X_T\: \wedge\: X_2 \geq X_T}\\
& = \frac{-22 + 32 \eulerv{\frac{\epsilon}{4}} -3 \epsilon
}{48 \euler^{\frac{\epsilon}{2}}}.
\end{array}
\]
\end{example}

The focus of this paper is to study the computational problem of checking differential privacy for {\dipautop}. We conclude this section with a precise definition of this problem. In order to do that we first specialize the definition of differential privacy to the setting of {\dipa}. \blue{Recall that two paths are equivalent if and only if they have the same output sequences, and a path is uniquely determined by its input and output sequences (See Proposition~\ref{prop:factspath}).
\begin{definition}
\label{def:diff-priv-auto}
A {\dipa} $\cA$ is said to be $d\epsilon$-differentially private (for $d > 0$, $\epsilon > 0$) if for every pair of \emph{equivalent} paths 
$\rho_1, \rho_2$ such that $\inseq(\rho_1)$ and $\inseq(\rho_2)$ are \emph{adjacent}~\footnote{See Definition~\ref{def:adjacency} on Page~\pageref{def:adjacency}},
\[
\pathprob{\epsilon,\rho_1} \leq e^{d\epsilon} \;
\pathprob{\epsilon,\rho_2}.
\]
\end{definition}}
%
% \begin{definition}
% \label{def:diff-priv-auto}
% A {\dipa} $\cA$ is said to be $d\epsilon$-differentially private (for $d > 0$, $\epsilon > 0$) if for every output sequence $\overline{o} \in (\outalph \cup (\Reals_\infty \times \Reals_\infty))^*$ and pair of \emph{adjacent}~\footnote{See Definition~\ref{def:adjacency} on Page~\pageref{def:adjacency}.} input sequences $\overline{i_1}, \overline{i_2} \in \Reals^*$, 
% \[
% \sum_{\substack{\rho:\ \inseq(\rho) = \overline{i_1} \\ \hspace*{0.05in}\outseq(\rho) = \overline{o}}} 
% \pathprob{\epsilon,\rho} \leq e^{d\epsilon} 
% \sum_{\substack{\rho:\ \inseq(\rho) = \overline{i_2} \\ \hspace*{0.05in} \outseq(\rho) = \overline{o}}}
% \pathprob{\epsilon,\rho}.
% \]
% \end{definition}

\vspace*{0.05in}
\noindent
{\bf Differential Privacy Problem:} Given a {\dipa} $\cA$ (with privacy parameter $\epsilon$), determine if there is a $d > 0$ such that for every $\epsilon > 0$, $\cA$ is $d\epsilon$-differentially private.

%\mahesh{observe that our results will show that the problem is equivalent to one where the parameter function is also a parameter. }

% !TEX root =  main.tex

\section{Deciding Differential Privacy}
\label{sec:decidability}

The central computational problem that this paper studies is the following: Given a {\dipa} ${\cA}$ determine if there is a $d > 0$ such that for all $\epsilon > 0$, $\cA$ is $d\epsilon$-differentially private. In this section we present the main result of this paper, namely, that this problem is efficiently decidable in {linear} time. We also show that we can compute an upper bound on $d$ in linear time if $\cA$ is differentially private.  The crux of the proof is the identification of simple graph-theoretic conditions that are both \emph{necessary and sufficient} to ensure a {\dipa} is $d\epsilon$-differentially private for all $\epsilon$ and some $d$.

Before presenting the properties that are needed to guarantee differential privacy, we first define the notion of reachability. Let us fix a {\dipa} $\cA = \defaut$. A state $q$ is said to be \emph{reachable} if there is a path $\rho$ starting from state $\qinit$ and ending in $q$. In addition, we say that a path (cycle) $\rho$ is reachable if there is a path $\rho'$ from $\qinit$ to $\fstst(\rho)$. We now start by identifying the first interesting property.
\begin{definition}
\label{def:leaky-paths}
A path $\rho$ in a {\dipa} $\cA$ is said to be a \emph{\criticalpath} if there exist indices $i,j$ with $0 \leq i < j < \len{\rho}$  such that the $i$th transition $\trname(\ith{\rho})$ is an assignment transition and the guard of the $j$th transition $\grdname(\ith[j]{\rho}) \neq \true$.
A {\criticalpath} $\rho$ is said to be a \emph{\criticalcycle} if it is also a cycle.
\end{definition}

Intuitively, in a {\criticalpath}, the variable $\rvar$ is assigned a value in some transition which  is used in the guard of a later transition. Observe that if a path is {\critical} then all paths equivalent to it are also {\critical}. The presence of a reachable {\criticalcycle} is a witness that the {\dipa} is not differentially private. The intuition behind this is as follows. One can show that there are a pair of adjacent inputs such that traversing {\criticalcycle} $C$ on these inputs results in two paths the ratio of whose probability is at least $\eulerv{k\epsilon}$ for some number $k$. Thus, given $d$, we can find an $\ell$ and $\epsilon$ such that traversing the cycle $\ell$ times ``exhausts the privacy budget'', i.e., the adjacent input corresponding to these $\ell$ repetitions have probabilities that are more than $\eulerv{d\epsilon}$ apart. We illustrate this through our next example.

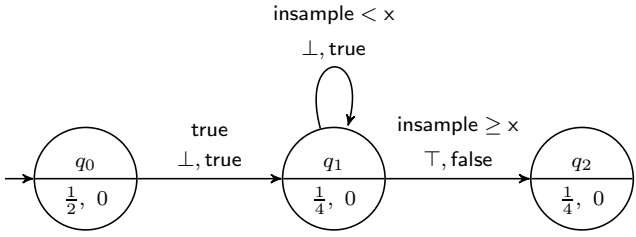
\begin{figure}
\begin{center}
\begin{tikzpicture}
\footnotesize
\node[myinstate, initial] (q0) {$q_0$ \nodepart{lower} $\frac{1}{2},\ 0$};
\node[myinstate, right of=q0] (q1) {$q_1$ \nodepart{lower} $\frac{1}{4},\ 0$};
\node[myinstate, right of=q1] (q2) {$q_2$ \nodepart{lower} $\frac{1}{4},\ 0$};
\draw (q0) edge node[myedgelabel] {$\true$ \nodepart{two} $\bot, \true$} (q1);
\draw (q1) edge[loop above] node[myedgelabel] {$\lttest$ \nodepart{two} $\bot, \true$} (q1)
                 edge node[myedgelabel] {$\getest$ \nodepart{two} $\top, \false$} (q2);
\end{tikzpicture}
\end{center}
\caption{{\dipa} $\sortauto$ modeling an algorithm that checks whether the sequence of real numbers given as input are sorted in descending order. Since $\svar'$ is not output in any state, the parameters used in sampling $\svar'$ are not important.}
\label{fig:sort-auto}
\end{figure}
\begin{example}
\label{ex:leakcycle}
Consider an algorithm that checks whether the input sequence of real numbers is sorted in descending order. The goal of the algorithm is to read a sequence of numbers, output $\bot$ as long as it is sorted, and output $\top$ the first time it encounters two numbers in the wrong order and stop. A ``differentially private'' version of this algorithm is modeled by {\dipa} $\sortauto$ shown in Fig.~\ref{fig:sort-auto}. It works as follows. It starts by reading an input in state $q_0$, perturbing it by sampling from the Laplace distribution, outputting $\bot$, and storing the perturbed input in $\rvar$. In state $q_1$, $\sortauto$ repeatedly reads an input, perturbs it, and checks if it is less than the previous perturbed value read by the automaton, which is stored in $\rvar$. If it is, the automaton outputs $\bot$, saves the new perturbed value, and stays in $q_1$ to read the next input symbol. On the other hand, if the new value is greater, then it outputs $\top$ and moves to a terminal state. $\sortauto$ is almost identical to the automaton $\svtauto$ (Fig.~\ref{fig:svt-auto}) --- the only difference is that initial state of $\sortauto$ is an input state as opposed to a non-input state, and the self loop on state $q_1$ is an assignment transition. 

This difference (that the self loop on $q_1$ is an assignment transition) turns out to be critical; $\sortauto$ is not differentially private even though $\svtauto$ is. Observe that the cycle $q_1 \trns{a_0,\bot} q_1 \trns{a_1,\bot} q_1$ is a {\criticalcycle} as the $0$th transition is an assignment transition and the $1$st transition's guard is $\lttest$. We can exploit this cycle to demonstrate why $\sortauto$ is not differentially private. Consider the paths of length $n$ given as
\[
\begin{array}{l}
\rho_1^n = q_0 \trns{0,\bot} q_1 \trns{-1,\bot} q_1 \trns{-2,\bot} q_1 \trns{-3,\bot} q_1 \trns{-4,\bot} q_1 \cdots\\
\rho_2^n = q_0 \trns{0,\bot} q_1 \trns{-2,\bot} q_1 \trns{-1,\bot} q_1 \trns{-4,\bot} q_1 \trns{-3,\bot} q_1 \cdots
\end{array}
\]
Observe that for all $n$, $\inseq(\rho_1^n)$ and $\inseq(\rho_2^n)$ are adjacent (Definition~\ref{def:adjacency}). Moreover, for any $d > 0$, there is an $n$ and $\epsilon$, such that the ratio of $\pathprob{\epsilon,\rho_1^n}$ and $\pathprob{\epsilon,\rho_2^n}$ is $> \eulerv{d\epsilon}$. Thus, $\cA$ is not $d\epsilon$-differentially private for any $d$.
\end{example}

Absence of a {\criticalcycle} does not guarantee differential privacy. \blue{Privacy leaks can occur with other types of paths and cycles. We define one such path next.}
%
%This is because the effect of a {\criticalcycle} can also be achieved through a pair of cycles of a special kind. We define this notion next.

\begin{definition}
\label{def:lg-cycle}
A cycle $\rho$ of a {\dipa} $\cA$ is called an {\lcycle} (respectively,  {\gcycle}) if there is an $i < \len{\rho}$ such that $\grdname(\ith{\rho}) = \lttest$ (respectively, $\grdname(\ith{\rho}) = \getest$).

We say that a path $\rho$ of a {\dipa} $\cA$ is an {\alpath} (respectively, {\agpath}) if all assignment transitions on $\rho$ have guard $ \lttest$ (respectively, $\getest$).
\end{definition}

Observe that a cycle can be both an {\lcycle} and a {\gcycle}. Further, a path with no assignment transitions (including the empty path) is simultaneously both an {\alpath} and an {\agpath}.

\begin{definition}
\label{def:leakingpair}
A pair of cycles $(C,C')$ in a {\dipa} $\cA$ is called a \emph{\criticalpair} if one of the following two conditions is satisfied.
\begin{enumerate}
\item  $C$ is an {\lcycle}, $C'$ is a {\gcycle} and there is an {\agpath} from a state in $C$ to a state in $C'.$ 

\item  $C$ is a {\gcycle}, $C'$ is an {\lcycle} and there is an {\alpath} from a state in $C$ to a state in $C'.$ 
\end{enumerate}
\end{definition}

Observe that if $C$ is an {\lcycle} as well as a {\gcycle}, then the pair $(C,C)$ is a {\criticalpair} with the empty path connecting $C$ to itself. Also, if $(C,C')$ is a {\criticalpair}, then for any $C_1,C_2$ that are equivalent to $C,C'$ respectively, the pair $(C_1,C_2)$ is also a {\criticalpair}. 

The presence of a {\criticalpair} is also a witness to a {\dipa} not being differentially private. Consider a {\dipa} $\cA$ that has no {\criticalcycle} but has a {\criticalpair} of cycles $(C,C')$ such that $C$ is reachable. Assume that $C'$ is a {\gcycle}. The case when $C'$ is an {\lcycle} is symmetric. Since $\cA$ has no {\criticalcycle}s, the value stored in $\rvar$ does not change while the automaton is executing the transitions in either $C$ or $C'$. Let $y$ be the value of $\rvar$ when $C'$ starts executing. One can show that if $y > 0$ then there are a pair of adjacent inputs such that traversing $C'$ on those inputs results in paths whose probabilities have ratios that are at least $\eulerv{k\epsilon}$ for some $k$. Moreover, this pair of inputs does not depend on the actual value of $y$. This once again means that by repeating $C'$ $\ell$ times, we can get adjacent inputs whose probabilities violate the $d\epsilon$ privacy budget (for any $d$). A similar observation holds for {\lcycle} $C$ --- if the value of $\rvar$ at the start of $C$ is $\leq 0$ then we can find adjacent inputs such that traversing $C$ for those inputs results in paths whose probabilities have a ``high'' ratio. The next observation is that value stored in $\rvar$ at the end of an {\agpath} is at least the value at the beginning of the path. We can now put all these pieces together to get our witness for a violation of differential privacy. If the value of $\rvar$ is $\leq 0$ at the start of $C$, then repeating $C$ $\ell$ times gives us a pair of adjacent inputs that violate the privacy budget. On the other hand, if $\rvar$ at the start of $C$ is $> 0$ then it will be $> 0$ even at the start of $C'$, and then repeating $C'$ $\ell$ times gives us the desired witnessing pair. Let us illustrate this through an example.

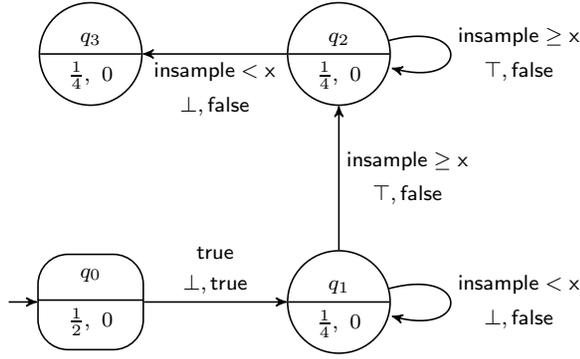
\begin{figure}
\begin{center}
\begin{tikzpicture}
\footnotesize
\node[mynonstate, initial] (q0) {\mystrut $q_0$ \nodepart{two} \mystrut $\frac{1}{2},\ 0$};
\node[myinstate, right of=q0] (q1) {$q_1$ \nodepart{lower} $\frac{1}{4},\ 0$};
\node[myinstate, above of=q1] (q2) {$q_2$ \nodepart{lower} $\frac{1}{4},\ 0$};
\node[myinstate, left of=q2] (q3) {$q_3$ \nodepart{lower} $\frac{1}{4},\ 0$};
\draw (q0) edge node[myedgelabel] {$\true$ \nodepart{two} $\bot, \true$} (q1);
\draw (q1) edge[loop right] node[myedgelabel] {$\lttest$ \nodepart{two} $\bot, \false$} (q1)
                 edge node[myedgelabel, right] {$\getest$ \nodepart{two} $\top, \false$} (q2);
\draw (q2) edge[loop right] node[myedgelabel] {$\getest$ \nodepart{two} $\top, \false$} (q2)
                 edge node[myedgelabel] {$\lttest$ \nodepart{two} $\bot,\false$} (q3);
\end{tikzpicture}
\end{center}
\caption{{\dipa} $\svtpauto$ modeling an algorithm that processes a sequence of real numbers and implements a ``noisy' version'' of the following process. As long as the input numbers are less than threshold $T$ ($=0$) it outputs $\bot$. Once it sees the first number $\geq T$, it moves to the second phase. In the phase, it outputs $\top$ as long as the numbers are $\geq T$. When it sees the first number $< T$, it outputs $\bot$ and stops. Since $\svar'$ is never output, parameters used in its sampling are not shown and not important.}
\label{fig:svtp-auto}
\end{figure}
\begin{example}
\label{ex:leakingpair}
Consider the automaton $\svtpauto$ shown in Fig.~\ref{fig:svtp-auto}. It implements an algorithm that is a slight modification of Algorithm~\ref{fig:SVT} (or the {\dipa} $\svtauto$ in Fig.~\ref{fig:svt-auto}). Like in SVT, the automaton starts in state $q_0$ by sampling a value that is a perturbed value of a threshold $T$ (which is $0$ here). It stores this sampled value in $\rvar$ and moves to the first phase (state $q_1$). In this phase, the automaton outputs $\bot$ and stays in $q_1$ as long as a perturbed value of the input read is less than the perturbed threshold stored in $\rvar$. The first time it encounters a perturbed value that is at least $\rvar$, it moves to phase two (state $q_2$) and outputs $\top$. In state $q_2$, it outputs $\top$ as long as the perturbed inputs it samples are $\geq \rvar$. The first time it encounters a value $< \rvar$ it outputs $\bot$ and terminates. Throughout the computation, the automaton never over-writes the value stored in the first step in variable $\rvar$. 

$\svtpauto$ has a {\criticalpair}. Observe that $C=q_1\trns{a_1,\bot}q_1$ is an {\lcycle} and $C'=q_2\trns{a_2,\top}q_2$ is a {\gcycle}. The path $q_1\trns{a_3,\top}q_2$ is an {\agpath} from $C$ to $C'.$ Hence $(C,C')$ is a {\criticalpair}. The presence of this {\criticalpair} can be exploited to show that $\svtpauto$ is not $d\epsilon$-differentially private for any $d > 0$.

Consider the following two paths.
\[
\begin{array}{l}
\rho_1^\ell = q_0 \trns{\emptystr,\bot} \left[ q_1 \trns{-\frac{1}{2},\bot} q_1 \right]^\ell \trns{0,\top} \left[ q_2 \trns{\frac{1}{2},\top} q_2 \right]^\ell \trns{0,\bot} q_3\\
\rho_2^\ell = q_0 \trns{\emptystr,\bot} \left[ q_1 \trns{\frac{1}{2},\bot} q_1 \right]^\ell \trns{0,\top} \left[ q_2 \trns{-\frac{1}{2},\top} q_2 \right]^\ell \trns{0,\bot} q_3
\end{array}
\]
In the above $[ p \trns{a,o} q ]^\ell$ means that the path consists of repeating this transition $\ell$ times. Notice that the $\inseq(\rho_1^\ell) = (-\frac{1}{2})^\ell 0 (\frac{1}{2})^\ell 0$ and $\inseq(\rho_2^\ell) = (\frac{1}{2})^\ell 0 (-\frac{1}{2})^\ell 0$ are adjacent. Moreover, for any $d > 0$, there is a $\ell$ such that for every $\epsilon$ the ratio of $\pathprob{\epsilon,\rho_1^\ell}$ and $\pathprob{\epsilon,\rho_2^\epsilon}$ is $> \eulerv{d\epsilon}$. Thus, for an appropriately chosen value for $\ell$, $\rho_1^\ell$ and $\rho_2^\ell$ witness the violation of differential privacy. 
\end{example}

The two conditions we have identified thus far --- existence of reachable {\criticalcycle} or {\criticalpair} --- demonstrate differential privacy violations even in {\dipa}s that do not output any real value. In automata that output real values, there are additional sources of privacy violations. We identify these conditions next.

\begin{definition}
\label{def:disclosingcycle}
A cycle $C$ of a {\dipa} $\cA$ is a \emph{\violatingc} if there is an $i$, $0 \leq i < \len{C}$ such that $\trname(\ith{C})$ is an input transition that outputs either $\svar$ or $\svar'$.
\end{definition}

Again the existence of a reachable  {\violatingc} demonstrates that the {\dipa} is not differentially private --- outputting a perturbed input repeatedly exhausts the privacy budget. %\red{Can we say something more?}

We now present the last property of importance that pertains to paths that have transitions that output the value of $\svar$. We say that a state $q$ \emph{is in a cycle} ({\gcycle} or {\lcycle}) if there is a cycle ({\gcycle}/{\lcycle}) $C$ and index $i$ such that $q = \stname(\ith{C})$.

\begin{definition}
\label{def:violating}
We say that a path $\rho = \defexec$ of length $n$ of {\dipa} $\cA$ is a \emph{\violatingp} if one of the following conditions hold.
\begin{itemize}
\item $\tl(\rho)$ is an {\agpath} (resp., {\alpath}) such that $\lstst(\rho)$ is in a {\gcycle} (resp., {\lcycle}) and the $0$th transition $\trname(\ith[0]{\rho})$ is an assignment transition that outputs $\svar$.
\item $\rho$ is an {\agpath} (resp., {\alpath}) such that $\lstst(\rho)$ is in a {\gcycle} (resp., {\lcycle}) and the $0$th transition has $\grdname(\ith[0]{\rho}) = \lttest$ (resp., $\grdname(\ith[0]{\rho}) = \getest$) and outputs $\svar$.
\item $\rho$ is an {\agpath} (resp., {\alpath}) such that $\fstst(\rho)$ is in an {\lcycle} (resp., {\gcycle}) and the last transition has guard $\grdname(\ith[n-1]{\rho}) = \getest$ (resp., $\grdname(\ith[n-1]{\rho}) = \lttest$) and outputs $\svar$.
\end{itemize}
\end{definition}

Once again, the presence of a reachable {\violatingp} demonstrates that the automaton is not differentially private. Let us provide some intuition why that is the case. We do this for some of the cases that form a {\violatingp} with reasoning for the missing cases being similar. As before, let us assume that there is no {\criticalcycle} because if there is one then we already know that the automaton is not differential privacy. A consequence of this that there are no assignment transitions in a {\gcycle} or {\lcycle} and hence the value stored in $\rvar$ remains unchanged in these cycles. Let us recall a couple of crucial observation that we used when we argued in the case of a {\criticalpair}. First, the value stored in $\rvar$ at the end of an {\agpath} is at least as large as the value at the beginning. Next, if a {\gcycle} ({\lcycle}) is traversed when the starting value in $\rvar$ is $>0$ ($\leq 0$) then we have a family of pairs of adjacent inputs that correspond to traversing the cycle multiple times with the property that the ratio of their probabilities diverges as the cycle is traversed more times. Let us now consider each of the cases in the definition of {\violatingp}. If $\rho$ starts with an assignment transition that outputs $\svar$ and if the output of this first step is in the interval $(0,\infty)$ then the value of $\rvar$ is $>0$ at the end of $\rho$ when a {\gcycle} can be traversed. These observations can be used to give us a pair of adjacent inputs that violate privacy. If $\rho$ starts with a transition whose guard is $\lttest$ that outputs $\svar$ and suppose the value output in this step is in the interval $(0,\infty)$ then the value in $\rvar$ at the start is $> 0$. Like in the previous case this can be used to get a violating pair of inputs. Finally, if $\rho$ ends in transition outputting $\svar$, guard $\getest$ and the value output in this last step in the interval $(-\infty,0)$, then we can conclude that the value in $\rvar$ at the end of $\rho$ is $\leq 0$. This combined with properties of {\agpath}s means that $\rvar$ has a value $\leq 0$ at the beginning of $\rho$. This means the {\lcycle} at the start of $\rho$ can be traversed with $\rvar$ having a value $\leq 0$ which means that a violating pair of inputs can be constructed.

Let us illustrate this last condition through another example.
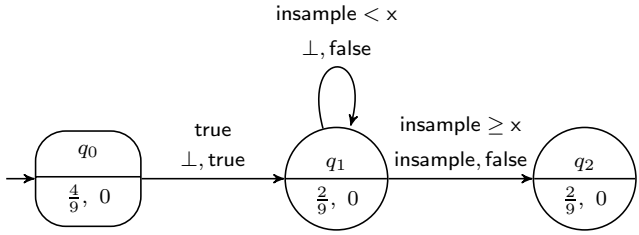
\begin{figure}
\begin{center}
\begin{tikzpicture}
\footnotesize
\node[mynonstate, initial] (q0) {\mystrut $q_0$ \nodepart{two} \mystrut $\frac{4}{9},\ 0$};
\node[myinstate, right of=q0] (q1) {$q_1$ \nodepart{lower} $\frac{2}{9},\ 0$};
\node[myinstate, right of=q1] (q2) {$q_2$ \nodepart{lower} $\frac{2}{9},\ 0$};
\draw (q0) edge node[myedgelabel] {$\true$ \nodepart{two} $\bot, \true$} (q1);
\draw (q1) edge[loop above] node[myedgelabel] {$\lttest$ \nodepart{two} $\bot, \false$} (q1)
                 edge node[myedgelabel] {$\getest$ \nodepart{two} $\svar, \false$} (q2);
\end{tikzpicture}
\end{center}
\caption{{\dipa} $\cA_{\s{mod}}$ is a modification of $\numspauto$. Label of each state below the line shows the parameters for sampling $\svar$. Parameters for sampling $\svar'$ are not shown in the figure; they are $\frac{1}{9}$ (scaling factor) and $0$ (mean) in every state.}
\label{fig:mod-auto}
\end{figure}
\begin{example}
\label{ex:violating}
Consider automaton $\cA_{\s{mod}}$ (Fig.~\ref{fig:mod-auto}) which is a modification of the Numeric Sparse algorithm modeled by automaton $\numspauto$ (Fig.~\ref{fig:numsp-auto}). The only difference is that the transition from $q_1$ to $q_2$ outputs $\svar$ as opposed to $\svar'$. This change causes this automaton to be not differentially private.

Observe that the state $q_1$ is in a $\lcycle$ $q_1 \trns{a,\bot} q_1$ and then path $\rho = q_1 \trns{a,(\svar, (0,\infty))} q_2$ is an {\agpath}. Finally, the last transition (or rather the only transition) of $\rho$ has guard $\getest$ that outputs $\svar$. Thus, $\rho$ is a {\violatingp}.

We can use $\rho$ to find a violation for privacy. Consider the following pair of paths.
\[
\begin{array}{l}
\rho_1^\ell = q_0 \trns{\emptystr,\bot} \left[ q_1 \trns{-\frac{1}{2},\bot} q_1 \right]^\ell \trns{0,(\svar,(0,\infty))} q_2\\
\rho_2^\ell = q_0 \trns{\emptystr,\bot} \left[ q_1 \trns{\frac{1}{2},\bot} q_1 \right]^\ell \trns{0,(\svar,(0,\infty))} q_2
\end{array}
\]
Observe that $\inseq(\rho_1^\ell) = (-\frac{1}{2})^\ell 0$ and $\inseq(\rho_2^\ell) = (\frac{1}{2})^\ell 0$ are adjacent. Moreover, for any $d > 0$, there is an $\ell$ such that for any $\epsilon$, the ratio of $\pathprob{\epsilon,\rho_1^\ell}$ and $\pathprob{\epsilon,\rho_2^\ell}$ is $>  \eulerv{d\epsilon}$. Thus, $\rho_1^\ell$ and $\rho_2^\ell$ demonstrate the violation of privacy.
\end{example}

As the discussion and examples above illustrate, absence of {\criticalcycle}s, {\criticalpair}s, {\violatingc}s, and {\violatingp}s is necessary for a {\dipa} to be differentially private. We call such automata \emph{well-formed}.

\begin{definition}
\label{def:well-formed}
A {\dipa} $\cA$ is said to be {\em well-formed} if $\cA$ has no reachable {\criticalcycle}, no {\criticalpair} $(C,C')$ where $C$ is reachable, no reachable {\violatingc}, and no reachable {\violatingp}. 
\end{definition}

Our main theorem is that well-formed {\dipa}s are exactly the class of automata that are differentially private. 
\ifdefined\AppendixTrue
The proof of this Theorem is carried out in the Appendix (See Appendix~\ref{app:necessity} for the \lq\lq only if\rq\rq\ direction and Appendix~\ref{app:sufficient} for the \lq\lq if\rq\rq\ direction).
\fi
\begin{theorem}
\label{thm:main}
Let $\cA$ be a {\dipa}. There is a $d > 0$ such that for every $\epsilon > 0$, $\cA$ is $d\epsilon$-differentially private if and only if $\cA$ is well-formed.
\end{theorem}

\begin{remark}
Before presenting a proof sketch for Theorem~\ref{thm:main}, it is useful to point out one special case for the result. Observe that {\violatingc}s and {\violatingp}s pertain to paths that have transitions that output real values. For {\dipa}s that do not have real outputs, {\violatingc}s and {\violatingp}s are not needed to get an exact characterization of differential privacy. More precisely, we say that a {\dipa} $\cA = \defaut$ has \emph{finite valued outputs} if every transition in $\cA$ outputs a value in $\outalph$. Now, a {\dipa} with finite valued outputs is differentially private if and only if it has no reachable {\criticalcycle}s and {\criticalpair}s. 
\end{remark}

Discussion in this section has provided intuitions for why well-formed-ness is necessary for an automaton to be differentially private; the formal proof that captures these intuitions is subtle, long, and non-trivial.
\ifdefined\AppendixTrue
The proof  is postponed to Appendix~\ref{app:necessity}.
\fi
We sketch some key properties that show why it is sufficient.

Let us fix a transition $t = (p,c,q,o,b)$ in a {\dipa} $\cA = \defaut$. The transition $t$ is said to lie on a cycle if there is a reachable cycle $\rho$ and index $i$ such that $\trname(\ith{\rho}) = t$. On the other hand, we will say $t$ is a \emph{\criticaltransition} if $t$ does not lie on a cycle. Let $\parf(p) = (d,\mu,d',\mu')$ be the parameters for sampling $\svar$ and $\svar'$ in state $p$. We define the cost of $t$ as follows.
\[
\cost{t}= \begin{cases}
                          d & t \mbox{ is a {\criticalnotransition}}\\
                          2 d & t \mbox{ is a {\criticalintransition} and }\\
                           & o \neq \svar'\\
                          2d+d' & t \mbox{ is a {\criticalintransition} and }\\
                           & o = \svar'\\
                          0 & \mbox{otherwise}                                              
\end{cases}.
\]
For a path $\rho$, define weight of $\rho$ as $\weight{\rho} = \sum_{i=0}^{\len{\rho}-1} \cost{\trname(\ith{\rho})}$, i.e., the sum of the costs of all the transitions in $\rho$. Finally, define $\weight{\cA}$ to be the supremum over all paths $\rho$, $\weight{\rho}$. In fact, the weight of $\cA$ could have been defined as a maximum (as opposed to a supremum) because they are the same in this case. The crucial observation about weight of an automaton that is used in proving the sufficiency of well-formed-ness for differential privacy, is that it provides an upper bound on the privacy budget for $\cA$.
\begin{lemma}
\label{lem:sufficiency}
A well-formed {\dipa} $\cA$ is $\weight{\cA}\epsilon$-differentially private for all $\epsilon > 0$.
\end{lemma}
\begin{proof}(Sketch.)
\blue{\ifdefined\AppendixTrue
\else
 FIX THE PROOF.
\fi}
The Lemma is a consequence of the proof of Lemma~\ref{lem:main2} given in Appendix~\ref{app:sufficient}.
This lemma relates the probabilities of two paths,
$\exec$ and $\exec'$ of $\cA$, such that $\exec$ and $\exec'$ are equivalent, $\inseq(\exec)$ and $\inseq(\exec')$ are neighbors, and the initial transition of $\exec$ and $\exec'$ are assignment transitions. More precisely, for an initial value of $\rvar$, $x_0,$ Lemma~\ref{lem:main2} shows that $\pathprob{\epsilon,x_0,\exec'}$ is at least $\euler^{-\weight{\exec}\epsilon}$ times one of three quantities: $\pathprob{\epsilon,x_0,\exec}$, $\pathprob{\epsilon,x_0+1,\exec}$ or $\pathprob{\epsilon,x_0-1,\exec}.$   The specific quantity the Lemma compares $\pathprob{\epsilon,x_0,\exec'}$ to depends on some properties of the path $\exec$ stated in  Lemma~\ref{lem:main2}. Together these mutually exclusive properties serve as an exhaustive list of properties that the path $\exec$ can satisfy. The fact that the list is exhaustive is a consequence of well-formed-ness. In particular, one of the parts of the Lemma is that when the guard of the initial transition is $\true$
then $\pathprob{\epsilon,x_0,\exec'} \geq \euler^{-\weight{\exec}\epsilon} \pathprob{\epsilon,x_0,\exec}.$ This immediately implies the statement of the current Lemma.  
The proof of Lemma~\ref{lem:main2} itself is intricate and proceeds by induction on the number of assignment transitions in $\exec$.  
\end{proof}

\begin{example}
\label{ex:diff-priv-proof}
Let us consider the automata $\svtauto$ (Fig.~\ref{fig:svt-auto}) and $\numspauto$ (Fig.~\ref{fig:numsp-auto}). Both these automata are well-formed and hence they are differentially private. Moreover, we can use Lemma~\ref{lem:sufficiency} to provide an upper bound on the required privacy budget.

Observe that the only {\criticaltransition}s in $\svtauto$ are $t_{01}$, the transition from $q_0$ to $q_1$, and $t_{12}$, the transition from $q_1$ to $q_2$. Now $\cost{t_{01}} = \frac{1}{2}$, while $\cost{t_{12}} = 2(\frac{1}{4}) = \frac{1}{2}$. Thus, $\weight{\svtauto} = \frac{1}{2}+\frac{1}{2} = 1$, or $\svtauto$ is $\epsilon$-differentially private for all $\epsilon$.

Similarly, the only {\criticaltransition}s in $\numspauto$ are again transition $t_{01}$ from $q_0$ to $q_1$ and transition $t_{12}$ from $q_1$ to $q_2$. They have the following costs: $\cost{t_{01}} = \frac{4}{9}$ and $\cost{t_{12}} = 2(\frac{2}{9}) + \frac{1}{9} = \frac{5}{9}$. Thus, $\weight{\numspauto} = \frac{4}{9} + \frac{5}{9} = 1$ and $\numspauto$ is $\epsilon$-differentially private for all $\epsilon > 0$. 
\end{example}

\begin{remark}
Observe that the means used in sampling $\svar$ and $\svar'$ do not play any role in the definition of well-formed (Definition~\ref{def:well-formed}). They also do not play any role in the calculation of the weight of an automaton or Lemma~\ref{lem:sufficiency}. This allows one to make some simple observations. Recall that $\svtauto$ and $\numspauto$ were defined by taking the threshold $T = 0$. However, these observations allow us to conclude that no matter what value is chosen for the threshold $T$, $\svtauto$ and $\numspauto$ are $\epsilon$-differentially private for all $\epsilon > 0$.
\end{remark}

We get as a corollary of Theorem~\ref{thm:main} that the problem of checking whether a {\dipa} $\cA$ is differentially private can be checked using graph-theoretic algorithms in linear time.
\begin{corollary}
The differential privacy problem for {\dipautop} is decidable in linear time. In addition, $\weight{\cA}$ can be computed in linear time, assuming addition and comparison of numbers takes constant time.
\end{corollary}
\begin{proof}
We describe a linear time algorithm that checks whether a {\dipa} $\cA$ is well-formed. The Corollary then follows from Theorem~\ref{thm:main}.  

Let us fix $\cA = \defaut$. Consider the edge-labeled directed graph $\cG$ whose vertex set is $\states$ and there is an edge-labeled $(c,b)$ from $p$ to $q$ if $\transf(p,c) = (q,o,b)$ for some $o$. Without loss of generality, we can assume that every state is reachable from $\qinit$. It is worth observing that because of the {\detcond} condition of {\dipa}s, the number of edges in $\cG$ is at most twice the number of vertices. The subgraph $\cG_{\s{AG}}$ of $\cG$ has the same vertex set but an edge labeled $(c,b)$ is present in $\cG_{\s{AG}}$ only if whenever $b = \true$, $c = \getest$. Similarly, the subgraph $\cG_{\s{AL}}$ of $\cG$ only has those edges labeled $(c,b)$ with the property that if $b = \true$ then $c = \lttest$. Notice that the graphs $\cG$, $\cG_{\s{AG}}$ and $\cG_{\s{AL}}$ can each be constructed in linear time from $\cA$.

Next, we compute the maximal strongly connected components (SCC) of $\cG$; this can also be done in linear time. Observe that a state $q$ is part of some {\gcycle} if it's SCC has an edge with label $(\getest,b)$. Similarly, $q$ is part of some {\lcycle} if it's SCC has an edge with label $(\lttest,b)$. Notice that the set of all states that belong to some {\gcycle} and those that belong to some {\lcycle} can be computed in linear time. Next, the set of all vertices that can be reached by an {\agpath} from an {\lcycle} can be computed in linear time by performing a BFS on $\cG_{\s{AG}}$ starting from vertices that are on {\lcycle}s. Similarly, we can compute all vertices from which a {\gcycle} can be reached by an {\agpath} in linear time. Using BFS on $\cG_{\s{AL}}$ we can also compute the set of all vertices that can be reached from a {\gcycle} by an {\alpath}, and the set of all vertices from which an {\lcycle} can be reached by an {\alpath} in linear time.

We can now check each of the conditions of well-formed-ness in linear time using the sets computed in the previous paragraph.
\begin{itemize}
\item \emph{\criticalcycle}: Check if there is a SCC of $\cG$ that has an edge labeled $(c,\true)$ and an edge labeled $(c',b')$ where $c' \neq \true$.
\item \emph{\criticalpair}: Check if there is a state on an {\lcycle} that can reach a {\gcycle} by an {\agpath} and check if there is a state on an {\gcycle} that can reach a {\lcycle} by an {\alpath}.
\item \emph{\violatingc}: Check if there is a SCC of $\cG$ that contains an edge from an input state that outputs $\svar$ or $\svar'$.
\item \emph{\violatingp}: Check if any of the following conditions holds: (a) there is an {\agpath} ({\alpath}) from the target of an assignment transition to a state on a {\gcycle} ({\lcycle}); (b) there is an {\agpath} ({\alpath}) from the target of a non-assignment transition with output $\svar$ and guard $\lttest$ ($\getest$) to a state on a {\gcycle} ({\lcycle}); (c) there is an {\agpath} ({\alpath}) from a state on an {\lcycle} ({\gcycle}) to the source of a transition with guard $\getest$ ($\lttest$) that outputs $\svar$.
\end{itemize}

We now show how $\weight{\cA}$ can be computed in linear time assuming that arithmetic operations take constant time. Observe that we can construct the graph of SCCs of $\cG$ in linear time and that {\criticaltransition}s are those that correspond to edges in this graph of SCCs. $\weight{\cA}$ is the length of the longest path in this graph, where the weight of an edge is the cost of the corresponding transition. Note that this can be computed in linear time because the graph of SCCs is a DAG.
\end{proof}
\begin{remark}
Observe that the well-formed-ness of an automata $\cA$ does not depend on the parameter function $\parf$ of the automata. Hence, once we have established that $\cA$ is differentially private, we establish it for all possible parameter functions. The weight of a well-formed $\cA$, however, does indeed with the scaling parameters given by $\parf.$  It is independent of the mean parameters given by $\parf.$
\end{remark}
%\input{Section4_decidability}
%\input{sufficient}

%\input{edip-prelims}
%\input{extensions}
% !TEX root =  main.tex

\section{Related Work}
\label{sec:related}

\paragraph*{Privacy proof construction}
Several works~\cite{GHHNP13, AGHK18,RP10,AmorimAGH15,ZK17,CheckDP} have proposed the use of type systems to construct proofs of differential privacy. Some of the type-based approaches such as \cite{GHHNP13, AGHK18,RP10,AmorimAGH15} rely on linear dependent types, for which the type-checking and type-inference may be challenging. For example, the type checking problem for the type system in~\cite{AmorimAGH15} is undecidable. The type systems in Zhang and Kifer~\cite{ZK17}, later expanded on in~\cite{CheckDP}, rely on using the techniques of randomness alignments and can handle advanced examples such as the sparse vector technique. 
Barthe et al.~\cite{BKOZ13,BGGHS16,BFGGHS16} develop several program logics based on probabilistic couplings for reasoning about differential privacy, which have been used successfully to analyze standard examples from the literature, including the sparse vector technique. The probabilistic couplings and randomness alignment arguments are synthesized into coupling strategies by Albarghouthi and Hsu~\cite{AH18}. A shadow execution based method is introduced in~\cite{WDWKZ19}.
Both~\cite{AH18} and~\cite{WDWKZ19} are automated and can handle advanced examples such as sparse vector technique efficiently.
\blue{Probabilistic I/O automata are used in~\cite{TschantzKD11} to model interactive differential privacy algorithms. Simulation-based methods are used to verify differential privacy. They assume that inputs and outputs take values from a discrete domain and that the sampling is from discrete probability distributions.}
While these approaches can handle arbitrarily long sequences of inputs and verify $\epsilon$-differential privacy,
they are not shown to be complete and may fail to construct a proof of differential privacy even when the mechanism is differentially private.  

\paragraph*{Counterexample generation} 
Another investigation line develops automated techniques to search for privacy violations. %~\cite{DingWWZK18,BichselGDTV18} by searching amongst a bounded sequence of inputs. 
Ding et al. ~\cite{DingWWZK18} use statistical techniques based on hypothesis testing for automatic generation of counterexamples. Bischel et al. ~\cite{BichselGDTV18} use optimization-based techniques and symbolic differentiation to search for counterexamples. These methods search only amongst a bounded sequence of inputs and assume a concrete 
value of the parameter $\epsilon.$ Wang et al.~\cite{CheckDP} use program analysis techniques to generate counterexamples when it fails to construct a proof.  

\paragraph*{Model-checking/Markov Chain approaches}
\blue{The probabilistic model checking approach for verifying $\epsilon$-differential privacy  is employed in~\cite{ChatzGP14,LiuWZ18}, where it is assumed that the program is given as a Markov Chain.
These approaches do not allow for sampling from continuous random variables. Instead, they assume that the program behavior is given as a finite Markov Chain, and the transition probabilities are specified as inputs. Thus, they also implicitly assume a bounded sequence of inputs and a concrete value of $\epsilon.$
In~\cite{ChistikovKMP20}, the authors use labeled Markov Chains to model differential privacy algorithms. 
They consider discrete probability only, and can only model inputs taking values from a finite set. They also implicitly assume a concrete value of $\epsilon.$ Further, they check whether the ratio of probabilities of observations on neighboring inputs is bounded by a constant. If it is bounded, it implies  the algorithm is $\epsilon$-differentially private for sufficiently large epsilon. However, they do not provide a method to compute a possible $\epsilon$.}
%Labeled Markov Chains are eomployed in ~\cite{ChistikovMP18,ChistikovMP19} for verifying approximate differential privacy.}

%assume that the Markov Chain is given and that the transit

\paragraph*{Decision Procedures}
The decision problem of checking whether a randomized program is differentially private is studied in~\cite{BartheCJS020}, where it is shown to be undecidable for programs with a single input and single output, assuming that the program can sample from Laplacian distributions. They identify a language that restricts the mechanisms in order to obtain decidability. The restriction forces sampling from the Laplace distribution only a bounded number of times. The number of inputs and outputs are also bounded and constrained to take values from a finite domain.  The decision procedure in~\cite{BartheCJS020}
relies on the decision procedure for checking the validity of a sentence in the fragment of the theory of Reals with exponentiation identified in~\cite{mccallum2012deciding}, and has very high complexity.
% The decision procedure allows for verification of differential privacy for all $\epsilon.$
%
%
% The decision procedure developed in~\cite{BartheCJS020} converts the problem of checking differential privacy to checking the validity of first-order formulas in the theory of Reals with the exponential function. The formulas obtained in~\cite{BartheCJS020} fall into a fragment identified by~\cite{mccallum2012deciding} for which the validity problem is decidable. However, the verification algorithm, which relies on the decision procedure for the fragment of the theory of Reals with exponentiation, has very high complexity.
The decision procedure allows for verification of differential privacy for all $\epsilon.$

\paragraph*{Complexity} Gaboardi et. al~\cite{GaboardiNP19} study the complexity of deciding differential privacy for randomized Boolean circuits, and show that the problem is $\mathbf{coNP^{\#P}}$-complete. Their results are proved by reduction to majority problems. They assume finite number of inputs, the only probabilistic choices in~\cite{GaboardiNP19}  are fair coin tosses, and $\euler^{\epsilon}$ is taken to be a fixed rational number.
%To our best knowledge, no prior work is able to both prove
%differential privacy and detect its violations for a non-trivial class
%of programs.
%
%Our work is also loosely connected to prior attempts to relate
%differential privacy and information flow. In particular, Barthe and
%K\"opf~\cite{BartheK11} study information-theoretic bounds for
%differentially private channels and provide a decision procedure for
%rational bounds. Their decision procedure is based on a reduction to
%the theory of real closed fields (without exponential). However, there
%approach considers channels and is not directly applicable to a
%language-based setting.

\section{Conclusion}
\label{sec:conclusions}
In this paper, we introduced a model called {\dipautop} for modeling differential privacy mechanisms. Such automata can be used to model some of the interesting classes of mechanisms presented in the literature. We studied the problem of checking if a mechanism given by a {\dipa}  is {\em differentially private}, i.e., it is  $d\epsilon$-differentially private, for some constant $d>0$ and  for all values of the scaling parameter $\epsilon>0.$ We showed that this problem is decidable in time that is linear in the size of the automaton. Our decidability result is based on checking  the necessary and sufficient conditions for differential privacy, presented in the paper. If the mechanism, given by an automaton, is differentially private, then it outputs a constant $w$ such that the mechanism is $w\epsilon$-differentially private, for all $\epsilon>0.$ If the mechanism is not differentially private, 
%it outputs a 
a counterexample can be constructed explaining why it is not differentially private. For the published mechanisms presented in the literature, that are differentially private, the constant \blue{$d$} computed by our method matches the published values. The proofs showing that the given conditions presented in the paper, are necessary and sufficient for differential privacy, are highly non-trivial.

 As part of future work, it will be interesting to come up with computation of a smaller constant $d$, than the one given in the paper, for mechanisms modeled by {\dipa}, that are differentially private. Furthermore, it will be interesting to investigate new models of automata, that can describe other interesting sub-classes of mechanisms \blue{that are currently out-of-scope such as private smart sum
algorithm~\cite{CSS10}, private vertex cover~~\cite{GLMRT10} and NoisyMax~\cite{DR14}}, for which the problem of checking differential privacy can be decided efficiently. \blue{We also plan to investigate decision procedures for verifying approximate differential privacy when for unbounded sequence of inputs and outputs.}
 
 \section*{Acknowledgment}
 The authors would like to thank anonymous reviewers for their interesting and valuable comments.
 Rohit Chadha was partially supported by NSF CNS 1553548 and NSF CCF 1900924. A. Prasad Sistla was partially supported by NSF CCF 1901069, and Mahesh Viswanathan was partially supported by NSF NSF CCF 1901069 and NSF CCF 2007428.


%\section*{Acknowledgment}




\bibliographystyle{IEEEtran}
\bibliography{header,main}

  \ifdefined\AppendixTrue 
\appendices

\section{Auxiliary definitions}
\label{app:auxdefintiions}
% !TEX root =  main.tex

We shall start by defining some auxiliary definitions that shall help us in the proof of Theorem~\ref{thm:main}.

\paragraph*{Path Suffixes} Let $\cA=\defaut$ be an {\dipautos}. For any execution/path $\absexec=\execl{n}$ of $\cA$ and $i< n$, the suffix of $\exec$ starting from state $q_i$ (or position $i$) is the 
path $\execsf{i}{n}$ and is denoted as $\exec||i.$
 
\paragraph*{Abstract paths} For any execution/path $\absexec=\execl{n}$ of $\cA,$
the \emph{abstraction} of $\rho$, denoted $\abst(\rho)$, will be the word
$\eabsexecl{n}$ where  

$$\sigma_i=\begin{cases}
                          o_i & \mbox{if }o_i\in \outalph \\
                          \svar & \mbox{if } o_i=(\svar,r,s)\\
                          \svar' & \mbox{otherwise}
                    \end{cases}$$
Note that for {\dipautop}, $\sigma_i=o_i$ for each $i.$

A sequence $\absexec=\eabsexecl{n}$ is said to be an \emph{abstract path} if $\absexec=\abst(\exec)$ for some execution $\exec$. By abuse of notation, we shall say that the length pf the execution $\absexec$ is $n.$
Further such a $\exec$ shall be called an execution of $\absexec$ on input $\inalphaseq=a_0\cdots a_n.$ Note that $\exec$ is unique if $\sigma_i\in \outalph$
for each $i.$ In general, two distinct sequences $\exec$ and $\exec'$ having the same abstraction $\absexec$ will only differ at indices $i$ such that $\sigma_i \notin\outalph.$
At those indices, we would need to specify the values of the interval end-points, $r_i,s_i,$ where the real output is assumed to belong to.

Fix an abstract path  $\absexec=\eabsexecl{n}.$ The $i$th-transition, denoted $\trname[i]$, is the word  $q_i\sigma_{i}q_{i+1}.$  The guard of the $i$th transition, denoted $\grdname{\trname[i]}$ is the unique $c$
such that $\delta(q_i,c)=(q_i,\sigma_i,b).$ The output sequence of $\absexec$, denoted 
$\outseq(\absexec)$ is the sequence $\sigma_0\cdots \sigma_n.$   Note that we can classify transitions of an abstract path as input, non-input, assignment and non-assignment as expected. The notions of paths, cycles, reachability, {\criticalcycle}, {\criticalpair}, {\violatingc}, {\violatingp} and {\criticaltransition} 
extends  naturally to abstract paths.

\section{Necessity of well-formedness}
\label{app:necessity}
% !TEX root =  main.tex

%\newcommand{\defautprime} {(\states, \inalph, \outalph', \qinit, \vars, \parf, \transf')}

We shall now show that if the {\dipa} $\cA$ is not well-formed then $\cA$ is not differentially private, thus establishing the \lq\lq only if\rq\rq\ part of Theorem~\ref{thm:main}.
The proof of necessity will be broken into four Lemmas. Lemma~\ref{lem:criticalcnec} shall show that if $\cA$ has a {\criticalcycle} then $\cA$ is not differentially private. 
Lemma~\ref{lem:criticalpnec} will deal with presence of  {\criticalpair}s, Lemma~\ref{lem:violatingcnec}  with presence of {\violatingc}s, and Lemma~\ref{lem:violatingpnec} with presence of {\violatingp}s. Please note that we shall use the notions of path suffixes and abstract paths introduced in Appendix~\ref{app:auxdefintiions}.

Before we proceed, we need a technical lemma that characterizes the probability of two samples from Laplace distributions being ordered.
\begin{lemma}
\label{lem:problessequal}
Suppose $X_i$, for $i=1,2$, are  random variables with $X_i \sim \Lap{k_i,\mu_i}$. Then $\prbfn{X_1\leq X_2}$  is given as follows. When $k_1 \neq k_2$
\begin{dmath*}
\prbfn{X_1 \leq X_2} = \frac{1}{2}\left[1 + \sgn(\mu_2-\mu_1)\left(1- \frac{k_2^2}{2(k_2^2-k_1^2)}\eulerv{-k_1\card{\mu_2-\mu_1}} + \frac{k_1^2}{2(k_2^2-k_1^2)}\eulerv{-k_2\card{\mu_2-\mu_1}}\right)\right].
\end{dmath*}
On the other hand, when $k_1=k_2=k$
\begin{dmath*}
\prbfn{X_1 \leq X_2} = \frac{1}{2}\left[1+\sgn(\mu_2-\mu_1)\left(1 - \eulerv{-k\card{\mu_2-\mu_1}}(1+\frac{k}{2}\card{\mu_2-\mu_1})\right)\right]
\end{dmath*}
\end{lemma}

%We make another observation that is an 
%\red{refer to the lemmas}

\subsection*{\textbf{{\Criticalcycle}s implies no privacy}}

\begin{lemma}
\label{lem:criticalcnec}
A  {\dipa} $\cA $ is 
not differentially private if  it has a reachable {\criticalcycle}.
\end{lemma}

Let $\cA= \defaut.$ Assume that $\cA$ has a {\criticalcycle} reachable from the state $\qinit.$
We give the proof first assuming that all states of $\cA$ are input states. The proof for the case when the automata has both input and non-input states can be proved along similar lines and is left out. 
%Later we will show how we can modify the construction for automata that have both input and non-input states.

Let $\absexec=\eabsexecl{m+n}$  for $k=0,\ldots,m+n-1$ be an {abstract path} such that $q_0=\qinit$, $q_m=q_{m+n}$, and the  final $n$ transitions of $\rho$, i.e., the abstract path 
 $C=\eabsexecsf{m}{m+n}$ is a
{\criticalcycle}. 

%Further, we assume that for each $k,$ if $o_k=(\svar,u,v)$ or $(\svar',u,v)$ then $u=-\infty$ and $v=\infty.$

 Let $t_k$ be the {$k$-th transition} of $\absexec$ and $c_k$ be the {guard} of the $k$-th transition. 
  Further, let $d_k$ and $\mu_k$ be such that $\parf(q_k) = (d_k,\mu_k)$ for each $k.$
We have that
$c_0\:=\true$ and  $t_0$ is an
assignment transition.
Let $i,j$ be the smallest integers such that $m\leq
i<j<m+n$ and the following properties are satisfied: (a)   
$t_i$ is an assignment transition, (b) $c_j\neq \true$ and (c) for every
$k_1$ such that $i<k_1<j$, $t_{k_{1}}$ is a non-assignment transition and
$c_{k_{1}}=\true.$ We fix $i,j$ as above.
Consider any integer $\ell>0$. We define an abstract path $\absexec_\ell$ starting
from $\qinit$ by repeating the cycle $t_m,\ldots t_{m+n-1}$, $\ell$
times. Formally, $\absexec_\ell=\eabsexecl{m+\ell n}$ such that 
$q_k=q_{k-n}$ and $\sigma_k=\sigma_{k-n}$ for $m+n\leq
k\leq m+\ell n.$ Let
$\outgammaseq(\ell)=o_0\cdots o_{m+\ell n-1}$ be the output sequence of length $m+\ell n$ such that $o_k=\sigma_k$ if $\sigma_k\in \outalph,$ otherwise $o_k=(\sigma_k,-\infty,\infty).$
Once again, we let $t_k$ be the {$k$-th transition} of $\absexec_\ell$ and $c_k$ be the {guard} of the $k$-th transition. 
Now, given $\ell>0$, we define two neighboring input sequences
$\inalphaseq(\ell)=\inalpha_0\cdots \inalpha_{m+\ell n-1}$ and $\inbetaseq(\ell)=\inbeta_0\cdots \inbeta_{m+\ell n-1}$ each of length $m+\ell n.$  

%Both $\inalphaseq(\ell),\inbetaseq(\ell)$ are
%given by the sequences $\inalpha_k$ and $\inbeta_k$,$k=0,...,m+\ell n-1.$

The sequence $\inalphaseq(\ell)$ is chosen so that all the guards in
the transitions of $\absexec_{\ell}$ are satisfied with joint probability $>\frac{1}{2}$ 
{for large $\epsilon$}.
The input $\inalpha_0=0$ and for $0<k<m+\ell n$, $\inalpha_k$ is defined
inductively as given below: let $k'<k$ be the largest
integer such that $t_{k'}$ is an assignment transition, then $\inalpha_k$ is
given as follows: if  $c_k$ is the guard ${\getest}$ then
$\inalpha_k\:=\mu_{k'}-\mu_{k}+\inalpha_{k'}+1$, otherwise $\inalpha_k\:=\mu_{k'}-\mu_{k}+\inalpha_{k'}-1.$

%of outputs given by $o_k$, for $0\leq
%k<m+\ell n.$  

Now, consider any $k$, $0\leq k<m+\ell n$, such that
$c_k\neq \true$ and fix it. Let $k'<k$ be the largest integer such that $t_{k'}$ is
an assignment transition. Let $X_{k'},X_{k}$ be the two random
variables with distributions given by
$\Lap{d_{k'}\epsilon,\inalpha_{k'}}$ and
  $\Lap{d_{k}\epsilon,\inalpha_{k}}.$  Let
    $Y_k$ denote the random variable denoting the $k^{th}$ output of $\absexec$ on
    the input sequence $\inalphaseq(\ell).$ 
 Now
    consider the case when $c_k$ is ${\getest}.$ From the way, we
    defined $\inalphaseq(\ell)$ it is the case that $\mu_k+a_k\:=\mu_{k'}+a_{k'}+1.$ Now
    $\prbfn{Y_k\neq o_k}
\:= \prbfn{X_k < X_{k'}}\:=\prbfn{X_k\leq X_{k'}}.$ Let $d_{\mathsf{mx}}=\max(d_k,d_{k'})$ and $d_{\mathsf{mn}} = \min (d_k,d_{k'}).$
From Lemma \ref{lem:problessequal}, we see that 
if $d_k\neq d_{k'}$ then 
%$$\prbfn\{X_k \leq X_{k'}\}<
%\frac{(\max(d_{k'},d_{k}))^2}{2((\max(d_{k'},d_{k}))^2-(\min(d_{k'},d_{k}))^2)}\eulerv{-\min(d_{k'},d_{k})\epsilon}$$
$$\prbfn{X_k \leq X_{k'}}<
\frac{{d_{\mathsf{mx}}}^2}{2({d_{\mathsf{mx}}}^2-{d_{\mathsf{mn}}}^2)}\eulerv{-d_{\mathsf{mn}}\epsilon}.$$
If $k=k'$ then 
$$\prbfn{X_k \leq
X_{k'}}<\frac{1}{2}\eulerv{-d_{k}\epsilon}(1+\frac{d_{k}\epsilon}{2}).$$
From the above, we see that 
$$\prbfn{Y_k\neq o_k}\leq r
\eulerv{-d_{\mathsf{mn}}\epsilon}(1+\frac{d_{\mathsf{mx}}\epsilon}{2})$$
where $r$ is a constant that depends only on $\cA$ (and not on $k$).
Now consider the case when $c_k$ is $\lttest.$ In this case,
$\mu_k+\inalpha_k\:=\mu_k+\inalpha_{k'}-1$ and
    $\prbfn{Y_k\neq o_k}
\:= \prbfn{X_{k'}< X_{k}}.$ 
By a similar analysis, in this case also,
$$\prbfn{Y_k\neq o_k}\leq r
\eulerv{-d_{\mathsf{mn}}\epsilon}(1+\frac{d_{\mathsf{mx}}\epsilon}{2}).$$

%$$\prbfn\{Y_k\neq o_k\}\leq c
%\eulerv{-\min(d_{k'},d_{k})\epsilon}(1+\frac{\max(d_{k},d_{k'})\epsilon}{2})$$
%where $c$ is a constant. 

Let $d_{\max}\:=\max\set{\pi_1(\parf(q))\st q\in Q}$
and $d_{\min}\:=\min\set{\pi_1(\parf(q))\st q\in Q}.$ Then, for every $k,0\leq
k<m+\ell n$, 
$$\prbfn{Y_k\neq o_k}\leq r
\eulerv{-d_{\min}\epsilon}(1+\frac{d_{\max}\epsilon}{2})$$
%where $c$ is a constant.
Using the union rule of probabilities, we see that,
 $$\prbfn{\exists k<m+\ell n,\: Y_k\neq o_k}\:\leq r(m+\ell n)
\eulerv{-d_{\min}\epsilon}(1+\frac{d_{\max}\epsilon}{2}).$$ 
Given $\ell>0$, let $\epsilon_\ell\in\Reals$, be the smallest value such that
$$\forall \epsilon\geq \epsilon_\ell,\:  r(m+\ell n)
\eulerv{-d_{\min}\epsilon}(1+\frac{d_{\max}\epsilon}{2}) \leq \frac{1}{2}.$$

Now,  $$\pathprob{\epsilon,\exec_{\inalphaseq}(\ell)}\:=1-\prbfn{\exists k<m+\ell n,\: Y_k\neq o_k}.
%\geq  1-  c(m+\ell n)
%\eulerv{-d_{\min}\epsilon}(1+\frac{d_{max}\epsilon}{2}).
$$ 
From the construction of $\epsilon_\ell$ and above observations, we see that $\forall \epsilon\geq
\epsilon_\ell,\:{\pathprob{\epsilon,\exec_{\inalphaseq}(\ell)}}\:\geq \frac{1}{2}.$

Now, recall the integers $i,j$ fixed earlier. Intuitively, we define $\inbetaseq(\ell)$
so that each of the guards in the transitions $t_{j+\ell' n}, 0\leq
\ell'<\ell$ are satisfied with probability $<\frac{1}{2}$.
For each $\ell',\:0\leq \ell'<\ell$,  we let $\inbeta_{i+\ell'
  n}\:=\inalpha_{j+\ell' n}+\mu_j -\mu_i$ and $\inbeta_{j+\ell' n}\:=\inalpha_{i+\ell' n}+\mu_i-\mu_j.$
We observe the following. Now, for each $\ell',\:0\leq \ell'<\ell$,
the following hold. $c_{j+\ell' n}\:=c_j\neq \true.$ If $c_{j+\ell' n}$ is the
guard ${\getest}$ then $\inbeta_{i+\ell' n}+\mu_i\:=\inbeta_{j+\ell' n}+\mu_j+1$ since 
$\inalpha_{j+\ell' n}+\mu_j\:=\inalpha_{i+\ell' n}+\mu_i+1.$  If $c_{j+\ell' n}$ is the
guard $\lttest$ then $\inbeta_{j+\ell' n}+\mu_j\:=\inbeta_{i+\ell' n}+
\mu_i+1$ since
$\inalpha_{i+\ell' n}+\mu_i\:=\inalpha_{j+\ell' n}+\mu_j+1.$  We define $\inbeta_{i'}$, for all values of $i'<m+\ell
n$ and $i'\notin \set{ i+\ell' n, j+\ell' n \st 0\leq \ell'<\ell}$, so
that $\inbetaseq(\ell)$ is a neighbour of $\inalphaseq(\ell).$ It is not difficult to see
that such a sequence $\inbetaseq(\ell)$ can be defined. Let $\exec_{\inbetaseq}(\ell)$ be the path such that $\abst(\exec_{{\inbetaseq}(\ell)})=\absexec(\ell)$ 
and $\inseq(\exec_{\inbetaseq}(\ell))=\inbetaseq(\ell).$

For each $k, 0\leq
k<m+\ell n$, let $U_k$ be the random variable with 
distribution given by $\Lap{d_{q_{k}}\epsilon,\inbeta_k}$ and $Z_k$ be 
denoting the $k^{th}$ output of $\absexec$ on
    the input sequence $\inbetaseq(\ell).$ 
%the random variable denoting the output value of $\cA$ on input $\inbeta_k.$  
Let $d'\:=\min(d_{i},d_{j})$ and
$d''\:=\max(d_{i},d_{j}).$ Now, $\prbfn{Z_j=o_j}$ is given by
$\prbfn{U_j\geq U_i}$ if $c_j$ is the guard ${\getest}$,
otherwise it is given by $\prbfn{U_j\leq U_i}.$ Using Lemma
\ref{lem:problessequal} and similar reasoning as given earlier, we see that 
$$\prbfn{Z_j=o_j}\leq
r'\eulerv{-d'\epsilon}(1+\frac{d''\epsilon}{2})$$
for some constant $r'.$
 For each
$\ell',\:0<\ell'<\ell$, using the same reasoning as above with the random variables $U_{i+\ell' n},U_{j+\ell'n}$, we see that 
$$\prbfn{Z_{j+\ell' n}=o_{j+\ell' n}}\leq
r'\eulerv{-d'\epsilon}(1+\frac{d''\epsilon}{2}).$$
 Since for any
$\ell_1,\ell_2$ such that $0\leq\ell_1< \ell_2<\ell$, the random
variables $U_{i+\ell_1 n},U_{j+\ell_1 n}$ are independent of
$U_{i+\ell_2 n},U_{j+\ell_2 n}$,
we see that 
$$\prbfn{\forall \ell',0\leq \ell'<\ell,\:Z_{j+\ell'
  n}=o_{j+\ell' n}}\leq {r'}^\ell \eulerv{-d'\ell
  \epsilon}(1+\frac{d''\epsilon}{2})^\ell.$$
Thus,
$$\prbfn{\forall k, 0\leq k<m+\ell n,
Z_k=o_k}\leq  {r'}^\ell \eulerv{-d'\ell
  \epsilon}(1+\frac{d''\epsilon}{2})^\ell.$$
The LHS of the above equation is exactly
{$\pathprob{\epsilon,\exec_{\inbetaseq}(\ell)}$.}

Thus,  for any $\ell>0$, 
$$\forall \epsilon\geq \epsilon_\ell,\:
\frac{{\pathprob{\epsilon,\exec_{\inalphaseq}(\ell)}}}{{\pathprob{\epsilon,\exec_{\inbetaseq}(\ell)}}}\geq\frac{1}{2}(\frac{\eulerv{d'\epsilon}}{r'(1+\frac{d''\epsilon}{2})})^\ell.$$  
We claim that for any $s>0$, $\exists \ell,\epsilon$ such that
$$\frac{1}{2}(\frac{\eulerv{d'\epsilon}}{r'(1+\frac{d''\epsilon}{2})})^\ell>\eulerv{s\epsilon}.$$
Now the above inequality holds if 
$$\frac{\eulerv{(d'\ell-s)\epsilon}}{(1+\frac{d'\epsilon}{2})^\ell}>2{r'}^\ell.$$ 

Choose $\ell$ so that $d'\ell>s.$ Since the denominator of the left hand side term of the last
inequality grows polynomially in $\epsilon$, while its numerator grows
exponentially in $\epsilon$, it is easy to see that 
 $\exists
  \epsilon_{0}>\epsilon_\ell$ such that
  $$ \forall \epsilon\geq \epsilon_0, \:\frac{\eulerv{(d'\ell-s)\epsilon}}{(1+\frac{d'\epsilon}{2})^\ell}>2{r'}^\ell.$$
The crucial observation we now make is that, thanks to output determinism, for every input sequence $\alpha$ and output sequence $\gamma$, there is at most one path $\exec_{\alpha,\gamma}$ such that $\inseq(\exec_{\alpha,\gamma})=\alpha$ and $\outseq(\exec_{\alpha,\gamma})=\alpha.$  This observation combined with the above inequality shows
that $\cA$ is not differentially private.

%\red{What changes need to be made for non-input}

\subsection*{\textbf{{\Criticalpair}s implies no privacy}}

\begin{lemma}
\label{lem:criticalpnec}
A  {\dipa} $\cA$ is 
not differentially private if  it has a {\criticalpair} of cycles $(C,C')$ such that $C$ is reachable from the initial state of $\cA.$   

\end{lemma}
%\label{app:if2main}
%\begin{definition}
%Let $\inalphaseq$ be an input sequence, $\outgammaseq$ an output sequence, such that there is a path $\rho$ from initial state $\qinit$
%with $\inseq(\rho)=\inalphaseq$ and $\outseq(\rho)=\outgammaseq.$ Given $\epsilon>0$, let $Y(\epsilon)$ be the random variable that models the value of the variable $\rvar$ 
%at the end of execution of $\rho.$  Let $\pdfassgn{\alpha}{\gamma}{z}$ be the probability density function of the random variable $Y(\epsilon).$ 
%\end{definition}
%

\begin{proof}
Thanks to Lemma~\ref{lem:criticalcnec}, we can assume $\cA$ does not have a {\criticalcycle}. 
Let $\cA= \defaut.$
Assume that $\cA$  has a {\criticalpair} of cycles $(C,C')$ such that $C$ is reachable from $\qinit.$
Assume that $C$ is an {\lcycle} and $C'$ is a {\gcycle}. (The proof for
the case when $C$ is a {\gcycle} and $C'$ is an {\lcycle} is similar but
symmetric and is left out). 
%Without loss of generality, we assume that neither $C$ nor $C'$ by themselves  are {{\criticalcycle}}s; otherwise, by Lemma  \ref{lem:if1main}, $\cA$ is not differentially private. 
Thanks to our
assumption that we do not have {\criticalcycle}s, it means that both $C,C'$ do not have assignment transitions. We
further assume that $C,C'$ are distinct. If they are the same then it is
straightforward to prove that $\cA$ is not differentially private,
using  more or less the same proof. We also assume that all the states in $\cA$ are input states. The case when $\cA$ has both input and 
non-input states can also be proved using more or less the same proof. 

Let the lengths of $C,C'$ be $n_1,n_2$, respectively. 
Now, for any $\ell>0$, consider the following abstract path $\absexec_{\ell}$ in $\cA$ starting
from $\qinit$ in which the cycles $C,C'$ are repeated $\ell$ times
each. The path
$$\begin{array}{lcl}
\absexec_{\ell}&=& q_0\sigma_0 \cdots q_u \sigma_u \cdots q_v \sigma_v \cdots q_{v+n_1\ell-1} \sigma_{v+n_1\ell-1}\cdots \\
                        && \hspace*{1cm} \cdots q_{w}\sigma_w\cdots o_{w+n_2\ell-1} q_{w+n_2\ell}
\end{array}$$
where the following guards are satisfied. For each k, let $t_k$ be the $k$-th transition of $\absexec.$ and $c_k$ be the gaurd of the $k$-th transition.
\begin{enumerate}
\item $q_0=\qinit$
\item $\eabsexecsf{v}{v+n_1}$ is the cycle C
\item $t_{j+n_1}\:=t_{j}$ for all $j,\: v\leq j<v+n_1(\ell-1)$
\item $\eabsexecsf{w}{w+n_2}$  is the cycle $C'$
\item $t_{j+n_2}\:=t_{j}$ for all $j,\:w\leq j<w+n_2(\ell-1)$
\item $t_u$ is an assignment transition and $\forall\:j,
u<j<v+n_1\ell$ and $\forall\:j,
j\geq w$, $t_j$ is a non-assignment transition
\item for all
$j,\:v+n_1\ell\leq j<w $, if $t_j$ is an assignment transition then
$c_j$ is the guard $\getest.$

\end{enumerate}
%
%$\absexec_{\ell}\:=(t_0,...,t_u,...,t_v,....,t_{v+n_1\ell-1},..,t_{w},...,t_{w+n_2\ell-1})$
%where  $t_k\:=(q_k,a_k,o_k,f(a_k,\epsilon),c_k,A_k,q_{k+1})$, for
%$0\leq k <w+n_2\ell$, $u\leq v$, $v+n_1\ell\leq w$ and the following conditions are satisfied: (i)
%$q_0\:=\qinit$;(ii)  $(t_v,...,t_{v+n_1-1})$ is the cycle $C$ and
%$q_v\:=q_{v+n_1}$; (iii) $t_{j+n_1}\:=t_{j}$ for all $j,\: v\leq
%j<v+n_1(\ell-1)$; (iv) $(t_w,...,t_{w+n_2-1})$ is the cycle $C'$ and
%$q_w\:=q_{w+n_2}$; (v) $t_{j+n_2}\:=t_{j}$ for all $j,\:w\leq
%j<w+n_2(\ell-1)$; (vi) $t_u$ is an assignment transition and $\forall\:j,
%u<j<v+n_1\ell$ and $\forall\:j,
%j\geq w$, $t_j$ is a non-assignment transition; (vii) for all
%$j,\:v+n_1\ell\leq j<w $, if $t_j$ is an assignment transition then
%$c_j$ is the condition $\getest.$
 Observe that the last
assignment transition before $t_{v+n_1\ell}$ is $t_u$, all assignment
transitions from $t_{v+n_1\ell}$ up to $t_w$ have ${\getest}$ as
their guard, the segment of the path from $t_v$ to
$t_v+n_1\ell-1$ is the part where cycle $C$ is repeated $\ell$ times
and the segment of the path from $t_w$ to
$t_w+n_2\ell-1$ is the part where cycle $C'$ is repeated $\ell$ times.
Let $d_k$ and $\mu_k$ be such that $\parf(q_k) = (d_k,\mu_k)$ for each $k.$
We have that 
$c_0\:=\true$ and  $t_0$ is an
assignment transition.

Let
$\outgammaseq(\ell)=o_0\cdots o_{m+\ell n-1}$ be the output sequence of length $m+\ell n$ such that $o_k=\sigma_k$ if $\sigma_k\in \outalph,$ otherwise $o_k=(\sigma_k,-\infty,\infty).$
Once again, we let $t_k$ be the {$k$-th transition} of $\absexec_\ell$ and $c_k$ be the {guard} of the $k$-th transition. 
Now, given $\ell>0$, we define two neighboring input sequences
$\inalphaseq(\ell)=\inalpha_0\cdots \inalpha_{m+\ell n-1}$ and $\inbetaseq(\ell)=\inbeta_0\cdots \inbeta_{m+\ell n-1}$ each of length $m+\ell n.$

Now, we define two adjacent input sequences $\inalphaseq(\ell)=\inalpha_0\cdots\inalpha_{w+n_2\ell-1}$ and
$\inbetaseq(\ell)\:=\inbeta_0\cdots\inbeta_{w+n_2\ell-1}$ as follows. For all
$j,0\leq j< v$ and for all $j,\:v+n_1\ell \leq
j<w$,$\inalpha_j\:=\inbeta_j\:=0$; for all $j,\:v\leq j<v+n_1\ell$ and  for all $j,\:w\leq
j<w+n_2\ell$, if $c_j$ is the guard $\getest$ then
$\inalpha_j\:=\frac{1}{2}-\mu_j,\: \inbeta_j\:=-\frac{1}{2}-\mu_j$, if $c_j$ is the
guard  $\lttest$ then
$\inalpha_j\:=-\frac{1}{2}-\mu_j,\: \inbeta_j\:=\frac{1}{2}-\mu_j$ and if $c_j$ is
$\true$ then $\inalpha_j\:=\inbeta_j\:=0.$ It is not difficult to see that
$\inalphaseq(\ell)$ and $\inbetaseq(\ell)$ are adjacent.
Let $\exec_{\inalphaseq}(\ell)$ be the path such that $\abst(\exec_{{\inalphaseq}(\ell)})=\absexec(\ell)$ 
and $\inseq(\exec_{\inalphaseq}(\ell))=\inalphaseq(\ell).$ 
Let $\exec_{\inbetaseq}(\ell)$ be the path such that $\abst(\exec_{{\inbetaseq}(\ell)})=\absexec(\ell)$ 
and $\inseq(\exec_{\inbetaseq}(\ell))=\inbetaseq(\ell).$

% Let $\outgammaseq(\ell)\:=o_0\cdots o_j\cdots$ be the sequence of outputs in the transitions of $\absexec_{\ell}.$

Let $X_j,U_j$ be random variables with distributions given by
$\Lap{d_{j}\epsilon,\inalpha_j+\mu_j}$ and
$\Lap{d_{j}\epsilon,\inbeta_j+\mu_j}$, respectively. 
%Recall that $\pdfassgn{\inalphaseq(\ell)|v}{\outgammaseq(\ell)|v}{z}$ is the p.d.f of
%the random variable that models the value of $\rvar$ after $\ell$  steps of  $\absexec$ on input $\inalphaseq.$ 
%Since $\inalphaseq(\ell)|v\:=\inbetaseq(\ell)|v$ we see
%that $\pdfassgn{\inalphaseq(\ell)|v}{\outgammaseq(\ell)|v}{z}\:=
%\pdfassgn{\inbetaseq(\ell)|v}{\outgammaseq(\ell)|v}{z}.$ 
Observe that $t_u$
is the last assignment transition in $\absexec_{\ell}.$ For each $j>u$, for any
given $y\in \Reals$, let
$g_j(y),h_j(y)$ be the probabilities defined as follows: if $c_j$ is
the guard ${\getest}$ then $g_j(y)\:=\prbfn{X_j\geq y}$ and 
$h_j(y)\:=\prbfn{U_j\geq y}$; if $c_j$ is
the guard $\lttest$ then $g_j(y)\:=\prbfn{X_j< y}$ and 
$h_j(y)\:=\prbfn{U_j< y}$; if $c_j$ is $\true$ then
$g_j(y)\:=h_j(y)\:=1.$ It should be easy to see that,  for  all $j,\:u<j<v$ and for all $j,\:v+n_1\ell \leq
j<w$,
$\inalpha_j\:=\inbeta_j$ and hence
$g_j(y)\:=h_j(y).$
Now, we have the following claim.

%Let $z$ be a
%variable taking values from $\Reals$ and $\rho_z$ be the evaluation
%where $\rho(x)\:=z.$ From our definitions, we see that
%$\pdfassgn{\inalphaseq(\ell)|v}{\outgammaseq(\ell)|v}{z}$ is the probability
%that $\cA$ on the input sequence $\inalphaseq(\ell)|v$ outputs $\outgammaseq(\ell)|v$
%under the guard that the variable $x$ has value $z$ at the end of the
%input $\inalphaseq(\ell)|v.$ Since $\inalphaseq(\ell)|v\:=\inbetaseq(\ell)|v$ we see
%that $\pdfassgn{\inalphaseq(\ell)|v}{\outgammaseq(\ell)|v}{z}\:=
%\pdfassgn{\inbetaseq(\ell)|v}{\outgammaseq(\ell)|v}{z}.$ Observe that $t_u$
%is the last assignment transition in $\absexec_{\ell}.$ For each $j>u$, for any
%given $y\in \Reals$, let
%$g_j(y),h_j(y)$ be the probabilities defined as follows: if $c_j$ is
%the guard ${\getest}$ then $g_j(y)\:=\prbfn\{X_j\geq y\}$ and 
%$h_j(y)\:=\prbfn\{U_j\geq y\}$; if $c_j$ is
%the guard $\lttest$ then $g_j(y)\:=\prbfn\{X_j< y\}$ and 
%$h_j(y)\:=\prbfn\{U_j< y\}$; if $c_j$ is $\true$ then
%$g_j(y)\:=h_j(y)\:=1.$ It should be easy to see that,  for  all $j,\:u<j<v$ and for all $j,\:v+n_1\ell \leq
%j<w$,
%$\inalpha_j\:=\inbeta_j$ and hence
%$g_j(y)\:=h_j(y).$
%Now, we have the following claim.

{\bf Claim:} For all $j, v\leq j <v+n_1\ell$, and for all
$j,\:w\leq j<w+n_2\ell$, it is the case that $g_j(y)\geq h_j(y)$ for
all $y\in \Reals$, and the following
additional inequalities hold.
 
\begin{enumerate}

\item If $y\leq 0$ and  $c_j$ is the guard $\lttest$  then $g_j(y)
    \geq \eulerv{\frac{1}{2}d_{j}\epsilon}h_j(y).$ 

\item  If $y>0$ and $c_j$ is the guard
  ${\getest}$ then  $g_j(y)
    \geq \eulerv{\frac{1}{2}d_{j}\epsilon}h_j(y).$

\end{enumerate}

\begin{proof}
 Observe that when $c_j\:=\true$
then trivially $g_j(y)\:=h_j(y).$ Now, consider the case when
$y<-\frac{1}{2}.$ If $c_j$ is the guard $\getest$ then
$g_j(y)\:= 1-\frac{1}{2}\eulerv{-d_{j}\epsilon(\frac{1}{2}-y)}$ and
$h_j(y)\:=1-\frac{1}{2}\eulerv{-d_{j}\epsilon(-\frac{1}{2}-y)}$(this
is so since $\inalpha_j+\mu_j=\frac{1}{2}$ and $\inbeta_j+\mu_j=-\frac{1}{2}$) ; in
this case $\frac{1}{2}-y>-\frac{1}{2}-y$ and hence $g_j(y)\geq
h_j(y).$ If $c_j$ is the guard $\lttest$ then $g_j(y)\:=
\frac{1}{2}\eulerv{-d_{j}\epsilon(-\frac{1}{2}-y)}$ and
$h_j(y)\:=\frac{1}{2}\eulerv{-d_{j}\epsilon(\frac{1}{2}-y)}$; from
this we see that $g_j(y)\geq  \eulerv{d_{j}\epsilon}h_j(y).$

Now consider the case when $y\in [-\frac{1}{2},0].$ If $c_j$ is the
guard $\getest$ then
$g_j(y)\:= 1-\frac{1}{2}\eulerv{-d_{j}\epsilon(\frac{1}{2}-y)}$ and
$h_j(y)\:=\frac{1}{2}\eulerv{-d_{j}\epsilon(y+\frac{1}{2})}$;
since $g_j(y)\geq \frac{1}{2}$ and $h_j(y)\leq \frac{1}{2}$, we see that $g_j(y)
    \geq h_j(y).$   If $c_j$ is the
guard $\svar < x$ then 
$g_j(y)\:= 1-\frac{1}{2}\eulerv{-d_{j}\epsilon(y+\frac{1}{2})}$
and $h_j(y)\:=\frac{1}{2}\eulerv{-d_{j}\epsilon(\frac{1}{2}-y)}$;
since $g_j(y)\geq \frac{1}{2}$, we see that $g_j(y)\geq
\eulerv{\frac{1}{2}d_{j}\epsilon}h_j(y).$

Now consider the case when $y>0.$ If $y\leq \frac{1}{2}$ and $c_j$ is
${\getest}$ then $g_j(y)\:=
1-\frac{1}{2}\eulerv{-d_{j}\epsilon(\frac{1}{2}-y)}$ and 
$h_j(y)\:=\frac{1}{2}\eulerv{-d_{j}\epsilon(y+\frac{1}{2})}$;
observe that $g_j(y)\geq \frac{1}{2}$ and $h_j(y)\leq
\frac{1}{2}\eulerv{-\frac{1}{2}d_{j}\epsilon}$; from this we get the desired
  inequality. 

If $y\leq \frac{1}{2}$ and $c_j$ is
$\lttest$ then $g_j(y)\:=
1-\frac{1}{2}\eulerv{-d_{j}\epsilon(y+\frac{1}{2})}$ and $h_j(y)\:=\frac{1}{2}\eulerv{-d_{j}\epsilon(\frac{1}{2}-y)}$; since $g_j(y)\geq \frac{1}{2}$ and $h_j(y)\leq \frac{1}{2}$, we see $g_j(y)\geq h_j(y).$ If $y>\frac{1}{2}$ and $c_j$ is 
${\getest}$ then
$g_j(y)\:=\frac{1}{2}\eulerv{-d_{j}\epsilon(y-\frac{1}{2})}$ and
$h_j(y)\:=\frac{1}{2}\eulerv{-d_{j}\epsilon(y+\frac{1}{2})}$; from
this we see that the
desired inequality follows easily. If $y>\frac{1}{2}$ and $c_j$ is 
$\lttest$ then
$g_j(y)\:=1-\frac{1}{2}\eulerv{-d_{j}\epsilon(y+\frac{1}{2})}$ and
$h_j(y)\:=1-\frac{1}{2}\eulerv{-d_{j}\epsilon(y-\frac{1}{2})}$; it
is easy to see that $g_j(y)\geq h_j(y).$
\end{proof}

Let $S_1(\ell)$ be the set of all $j$ such that  $v\leq j<v+n_1\ell$
and $c_j$ is the guard $\svar < x.$ 
Let $S_2(\ell)$ be the set of all $j$ such that  
$w\leq j<w+n_2\ell$, and $c_j$ is the guard ${\getest}.$ 
Since $C$ is an {\lcycle} and $C'$ is a {\gcycle}, we see that the
cardinalities of both $S_1(\ell)$ and $S_2(\ell)$ are $\geq \ell.$
 Let $d_{\min}\:= \min\set{d_{j}\st j\in
  S_1(\ell)\cup S_2(\ell)}.$ Clearly $d_{\min}>0.$ 

%WIth out loss of generality, we assume that
%$inseq(\exec_{\inalphaseq}(\ell))\:=\inalphaseq(\ell).$ Let $\exec_{\inbetaseq}(\ell)$ be a path equivalent
%to $\exec_{\inalphaseq}(\ell)$ such that $inseq(\exec_{\inbetaseq}(\ell))\:=\inbetaseq(\ell).$ Recall that
%$\exec_{\inalphaseq}(\ell)||w$ is the suffix of $\exec_{\inalphaseq}(\ell)$ starting with $t_w$ and it
%is of length $n_2\ell.$

%Let $\exec_{\inalphaseq}(\ell)=\absexec[\inalphaseq]$ be the path that results from the execution of 
%$\absexec$ on input $\inalphaseq.$  Let $\exec_{\inbetaseq}(\ell)=\absexec[\inbetaseq]$ be the path that results from the execution of 
%$\absexec$ on input $\inbetaseq.$ 

Let $\exec_{\inalphaseq}(\ell)$ be the path such that $\abst(\exec_{{\inalphaseq}(\ell)})=\absexec(\ell)$ 
and $\inseq(\exec_{\inalphaseq}(\ell))=\inalphaseq(\ell).$ 
For $k\leq q+n_2\ell$, let
$\exec_{\inalphaseq}(\ell)||k$ (resp, $\exec_{\inbetaseq}(\ell)||k$) be the suffix of $\exec_{\inalphaseq}(\ell)$ (resp. $\exec_{\inalphaseq}(\ell)||k$) starting with $q_k.$ 
%These suffixesis of length $n_2\ell.$

Since $C'$ is a {\gcycle}, from the above claim, we see that 
$\forall y\in \Reals$, $\pathprob{\exec_{\inalphaseq}(\ell)||w,y}\geq
\pathprob{\exec_{\inbetaseq}(\ell)||w,y}$,  and $\forall y>0$, $\pathprob{\exec_{\inalphaseq}(\ell)||w,y}\geq
\eulerv{\frac{1}{2}d_{\min}\ell\epsilon}\pathprob{\exec_{\inbetaseq}(\ell)||w,y}.$ Using the
above property and the previous  claim, together with the
assumption that 
$\forall j,\:v+n_1\ell\leq j<w$, if $t_j$ is an assignment transition
then it's guard is $\getest$, the following can be proved
by downward induction on $k$,  $\forall k,\:v+n_1\ell \leq k <w$: 
$\forall y\in \Reals$, $\pathprob{\exec_{\inalphaseq}(\ell)||k,y}\geq
\pathprob{\exec_{\inbetaseq}(\ell)||k,y}$,  and $\forall y>0$, $\pathprob{\exec_{\inalphaseq}(\ell)||k,y}\geq
\eulerv{\frac{1}{2}d_{\min}\ell\epsilon}\pathprob{\exec_{\inbetaseq}(\ell)||k,y}.$

Now, it should be easy to see that $\forall y\in\Reals,$
$$
\begin{array}{ll}
\pathprob{y,\exec_{\inalphaseq}(\ell)||v} =  (\prod_{v\leq  j<v+n_1\ell}g_j(y)) 
\pathprob{y,\exec_{\inalphaseq}(\ell)||v+n_{1}\ell} \\
\pathprob{y,\exec_{\inbetaseq}(\ell)||v} = (\prod_{v\leq  j<v+n_1\ell}h_j(y)) 
\pathprob{y,\exec_{\inbetaseq}(\ell)||v+n_{1}\ell}.
\end{array}$$
 Observe that $\forall
j,\: v\leq  j<v+n_1\ell,$ 
$$\begin{array}{ll}
\forall y\leq 0: & g_j(y)\geq
\eulerv{\frac{1}{2}d_{\min}\epsilon} h_j(y), 
\\
&\hspace*{0.2cm}\pathprob{y,\exec_{\inalphaseq}(\ell)||j}\geq
\pathprob{y,\exec_{\inbetaseq}(\ell)||j} \\
%\end{array}$$ 
\mbox{and}\\
% $$\begin{array}{ll}
 \forall y>0: & g_j(y)\geq h_j(y),\\
   &\hspace*{0.2cm} \pathprob{y,\exec_{\inalphaseq}(\ell)||j}\geq
\eulerv{\frac{1}{2}d_{\min}\ell \epsilon} \pathprob{y,\exec_{\inbetaseq}(\ell)||j}.
\end{array}$$
From this we get the
following: $$\forall y\in \Reals,\: \pathprob{y,\exec_{\inalphaseq}(\ell)||v}\geq
\eulerv{\frac{1}{2}d_{\min}\ell \epsilon} \pathprob{y,\exec_{\inbetaseq}(\ell)||v}.$$
% Now, we observe that 
% $$
% \begin{array}{l}
% \red{\pathprob{\epsilon,\exec_{\inalphaseq}(\ell)}}= \int^{\infty}_{-\infty}
%\pdfassgn{\inalphaseq(\ell)|v}{\outgammaseq(\ell)|v}{y} \pathprob{y,\exec_{\inalphaseq}(\ell)||v}dy.\\
%\red{ \pathprob{\epsilon,\exec_{\inbetaseq}(\ell)}}= \int^{\infty}_{-\infty}
%\pdfassgn{\inbetaseq(\ell)|v}{\outgammaseq(\ell)|v}{y} \pathprob{y,\exec_{\inbetaseq}(\ell)||v}dy.\\
% \end{array}
% $$
%From the facts that $\pathprob{y,\exec_{\inalphaseq}(\ell)||v}
%\geq \eulerv{\frac{1}{2}d_{\min}\ell\epsilon}\pathprob{y,\exec_{\inbetaseq}(\ell)||v}$ and that $\pdfassgn{\inalphaseq(\ell)|v}{\outgammaseq(\ell)|v}{y}\:=
%\pdfassgn{\inbetaseq(\ell)|v}{\outgammaseq(\ell)|v}{y}$, we see
Using this we can show by the definition of probability of a path that
$${\frac{\pathprob{\epsilon,\exec_{\inalphaseq}(\ell)}}{\pathprob{\epsilon,\exec_{\inbetaseq}(\ell)}}}\geq
\eulerv{\frac{1}{2}d_{\min}\ell\epsilon}.$$ Since  $\ell$ can be made arbitrarily large, we see that $\cA$ is not
$d\epsilon$-differentially private, for any $d>0$. Hence $\cA$ is not
differentially private. 
\end{proof}

\subsection*{\textbf{{\Violatingc}s implies no privacy}}

\begin{lemma}
\label{lem:violatingcnec}
A  {\dipa} $\cA $ is 
not differentially private if  it has a reachable {\violatingc}.
\end{lemma}
\begin{proof}
Thanks to Lemma~\ref{lem:criticalcnec} and Lemma~\ref{lem:criticalpnec}, we can assume $\cA$ does not have  {\criticalcycle}s or {\criticalpair}s. 
Assume that $\cA$ is well-formed, but there is a reachable {\violatingc} $C$ in $\cA$ that
has a transition whose output is $\svar.$ The proof for the case
when $C$ has a transition whose output is $\svar'$ is simpler and
is left out.  Now, if the transition of $C$ whose output is $\svar$ has the guard $\true,$ 
then it can be shown easily that repeating the cycle $\ell$ times incurs a privacy cost linear in  
$\ell\epsilon,$ and hence $\cA$ cannot be $d\epsilon$-differentially private for any $d>0.$
Thus, we consider more interesting case when the guard is $\lttest$ or $\getest.$

We consider the case when $C$ has a transition with output  $\svar.$ 
Since $\cA$ is well-formed
the cycle $C$ has no assignment transitions. 
Let $\absexec=\eabsexecl{j+m}$  for $k=0,\ldots,j+m-1$ be an {abstract path} such that $q_0=\qinit$, $q_j=q_{j+m}$, and the final $m$ transitions of $\rho$
is the {abstract cycle} corresponding to C.  
Fix $0\leq r<m$ be such that $\sigma_{j+r}=\svar.$
We assume that the guard of the $(j+r)$-th transition is $\getest.$  The case when it is $\lttest$ is similar and left out. 
  Further, let $d_k$ and $\mu_k$ be such that $\parf(q_k) = (d_k,\mu_k)$ for each $k.$

Fix $\ell>0.$ We define an abstract path $\absexec_\ell$ starting
from $\qinit$ by repeating the cycle $C$ $\ell$
times. Formally, $\absexec_\ell=\eabsexecl{j+\ell m}$ such that 
$q_k=q_{k-m}$ and $\sigma_k=\sigma_{k-n}$ for $j+m\leq
k\leq j+\ell m.$ Let $t_k$ be the {$k$-th transition} of $\absexec_\ell$ and $c_k$ be the {guard} of the $k$-th transition.  We have that $\sigma_{j+nm+r}\:=\svar$, for all $n$ such that $0\leq
n<\ell.$

%Let $0\leq i<j$ be such that $t_i$ is the last assignment transition of $\absexec.$ 
Now we construct two input sequences 
$\inalphaseq(\ell)=a_0\cdots a_{j+\ell m-1}$ and $\inbetaseq(\ell)=b_0\cdots b_{j+\ell m-1}$ as follows. We take $a_k=-\mu_k$, for all $k, 0\leq k< j+\ell m$ such that $t_k$ is an input transition, otherwise we take $a_k=\tau.$
We take $b_k=-\mu_k-1$ if $k=j+nm+r$ for some $0\leq n<\ell$ and $b_k=a_k$ otherwise. 
Let $\exec(\ell)=\execl{j+\ell m}$ be the path such that 
\begin{itemize}
\item $\absexec=\abst(\exec(\ell)),$ 
\item $\inseq(\exec(\ell))=\inalphaseq(\ell),$  and  
\item all $k$, i) $o_k=\sigma_k$ if $\sigma_k \in \Gamma$, ii) $o_k=(\sigma_k,0,\infty)$ if $k=j+nm+r$ for some $0\leq n<\ell$, and iii) $o_k=(\sigma_k,-\infty,\infty)$ otherwise.
\end{itemize} 
Let $\exec'(\ell)=\execlb{j+\ell m}$  be the  path that is equivalent to $\exec$ and $\inseq(\exec'(\ell))=\inbetaseq(\ell).$
%such that $\absexec=\abst(\exec'(\ell))$ and $\inseq(\exec'(\ell))=\inbetaseq(\ell)$. $o_k$ is defined as above. 

Let $\exec(\ell)|| k$ and $\exec'(\ell) || k$ be the suffixes of executions $\exec(\ell)$ and $\exec'(\ell)$ starting from state $q_k.$ Using backward induction, we can easily show that for each $x_0,$
$\pathprob{x_0,\exec(\ell)||k}, \pathprob{x_0,\exec'(\ell)||k} $ are non-zero and that 
$$ \pathprob{x_0,\exec(\ell)||k} = \euler^{ {\#(k) d_{j+r}\epsilon} } \pathprob{x_0,\exec'(\ell)||k}$$
where $\#(k)$ is the number of indices $k_1$ such that $k\leq k_1 < j+m\ell-1$ and $k_1=j+nm+r$ for some $0\leq n<\ell.$
Thus,  $$\pathprob{\epsilon,\exec(\ell)}=\euler^{ {\ell d_{j+r}\epsilon} } \pathprob{\epsilon,\exec'(\ell)}.$$

Now, $\ell$ is arbitrary and hence for every $d>0$, there is an $\ell$ such that $\pathprob{\epsilon,\exec(\ell)}>\euler^{ {d\epsilon}} \pathprob{\epsilon,\exec'(\ell)}.$
Hence $\cA$ is not differentially private. 
\end{proof}

\subsection*{\textbf{{\Violatingp}s implies no privacy}}

\begin{lemma}
\label{lem:violatingpnec}
A  {\dipa} $\cA $ is 
not differentially private if  it has a reachable {\violatingp}.
\end{lemma}
\begin{proof}
Thanks to Lemma~\ref{lem:criticalcnec},  Lemma~\ref{lem:violatingcnec} and Lemma~\ref{lem:criticalpnec}, we can assume $\cA$ does not have  {\criticalcycle}s, {\violatingc}s or {\criticalpair}s. 
 We give the proof for one of the cases of a 
{violatingp}, where the path starts with a transition whose
guard is $\lttest$ and which lies on an {\lcycle} $C$ which is
followed by an {\agpath} ending in a transition with guard
$\getest$ and whose output is $\svar.$ (The proofs for other
cases of the privacy violating path are similar and are leftout.)
Since $\cA$ is well-formed, the cycle $C$ does not have an
assignment transition. 

Fix $\ell>0.$ Consider an abstract path $\absexec(\ell)=\eabsexecl{n}$ of length $n$ from the initial state $\qinit$ such that $\absexec(\ell)$ contains the cycle $C$ repeated $\ell$ times, and upon exiting the cycle continues onto the
{\agpath} $p$ such that the last transition of the {\agpath} has guard $\getest$ and outputs $\svar.$ Fix a transition of $C$ with guard $\lttest,$ and let $k_1,k_2,\ldots,k_\ell$ be the indices where this transition occurs in $\absexec(\ell).$
Let $\parf(q_k)=(d_k,\mu_k).$
Next, we construct two input sequences $\inalphaseq(\ell)=a_0\cdots a_n$ and $\inbetaseq(\ell)=b_0\cdots b_n$ of length $n$ as follows. 
If the $k$th transition of $\absexec(\ell)$ is a non-input transition then $a_k=b_k=\tau.$
If  $k\in \set{k_1,k_2,\ldots,k_\ell}$  then $a_k=-\mu_k$ and $b_k=-\mu_k+1.$
For all other $k$s, $a_k=b_k=-\mu_k.$ Let $\exec(\ell)=\execl{j+\ell m}$ be the path such that
\begin{itemize}
\item $\absexec=\abst(\exec(\ell)),$ 
\item $\inseq(\exec(\ell))=\inalphaseq(\ell),$  and  
\item for all $k$, i) $o_k=\sigma_k$ if $\sigma_k \in \Gamma$, ii) $o_k=(\sigma_k,-\infty,0)$ if $k=n$, and iii) $o_k=(\sigma_k,-\infty,\infty)$ otherwise. 

\end{itemize}
Let $\exec'(\ell)=\execlb{j+\ell m}$ ibe the  path that is equivalent to $\exec$ and $\inseq(\exec'(\ell))=\inbetaseq(\ell).$
 
Please note that in $\exec(\ell),\exec'(\ell),$ the last output is a non-positive number. As the path $p$ is also an {\agpath}, this implies that stored value of $\rvar$ during the $\ell$ executions of $C$ is also a non-positive number. 
Combined with the fact that $C$ is an {\lcycle} and the construction of $\exec(\ell),\exec'(\ell)$, it can be shown that   $$\pathprob{\epsilon, \exec(\ell)}=\euler^{ {\ell d_{k_1}\epsilon}} \pathprob{\epsilon,\exec'(\ell)}.$$ 
As in the case of {\violatingc} (See Lemma~\ref{lem:violatingcnec}),  we can conclude that $\cA$ is not differentially private. 
\end{proof}

\section{Sufficiency of well-formedness}
\label{app:sufficient}
% !TEX root =  main.tex

We shall now show that if the {\dipa} $\cA$ is well-formed then $\cA$ is differentially private, thus establishing the \lq\lq if\rq\rq\ part of Theorem~\ref{thm:main}. Please note that it suffices to prove Lemma~\ref{lem:sufficiency}.
In order to manage complexity, we shall first prove the Lemma for the case that $\cA$ outputs only elements of the discrete set $\outalph$ (See Lemma~\ref{lem:main}). Then we shall tackle the case of 
all outputs (See Lemma~\ref{lem:main2}). 
 Please note that we shall use the notions of path suffixes and abstract paths introduced in Appendix~\ref{app:auxdefintiions}.

Before we proceed, we need a technical lemma. 

\begin{lemma}
\label{lem:integralineq}
Let $f$ and $g_i$ for $i=1,\ldots,k$  be non-negative functions from $\Reals$ to $\Reals$, i.e.,
$f(y),g_i(y)\geq 0$ for all $i,y.$ For $i=1,\ldots,k$, let $\theta_i\in
[-1,1]$.
Let $x_0,x_1\in\Reals\cup \set{\infty,-\infty}$, be such that $x_0<x_1$.
Then, the following inequalities
are satisfied for all $k\geq 0$. The empty products (the case when $k=0$)
%$\prod^k_{i=1}$, 
in these inequalities are taken be $1.$

%\begin{enumerate}

$$\begin{array}{ll}
1. & \int^{x_{1}}_{x_{0}}\:f(x) 
  \prod^k_{i=1}\int^\infty_x\:g_i(y-\theta_i)dy dx\;\geq \\
  & \hspace*{1cm} \int^{x_{1}+1}_{x_{0}+1}\:f(x-1) 
  \prod^k_{i=1}\int^\infty_x\:g_i(y)dy\: dx \\
  \\
2. & \int^{x_{1}}_{x_{0}}\:f(x) 
  \prod^k_{i=1}\int^x_{-\infty}\:g_i(y-\theta_i)dy dx\;\geq \\
 & \hspace*{1cm}  \int^{x_{1}-1}_{x_{0}-1}\:f(x+1) 
  \prod^k_{i=1}\int^x_{-\infty}\:g_i(y)dy\: dx
  \end{array}$$
%  
%\item $\displaystyle{\int^{x_{1}}_{x_{0}}\:f(x) 
%  \prod^k_{i=1}\int^x_{-\infty}\:g_i(y-\theta_i)dy dx\;\geq \int^{x_{1}-1}_{x_{0}-1}\:f(x+1) 
%  \prod^k_{i=1}\int^x_{-\infty}\:g_i(y)dy\: dx}$
%\end{enumerate}
\end{lemma}
\begin{proof}
We prove the inequality (1) as follows. For each $i=1,\ldots,k$, by substituting
$z\:=y-\theta_i$, we get $\int_x^\infty  g_i(y-\theta_i)dy
\:=\int^\infty_{x-\theta_i}g_i(z)dz.$ Since $\theta_i\in [-1,1]$ and
$g_i$ is a positive function, we get
$\int^\infty_{x-\theta_i}g_i(z)dz\:\geq \int^\infty_{x+1}g_i(z)dz.$ By
rewriting the left hand side of the inequality (1) as specified above
and by substituting, $u = x+1$, we get the right hand side of the
inequality (1) where the outer integral is over the variable $u.$ By
replacing $u$ by $x$ and $z$ by $y$, we get the right hand side of the
inequality. 

We prove the inequality (2) as follows. As before, for each
$i=1,\ldots,k$, we rewrite the integral $\int^x_{-\infty}
g_i(y-\theta_i)dy$ as $\int^{x-\theta_i}_{-\infty}
g_i(z)dz$ and then observe that this is $\geq \int^{x-1}_{-\infty}
g_i(z)dz.$ Substituting $u\:=x-1$, and then replacing $u$ by $x$ later,
we get the inequality (2).
\end{proof}

\subsection*{\textbf{{\dipautop} with Finite Outputs}}

\begin{lemma}
\label{lem:main}
Let $\cA=\defaut$ be a well-formed {\dipa} with finite outputs. Let  
%$\exec=\defexec$  be a path of length $n>0,$ such that the first transition 
$\exec$ be a path of length $n>0$ such that the initial transition (i.e. the $0$th transition), $t_0$,
of $\exec$ is 
an assignment transition. 
%Further assume that the
%suffix of $\exec$ starting with $t_1$, i.e., the suffix $p||1$, has no
%assignment transitions whose condition is $\true.$
Let $c_0$ be the guard of $t_0.$
Let $\exec'$ be a path that is equivalent  to $\exec$ such that $\inseq(\exec')$
is a neighbor of $\inseq(\exec).$ 
Then the following properties
hold for all $x_0\in \Reals.$ 
\begin{enumerate}
\item  If the guard $c_0$ is $\getest$,  and  the first
  cycle transition in $\exec$ is a {\gcycle} transition and no assignment
  transition 
with guard   $\lttest$ appears before it, then 
$$\pathprob{x_0,\exec'}\geq \eulerv{-\weight{\rho}\epsilon} \pathprob{x_0+1,\exec}.$$
%$$\pathprob{x_0,\exec'}\geq \eulerv{-\weight{\rho}\epsilon}\int^{\infty}_{x_{0}+1}
%d_0\frac{\epsilon}{2}\eulerv{-d_0\epsilon|x-a_{0}|}G_{p||1}(x) 
%dx.$$
\item  If the guard $c_0$ is $\getest$ and one of the
  following holds:  (a) $\exec$ has no cycle transitions,  (b)  the first
  cycle transition in $\exec$ is a  {\gcycle} transition and an assignment transition with guard 
  $\lttest$ appears before it,  (c) the first cycle transition
  in $\exec$ is an {\lcycle} transition,
  then
$$\pathprob{x_0,\exec'}\geq \eulerv{-\weight{\rho}\epsilon} \pathprob{x_0-1,\exec}. 
$$
%$$\pathprob{x_0,\exec'}\geq \eulerv{-\weight{\rho}\epsilon}\int^{\infty}_{x_{0}-1}
%d_0\frac{\epsilon}{2}\eulerv{-d_0\epsilon|x-a_{0}|}G_{p||1}(x) 
%dx.$$ 

\item  If the guard $c_0$ is $\lttest$  and the first cycle
  transition in $\exec$ is a 
{\lcycle}   transition  and  no assignment transition with guard 
  $\getest$ appears before it, then 
$$\pathprob{x_0,\exec'}\geq \eulerv{-\weight{\rho}\epsilon}\pathprob{x_0-1,\exec}.$$
%$$\pathprob{x_0,\exec'}\geq \eulerv{-\weight{\rho}\epsilon}\int^{x_{0}-1}_{-\infty}
%d_0\frac{\epsilon}{2}\eulerv{-d_0\epsilon\mathbin{|}x-a_{0}\mathbin{|}}G_{p%||1}(x) 
%dx.$$
\item  If the guard $c_0$ is $\lttest$ and one of the
  following holds: (a) $\exec$ has no cycle transitions, (b) the first
  cycle transition in $\exec$  is  a 
{\lcycle}   transition  and an assignment transition with guard 
  $\getest$ appears before it, (c) the first cycle
  transition in $\exec$ is a {\gcycle} transition, then 
$$\pathprob{x_0,\exec'}\geq \eulerv{-\weight{\rho}\epsilon}\pathprob{x_0+1,\exec}.$$
%$$\pathprob{x_0,\exec'}\geq \eulerv{-\weight{\rho}\epsilon}\int^{x_{0}+1}_{-\infty}
%d_0\frac{\epsilon}{2}\eulerv{-d_0\epsilon\mathbin{|}x-a_{0}\mathbin{|}}G_{p%||1}(x) 
%dx.$$

\item If the guard $c_0$ is $\true$, then $\pathprob{x_0,\exec'}\geq
  \eulerv{-\weight{\rho}\epsilon} \pathprob{x_0,\exec}.$
\end{enumerate}
\end{lemma}
\begin{proof}
Let $\exec=\execl{n}$ and $\exec'=\execlb{n}.$ Let $t_0,\ldots,t_{n-1}$ be the transitions of $\exec$ and let $c_0,\ldots,c_{n-1}$ be their respective guards. 
For each $k\leq n,$ let $d_k,\mu_k$ be such that $\parf(q_k)=(d_k,\mu_k).$ Recall that, for any $k,$  $\exec||k$ denotes the 
suffix of $\exec$ starting from $q_k.$ We assume that there are no cycle transitions that are assignments. This is because if there is a cycle with an assignment then the guards on all other transitions must be $\true.$ Hence, we can never exit the cycle. Further, it is easy to see that this cycle has the same \lq\lq behavior\rq\rq\ in both $\exec$
and $\exec'.$

%\red{Set up exec tis ais,bis suffix Recall that, for any $i\geq 0$,  $\exec||i$ denotes the 
%suffix of $\exec$ starting from $t_i.$,d's and mu's}

For each $k$, such that $0\leq k<n$, let $g_k,g'_k,\theta_k$ be functions of a
single variable given by
$$g_k(y)\:=\begin{cases}
                     \frac{d_k\epsilon}{2}\eulerv{-d_k\epsilon \abs{y-a_k-\mu_k}} &{t_i} \mbox{ is an input transition}\\
                     \frac{d_k\epsilon}{2}\eulerv{-d_k\epsilon \abs{y-\mu_k}} & \mbox{otherwise}
                   \end{cases}, $$
$$g'_k(y)\:=\begin{cases}
                     \frac{d_k\epsilon}{2}\eulerv{-d_k\epsilon \abs{y-b_k-\mu_k}} &{t_i} \mbox{ is an input transition}\\
                     \frac{d_k\epsilon}{2}\eulerv{-d_k\epsilon \abs{y-\mu_k}} & \mbox{otherwise}
                   \end{cases}   $$ and
$$\theta_k\:= \begin{cases}
                    b_k-a_k &{t_i} \mbox{ is an input transition}\\
                    0  & \mbox{otherwise}.
                    \end{cases}
$$ Observe that, for each $k\geq 0$, $g'_k(y)\:=g_k(y-\theta_k).$ Since $\card{\theta_k}\leq 1$, we see that
$g'_k(y)\geq \eulerv{-d_k\epsilon}g_k(y)$, for all $y\in \Reals.$

We prove the lemma by induction on the number of assignment
transitions in $\exec.$

\paragraph*{\textbf{Base Case}} In the base case, $\exec$ has one assignment
transition which is $t_0.$ Let
$S_1$ and $S_2$  be the sets of $k>0$ such that $c_k$ is $\getest$ and $c_k$ is $\lttest$, respectively. Now, assume the condition of statement (1) of the Lemma is
satisfied. Observe that $S_1$ includes
all {\gcycle} transitions whose guard  is $\getest.$ Observe that, since
$\cA$ is well-formed, for all
$k\in S_2$, $t_k$ does not lie on a cycle and hence is a {\criticaltransition}. Similarly $t_0$ is also a {\criticaltransition}. 
Now, we see 
that $$\pathprob{x_0,\exec'}\:=\:\int^\infty_{x_{0}}f(x)\prod_{k\in 
    S_{1}}\int^\infty_xg'_k(y) dy\:  dx$$ where 
$\displaystyle{f(x)\:=\:g'_0(x)\prod_{k\in 
      S_{2}}\int^x_{-\infty} g'_k(y) dy}.$
Now, substituting $g'_k(y)=g_k(y-\theta_k)$ (for $k\in S_1$) in the
above equation and using inequality
  (1) of Lemma \ref{lem:integralineq}, we see that
  $$\displaystyle{\pathprob{x_0,\exec'}\geq\int^\infty_{x_{0}+1}f(x-1)\prod_{k\in
      S_{1}}\int^\infty_xg_k(y)dy\:dx}.$$ 
      Observe that
    $$f(x-1)\:=g_0(x-(1+\theta_0))\prod_{k\in
      S_{2}}\int^{x-1}_{-\infty}g_k(y-\theta_k) dy.$$ 
      Now, by
    introducing a new variable $z$ such that $z\:=y+1$, we see that
    $$\int^{x-1}_{-\infty}g_k(y-\theta_k) dy\:=
    \int^{x}_{-\infty}g_k(z-(1+\theta_k)) dz.$$
From this, it is easy to see that  $$f(x-1)\geq
\eulerv{-2(d_0+\sum_{k\in
    S_{2}}d_k)\epsilon}g_0(x)\prod_{k\in
  S_{2}}\int^{x}_{-\infty}g_k(y) dy.$$
Observe that $\weight{\rho}\:\geq 2(d_0+\sum_{k\in S_{2}}d_k).$
Putting all the above observations together, we get 
\begin{dmath*}
\pathprob{x_0,\exec'}\geq \eulerv{-\weight{\rho}\epsilon} \int^\infty_{x_{0}+1
}g_0(x)\prod_{k\in S_2}\int^x_{-\infty}g_k(y) dy\:\prod_{k\in S_{1}}
\int^\infty_xg_k(y) dy.
\end{dmath*}
Observe that the right hand side of the above inequality is
$\eulerv{-\weight{\rho}\epsilon}\pathprob{x_0+1,\exec}$. 
Property (1) of the lemma follows for the base case from this observation.  

Now, we prove the base case for property (2). Assume the condition of
(2a) is satisfied, i.e., there are no cycle transitions in $\exec.$ Now, we see that 
\begin{dmath*}\pathprob{x_0,\exec'}\:=\:\int^\infty_{x_{0}}g'_0(x)\prod_{k\in 
    S_{1}}\int^\infty_xg'_k(y) dy \prod_{k\in 
    S_{2}}\int^x_{-\infty}g'_k(z) dz\:dx.\end{dmath*}
By introducing new variables $u,v,w$ such $u\:=x-1,\:v=y-1,\:w=z-1$,
we get  
\begin{dmath*}
\pathprob{x_0,\exec'}\:=\:\int^\infty_{x_{0}-1}g'_0(u+1)\prod_{k\in 
    S_{1}}\int^\infty_ug'_k(v+1) dv \prod_{k\in 
    S_{2}}\int^u_{-\infty}g'_k(w+1) dw\:du.
    \end{dmath*}
Observing that, for each $k\geq 0$, $g'_k(u+1)\geq
\eulerv{-2d_k\epsilon}g_k(u)$ and $t_k$ is a {\criticaltransition}, we get the inequality of property (2). 

Now observe that condition of (2b) can not be satisfied as $t_0$ is the only assignment transition in
$\exec$. Now, assume the condition of (2c) is satisfied.  
Now, observe that, for all $k\in S_1$, $t_k$ is a {\criticaltransition.} As before, we see 
that $$\pathprob{x_0,\exec'}\:=\:\int^\infty_{x_{0}}f(x)\prod_{k\in S_{2}}\int^x_{-\infty} g'_k(y) dy\:  dx$$ where 
$\displaystyle{f(x)\:=\:g'_0(x)\prod_{k\in 
      S_{1}}\int^\infty_x g'_k(y) dy}.$  Now, using inequality
  (2) of Lemma \ref{lem:integralineq}, we see that
  $$\displaystyle{\pathprob{x_0,\exec'}\geq\int^\infty_{x_{0}-1}f(x+1)\prod_{k\in
      S_{2}}\int^x_{-\infty} g_k(y)dy\:dx}.$$ Now, observe that
    $$f(x+1)\:=g_0(x-(\theta_0-1))\prod_{k\in S_{1}}\int_{x+1}^\infty
    g_k(y-\theta_k)dy.$$ Introducing a new variable $z$ and
    setting $z=y-1$, we see that
  $$f(x+1)\:=g_0(x-(\theta_0-1))\prod_{k\in S_{1}}\int_x^\infty g_k(z-(\theta_k-1))dz$$
      and
$$f(x+1)\geq \eulerv{-2(d_0+\sum_{k\in
    S_{1}}d_k)\epsilon}g_0(x)\prod_{k\in S_{1}}\int_x^\infty
g_k(z)dz.$$ From this and the above inequality, it is easily seen that 
$$\pathprob{x_0,\exec'}\geq \eulerv{-2(d_0+\sum_{k\in
    S_{2}}d_k)\epsilon} \pathprob{x_0-1,\exec}.$$ From this we see that the
inequality of property (2) holds.

The proof for the base case of Properties (3) and (4) is symmetric to
those of properties (1) and (2) and is left out. To prove property (5)
for the base case, we see that the proof is similar to those of
properties (1) and (3) depending on whether {\gcycle} or {\lcycle}
transitions appear. There are two minor differences. The first difference is that 
if the first transition is a non-input transition
then $\theta_0=0$ and hence it only incurs a cost of $d_0$ and not $2 d_0.$ 
The second difference is that the lower limit of the
outer integral will be $-\infty$ in the former case, while the upper
limit of the outer integral being $\infty$ in the latter case. In
either case, it is straightforward to see that property (5) holds.

\paragraph*{\textbf{Inductive Step}} Now, we prove the inductive step as follows. Assume that all the
properties hold when $\exec$ has $\ell>0$ assignments. Now, consider the
case when $\exec$ has $\ell+1$ assignments. Let $t_i$, for $i>0$, be the
second assignment transition in $\exec.$ 
Let $S_1$ (resp., $S_2$) be the set of $k$, $0<k<i$, such that $c_k$ is $\getest$ (resp., $\lttest$).
%Let $S\:=\set{k\::0< k<i,\:c_k\neq \true}.$ 

Consider the case when $c_0$
is $\getest.$
Now, we consider two sub-cases. We first consider the sub-case when there is no cycle transitions before $t_i.$
We have $\pathprob{x_0,\exec'}\:=\int^\infty_{x_{0}}f'(x) 
\pathprob{\exec'||i,x}dx$ where 
$$f'(x) \:=g'_0(x)\prod_{k\in S_{1}}\int_x^{\infty}g'_k(y)dy\prod_{k\in S_{2}}\int^x_{-\infty}g'_k(y)dy.$$
 Applying the
inductive hypothesis for the suffix $\exec||i$, we get an inequality
involving $\pathprob{\exec'||i,x}$ and  $\pathprob{x+1,\exec||i}$, or $\pathprob{\exec'||i,x-1}$, or $\pathprob{x,\exec||i}$, based on which of  the five properties
of the lemma are satisfied by $\exec||i.$ Suppose the condition of property (1) is
satisfied by $\exec||i$, by using the inductive hypothesis,  we get
$\pathprob{x_0,\exec'}\geq \int^{\infty}_{x_{0}} f'(x) h(x) dx$, 
where $h(x)\:= \eulerv{-2 \weight{\exec||i}\epsilon}\pathprob{x+1,\exec||i} .$ 
Now, by taking $f(x)\:=f'(x)h(x)$, using inequality (1) of Lemma
\ref{lem:integralineq} and %\red
{by taking
$k=0$ in that inequality, we get property (1) for the path $\exec$
 using the same simplification/reasoning used in the base case} and by
 observing that 
 $$\begin{array}{l}\displaystyle{\pathprob{x_0+1,\exec}\:=\int^\infty_{x_{0}+1} g_0(x) \prod_{k\in
     S_1}\int_x^{\infty}g_k(y)dy}\\
     \hspace*{3cm}\displaystyle{\prod_{k\in S_{2}}\int^x_{-\infty}g_k(y)dy
   \pathprob{x,\exec||i}dx}.
   \end{array}$$
We can similarly prove the inductive step when the suffix $\exec||i$
satisfies the other properties (i.e., 2 through 5) of the lemma.

 Now consider the sub-case when a cycle transition appears before
$t_i.$ Assume  that
the cycle transitions are {\gcycle} transitions. If $c_i$ is also
$\getest$, then the suffix $\exec||i$ can satisfy any of the
conditions of the first two properties of the lemma; In this
situation, let 
$f(x)\:=f'(x)h(x)$ where $f'(x)\:=g'_0(x)\prod_{k\in S_{2}}\int^x_{-\infty}g'_k(y)dy$
and  $h(x) \:=\eulerv{-2 \weight{\exec||i}\epsilon}\pathprob{x+1,\exec||i}.$
Observe that, if $\exec||i$ satisfies the condition of property (1)
then $h(x)$ is the RHS of the inequality, we get,  by
applying the inductive hypothesis to $\exec||i.$ If $\exec||i$ satisfies the condition of property (2) of the lemma
then, by applying the inductive hypothesis to $\exec||i$, we get 
$\pathprob{\exec'||i,x}\geq \eulerv{-2 \weight{\exec||i}\epsilon}\pathprob{x-1,\exec||i}.$
Since, $\pathprob{x-1,\exec||i}\geq \pathprob{x+1,\exec||i}$, we see that 
$\pathprob{\exec'||i,x}\geq \eulerv{-2 \weight{\exec||i}\epsilon}\pathprob{x+1,\exec||i}.$
Now, we have $\pathprob{x_0,\exec'}\:\geq \int^\infty_{x_{0}}f'(x)h(x)\prod_{k\in
  S_{1}}\int^\infty_{x}g'_k(z) dz dx.$
Applying the  inequality (1) of
 Lemma \ref{lem:integralineq}, we get the desired result for the
 inductive step.
On the other hand, if $c_i$ is $\lttest$ then the suffix
$\exec||i$ can not satisfy the condition of property (3) of the lemma due
to well-formedness of $\cA$; however it can satisfy the condition of
property (4). In this sub-case also, we can
get the result for the induction case as above by using the inductive
hypothesis for $\exec||i$ and using similar reasoning as in the base
case and applying the first inequality of  Lemma
\ref{lem:integralineq}. 

Now consider the situation where  the cycle transitions appearing
before $t_i$ are  {\lcycle} transitions. Now, we apply
inequality (2) of Lemma \ref{lem:integralineq} to prove that property
(2) of the lemma is satisfied by $\exec.$ To do this, we define $f(x)\:=f'(x)h(x)$
where $f'(x)\:=g'_0(x)\prod_{k\in S_{1}}\int_x^{\infty}g'_k(y)dy$ and
$h(x) \:= \eulerv{-2 \weight{\exec||i}\epsilon}\pathprob{x-1,\exec||i}.$ Next, applying
the induction hypothesis to $\exec||i$, we show that 
$$\pathprob{x_0,\exec'}\geq \int^{\infty}_{x_{0}} f'(x)h(x) \prod_{k\in
  S_{2}}\int^x_{-\infty}g'_k(y)dy dx.$$
Since $\cA$ is well-formed, $\exec||i$ cannot satisfy the condition
of property (1) of the lemma.
If $\exec||i$ satisfies the condition
of property (2) or that of property (3) then, the above inequality
follows directly from the induction hypothesis;
If  $\exec||i$ satisfies the condition
of property (4), then the above inequality follows from the induction
hypothesis and the observation that $\pathprob{x+1,\exec||i}\geq \pathprob{x-1,\exec||i}$;
 If $\exec||i$ satisfies the condition
of property (5)  then the above inequality follows from the induction
hypothesis and the observation that $\pathprob{x,\exec||i}=\pathprob{x-1,\exec||i}$ as
$\pathprob{x,\exec||i}$ is independent of $x.$
Rewriting the above inequality, we get
$$\pathprob{x_0,\exec'}\geq \int^{\infty}_{x_{0}} f'(x)h(x) \prod_{k\in S_{2}}\int^x_{-\infty}g_k(y-\theta_k)dy.$$
Now, using the inequality (2) of
Lemma  \ref{lem:integralineq}, and using simplifications and reasoning
as in the base cases, we see that property (2) of the lemma
is satisfied by $\exec.$

The proof for the inductive step for the case when $c_0$ is
$\lttest$ is symmetric. For the case, when $c_0$ is $\true$, the
proof will be on the same lines excepting that if $t_0$ is a non-input transition then it incurs a cost of $d_0$ only and  the limits of the outer
integrals are $-\infty$ and $\infty.$
\end{proof}

\subsection*{\textbf{{\dipautop} with Finite and Infinite Outputs}}
% !TEX root =  main.tex

%\newtheorem*{lem:main2}{Lemma ~\ref{lem:main2}}

%\red{Fix Lemma in main and reference}
%\subsection{Sufficiency proof}
%\red{Comment about $\svar'$, input}

We shall now show that if a {\dipa} $\cA$ is well-formed then it is differentially private. For simplicity, we will assume that all states are input states. The case when the $\cA$ includes non-input states can be dealt with similarly. 
Finally, we also assume that there are no transitions that output the value of $\svar'.$ In case there are transitions from $\svar',$ Lemma~\ref{lem:main2} can be proved by appealing to the composition theorem of differential privacy (See Theorem 3.14 of~\cite{DR14}.)

The following proposition follows directly from the definition of 
well-formed {\dipautop}.
\begin{proposition}
\label{prop:abs}
Let $\cA$ be a well-formed {\dipa} and $\exec$ be a
path of $\cA$ starting from a reachable state. Then $\exec$ satisfies the
following properties.
\begin{itemize}
 \item If $\exec$ starts with an assignment transition
$t_0$ and has no further assignment transitions, and has a {\gcycle}
or an {\lcycle} transition then the output of $t_0$ is from $\outalph$.
\item If
$\exec$ has no assignment transitions and has a {\gcycle} (resp.,
{\lcycle}) transition then the output of every transition in $\exec$, with
guard $\lttest$ (resp., $\getest$),  is from $\outalph$.

\item If $\exec$ starts with an {\lcycle} (resp., {\gcycle}) transition and 
is an {\agpath} (resp., {\alpath}) then the output of every transition,
with guard $\getest$ (resp., $\lttest$), is from $\outalph.$ 
\item If $\exec$ is an {\agpath} (resp., {\alpath}) ending with a {\gcycle}
  (resp., {\lcycle}) then the output of every transition,
with guard $\lttest$ (resp., $\getest$) , is from $\outalph.$ 
\end{itemize}
\end{proposition}

Please note that Lemma~\ref{lem:sufficiency} is an immediate consequence of the following lemma. 
\begin{lemma}
\label{lem:main2}
Let $\cA=\defaut$ be a well-formed {\dipa} and 
$\exec$ be a path of length $n>0$   Let $t_0$ be the initial transition, i.e., the $0$th transition of $\exec$,  $c_0$ be its guard and $o_0$ be its output.  
Let $t_0$ be an assignment transition, and let $\exec'$ be a path that is equivalent to $\exec$ such that $\inseq(\exec')$
is a neighbor of $\inseq(\exec).$ 
%$p= (t_0,...,t_{n-1})$ be a path of length $n>0$ in $\cA$, where $t_i=(q_i,a_i,o_i,f(a_i,\epsilon),f'(a_i,\epsilon)c_i,A_i,q_{i+1})$ and $o_i\neq \svar'$, for $0\leq i<n$ and $t_0$ is an assignment transition. 
%Let $\exec'$ be a path that is equivalent  to $\exec$ such that
%$inseq(p')=(b_0,...,b_{n-1})$ is a neighbour of $inseq(p).$ Let
%$\sigma=(\sigma_0,...,\sigma_{n-1})$ be any sequence in $O_2^n$ such
%that $outseq(p)$ is consistent with $\sigma$. 
Then the following properties
hold for all $x_0\in \Reals.$ 
%In these properties $p||1,\sigma||1$ denote the
%suffixes of $p,\sigma$ starting from $t_1,\sigma_1,$ respectively. 
\begin{enumerate}
\item  If the guard $c_0$ is $\getest$,  and  the first
  cycle transition in $\exec$ is a {\gcycle} transition and no assignment
  transition 
with guard   $\lttest$ appears before it, $o_0\in \outalph$ and 
$$\pathprob{x_0,\exec'}\geq \eulerv{-\weight{\exec}\epsilon}
\pathprob{x_0+1,\exec}.$$

%$$\pathprob{x_0,\exec'}\geq \eulerv{-\weight{\exec}\epsilon}\int^{b'_0+1}_{\max(a'_0,x_{0})+1}
%d_{q_{0}}\frac{\epsilon}{2}\eulerv{-d_{q_{0}}\epsilon|x-a_{0}|}G_{p||1}(x,\sigma||1) 
%dx.$$
\item  If the guard $c_0$ is $\getest$ and either,  (a)
  $\exec$ has no cycle transitions; or  (b) the first
  cycle transition in $\exec$ is a  {\gcycle} transition and an assignment
  transition with guard   $\lttest$ appears before it; or (c)
  the first cycle transition  in $\exec$ is an {\lcycle} transition,
  then 
$$\pathprob{x_0,\exec'}\geq \eulerv{-\weight{\exec}\epsilon}
\pathprob{x_0,\exec}.$$
Furthermore, if  the output of every transition, whose guard is
$\getest$, is from $\outalph$,  until the first assignment
transition whose guard is $\lttest$  or until the end of $\exec$,
then  
$$\pathprob{x_0,\exec'}\geq \eulerv{-\weight{\exec}\epsilon}
\pathprob{x_0-1,\exec}.$$

\item  If the guard $c_0$ is $\svar < x$  and the first cycle
  transition in $\exec$ is a 
{\lcycle}   transition  and  no assignment transition with guard 
  $\getest$ appears before it, then $o_0\in \outalph$ and
$$\pathprob{x_0,\exec'}\geq 
\eulerv{-\weight{\exec}\epsilon}\pathprob{x_0-1,\exec}.$$

\item  If the guard $c_0$ is $\svar < x$, and either
(a) If $\exec$ has no cycle transitions; or
(b) The first
  cycle transition in $\exec$  is  an 
{\lcycle}   transition  and an assignment transition with guard 
  $\getest$ appears before it; or (c) the first cycle
  transition in $\exec$ is a {\gcycle} transition, then 
$$\pathprob{x_0,\exec'}\geq 
\eulerv{-\weight{\exec}\epsilon}\pathprob{x_0,\exec}.$$
Furthermore, if  the output of every transition, whose guard is
$\lttest$, is from $\outalph$,  until the first assignment
transition whose guard is $\getest$  or until the end of
$\exec$, then  
$$\pathprob{x_0,\exec'}\geq \eulerv{-\weight{\exec}\epsilon}
\pathprob{x_0+1,\exec}.$$

\item If the guard $c_0$ is $\true$, then $\pathprob{x_0,\exec'}\geq
  \eulerv{-\weight{\exec}\epsilon} \pathprob{x_0,\exec}.$
\end{enumerate}
\end{lemma}

\begin{proof}
Let $\exec=\execl{n}$ and $\exec'=\execlb{n}.$ Let $t_0,\ldots,t_{n-1}$ be the transitions of $\exec$ and let $c_0,\ldots,c_{n-1}$ be their respective guards. 
For each $k\leq n,$ let $d_k,\mu_k$ be such that $\parf(q_k)=(d_k,\mu_k).$ Recall that, for any $k,$  $\exec||k$ denotes the 
suffix of $\exec$ starting from $q_k.$ Once again, we  assume that there are no cycle transitions that are assignments.

%For succinctness, throughout the proof, when ever we refer
%to the functions $G_{p'},G_{p}$, we only specify the first argument to
%these functions as the second argument is understood from the context:
%for example, $G_p(x), \pathprob{x,\exec||i}$ are references $G_p(x,\sigma)$ and
%$G_{{\exec||i}}(x,\sigma||i)$, respectively.

We show, how the proof of Lemma \ref{lem:main} can be modified to
prove  this Lemma.
First, observe that properties (1), (3) and (5) of the Lemma are
identical to
the corresponding properties of the Lemma \ref{lem:main}. When $o_i\in
\outalph$, for all $i, 0\leq i<n$, the second parts of the properties (2) and
(4) subsume their first parts, and these two properties become
identical to properties (2) and (4) of the Lemma \ref{lem:main}, respectively.
For each $i, 0\leq i<n$, let 
$(u_i,v_i)$ be such that   $o_i=(\svar,u_i,v_i)$ if $o_i\notin\outalph$,
otherwise it is the interval $(-\infty,\infty).$
Let $g_k(y),g'_k(y)$ be the functions as defined in the proof of Lemma
\ref{lem:main}, and $\theta_k=b_k-a_k$ for $0\leq k<n.$

As before, we prove the Lemma by induction on the number of assignment
transitions in $\exec.$ In the base case, $\exec$ has one assignment
transition which is $t_0.$ Let
$S_1$ and $S_2$  be the sets of $k>0$ such that $c_k$ is $\svar
\geq x$ and $c_k$ is $\svar < x$, respectively. 

Now, assume the condition of (1) is
satisfied. Observe that $S_1$ includes
all {\gcycle} transitions whose guard  is $\getest.$ Let
$S_1'$ be the set of $k\in S_1$ such that $t_k$ is a {\gcycle}
transition and $S_1''= S_1 \setminus S_1'.$
 Observe that, using the fact that 
$\cA$ is  well-formed and using Proposition \ref{prop:abs}, we
see the following hold: (i) for all $k \in S_1'\cup S_2$, $o_k\in \outalph$;
(ii) $t_0$ is  a critical transition and $o_0\in \outalph$;
(iii) for all $k\in S_2\cup S_1''$, $t_k$ does not lie on a cycle and hence is a critical
transition. Note that, for any $k\in 
S_1''$, $o_k$ may be $\svar.$
Now, we see 
that $$\pathprob{x_0,\exec'}\:=\:\int^\infty_{x_{0}}f(x)\prod_{k\in 
    S_{1}'}\int^\infty_xg'_k(y) dy\:  dx$$ where 
$$\displaystyle{f(x)\:=\:g'_0(x)\prod_{k\in 
      S_{2}}\int^x_{-\infty} g'_k(y) dy \prod_{k\in
      S_1''}\int^{v_{k}}_{\max(x,u_{k})}g'_k(z) dz}.$$
Now, substituting $g'_k(y)=g_k(y-\theta_k)$ (for $k\in S_1$) in the
above equation and using inequality
  (1) of Lemma \ref{lem:integralineq}, we see that
  $$\displaystyle{\pathprob{x_0,\exec'}\geq\int^\infty_{x_{0}+1}f(x-1)\prod_{k\in
      S_{1}}\int^\infty_xg_k(y)dy\:dx}.$$ Now, using the same argument
  as in the proof of Lemma \ref{lem:main}, and observing that, for
  $k\in S_1''$, $\int^{v_{k}}_{\max(x-1,u_{k})}g'_k(z) dz\geq \int^{v_{k}}_{\max(x,u_{k})}g'_k(z) dz,$
it is easy to see that  
$$
\begin{array}{lcl}
f(x-1)&\geq&
\displaystyle{\eulerv{-2(d_{q_{0}}+\sum_{k\in
   S_{1}''\cup S_{2}}d_{q_{k}})\epsilon}g_0(x)}\\ 
  && \hspace*{0.6cm}
  \displaystyle{ \prod_{k\in
  S_{2}}\int^{x}_{-\infty}g_k(y) dy} \\
  && \hspace*{1.2cm} \displaystyle{\prod_{k\in S_{1}''}\int^{v_{k}}_{\max(x,u_{k})}g'_k(z) dz.} 
  \end{array}$$ 
Putting all the above observations together, we see that property (1) holds.

Now, we prove the base case for property (2). Assume the condition of
(2a) is satisfied, i.e., there are no cycle transitions in $\exec.$ Now, we see that 
$$\begin{array}{lcl}
\pathprob{x_0,\exec'}&=&\displaystyle{\int^{v_{0}}_{\max(x_{0},u_{0})}g'_0(x)\prod_{k\in 
    S_{1}}\int^{v_{k}}_{\max(x,u_{k})}g'_k(y) dy}\\
    &&
     \hspace*{2.2cm}
     \displaystyle{ \prod_{k\in 
    S_{2}}\int_{u_{k}}^{\min(x,v_{k})}g'_k(z) dz\:dx.}\\
\end{array} $$   
    It is fairly straightforward to see that 
$\pathprob{x_0,\exec'}\geq \eulerv{-\weight{\exec}\epsilon}\pathprob{x_0,\exec}$ since
$g'_k(y)\geq \eulerv{-d_{q_{k}}\epsilon}g_k(y)$, for all $y\in
\Reals$, $0\leq k<n.$ From this, we see that the first part of
property(2) holds. To see that the second part of property (2) holds,
assume that $o_0\in \outalph$, and for all $k\in S_{1}, o_k\in \outalph$. This means
that
\begin{dmath*}
\pathprob{x_0,\exec'}\:=\:\int^\infty_{x_{0}}g'_0(x)\prod_{k\in 
    S_{1}}\int^\infty_xg'_k(y) dy \prod_{k\in 
    S_{2}}\int_{u_{k}}^{\min(x,v_{k})}g'_k(z) dz\:dx.
    \end{dmath*}
Now introducing  new variables $w,y'$ and setting $w=x-1$ and $y'=y-1$, we see that
$$
\begin{array}{lcl}\pathprob{x_0,\exec'}&=&\displaystyle{\int^\infty_{x_{0}-1}g'_0(w+1)\prod_{k\in 
    S_{1}}\int^\infty_wg'_k(y'+1) dy' }\\
     &&\hspace*{1.4cm} \displaystyle{\prod_{k\in 
    S_{2}}\int_{u_{k}}^{\min(w+1,v_{k})}g'_k(z) dz\:dx.}
\end{array} $$
Now, observe that, for $k\in S_{2}$,
$\int_{u_{k}}^{\min(w+1,v_{k})}g'_k(z) dz\geq
\int_{u_{k}}^{\min(w,v_{k})}g'_k(z) dz$. Using this we get,
$$\begin{array}{lcl}
\pathprob{x_0,\exec'}&\geq&\displaystyle{\int^\infty_{x_{0}-1}g'_0(w+1)\prod_{k\in 
    S_{1}}\int^\infty_wg'_k(y'\blue{+1}) dy' }\\
   && \hspace*{1.6cm} \displaystyle{\prod_{k\in 
    S_{2}}\int_{u_{k}}^{\min(w,v_{k})}g'_k(z) dz\:dx.}
    \end{array}$$
Now, the second part of property (2), follows from the above
inequality and the reasoning employed earlier.

Now, condition of (2b) can not be satisfied as $t_0$ is the only assignment transition in
$\exec$. Now, assume the condition of  (2c) is satisfied.  Let $S_2'$ be
the set of all $k\in S_2$ such that $t_k$ is an {\lcycle} transition
and $S_2''= S_2\setminus S_2'.$ 
Now, using the fact that $\cA$ is  well-formed and using
Proposition \ref{prop:abs} we observe that the following 
hold: (i) for all $k\in S_1\cup S_2''$, $t_k$ is a critical
transition; (ii) $t_0$ is a critical transition and $o_0\in \outalph$; (iii)
for all $k\in S_1\cup S_2'$, $o_k\in \outalph.$
 Now, we see that 
that $$\pathprob{x_0,\exec'}\:=\:\int^\infty_{x_{0}}f(x)\prod_{k\in S'_{2}}\int^x_{-\infty} g'_k(y) dy\:  dx$$ where 
$$\displaystyle{f(x)\:=\:g'_0(x)\prod_{k\in 
      S_{1}}\int^\infty_x g'_k(y) dy \prod_{k\in 
      S''_{2}}\int_{u_{k}}^{\min(x,v_{k})} g'_k(y) dy}.$$  Now, using inequality
  (2) of Lemma \ref{lem:integralineq}, we see that
  $$\displaystyle{\pathprob{x_0,\exec'}\geq\int^\infty_{x_{0}-1}f(x+1)\prod_{k\in
      S'_{2}}\int^x_{-\infty} g_k(y)dy\:dx}.$$ Now, observe that
    $$\begin{array}{lcl}
    f(x+1)&=&\displaystyle{g_0(x-(\theta_0-1))\prod_{k\in S_{1}}\int_{x+1}^\infty
    g_k(y-\theta_k)dy}\\
    && \displaystyle{\prod_{k\in 
      S''_{2}}\int_{u_{k}}^{\min(x+1,v_{k})} g'_k(y) dy.}\end{array}$$ Introducing a new variable $z$ and
    setting $z=y-1$, we see that
  $$\begin{array}{lcl}
  f(x+1)&=&\displaystyle{g_0(x-(\theta_0-1))\prod_{k\in S_{1}}\int_x^\infty
  g_k(z-(\theta_k-1))dz} \\
  &&
\displaystyle{\prod_{k\in 
      S''_{2}}\int_{u_{k}}^{\min(x+1,v_{k})} g'_k(y) dy}\end{array}$$
      and
$$\begin{array}{lcl}f(x+1)&\geq& \displaystyle{\eulerv{-2(d_{q_{0}}+\sum_{k\in
    S_{1}\cup S''_2}d_{q_{k}})\epsilon}g_0(x)}\\
    && \hspace*{0.6cm}\displaystyle{\prod_{k\in S_{1}}\int_x^\infty
g_k(z)dz \prod_{k\in 
      S''_{2}}\int_{u_{k}}^{\min(x,v_{k})} g_k(y).}\end{array}$$ 
From this and the above inequality, it is easily seen that 
$\pathprob{x_0,\exec'}\geq \eulerv{-\weight{\exec}\epsilon} \pathprob{x_0-1,\exec}$. From this we see that the
inequalities of both parts of property (2) hold.

As before, the proof for the base case of Properties (3) and (4) is symmetric to
those of properties (1) and (2) and is left out. Property (5)
is proved as in the case of Lemma \ref{lem:main}.

Now, we prove the inductive step as follows. Assume that all the
properties hold when $\exec$ has $\ell>0$ assignments. Now, consider the
case when $\exec$ has $\ell+1$ assignments. Let $t_i$, for $i>0$, be the
second assignment transition in $\exec.$ 
Let $S_1$ (resp., $S_2$) be the set of $k$, $0<k<i$, such that $c_k$ is $\svar
\geq x$ (resp., $\lttest$).
%Let $S\:=\set{k\::0< k<i,\:c_k\neq \true}.$ 
Now, consider the case when $c_0$
is $\getest.$
Now, we consider two sub-cases. We first consider the sub-case when there is no cycle transitions before $t_i.$
We have $\pathprob{x_0,\exec'}\:=\int^{v_{0}}_{\max(x_{0},u_{0})}f'(x) 
\pathprob{x,\exec'||i} dx$ where 
$f'(x) \:=g'_0(x)\prod_{k\in S_{1}}\int^{v_{k}}_{\max(x,u_{k})}g'_k(y)dy\prod_{k\in S_{2}}\int_{u_{k}}^{\min(x,v_{k})}g'_k(y)dy.$ 
 Applying the
inductive hypothesis for the suffix ${\exec||i}$, we get an inequality
involving $\pathprob{x,\exec'||i}$ and  $\pathprob{x+1,\exec||i}$, or $\pathprob{x-1,{\exec||i}}$, or $\pathprob{x,\exec||i}$, based on which of  the five properties
of the Lemma are satisfied by ${\exec||i}.$ Suppose the condition of property (1) is
satisfied by ${\exec||i}$. Let $j\geq i$ be the smallest integer such that
$t_j$ is a {\gcycle} transition. Now, since $p|j$, the
prefix of $\exec$, is an {\agpath}, using the fact that 
 $\cA$ is  well-formed and using Proposition \ref{prop:abs}, it is easy
to see that $o_0\in \outalph$, and for all $k\in S_2$, $o_k\in \outalph.$
 By using the inductive hypothesis,  we get
$\pathprob{x_0,\exec'}\geq \int^{\infty}_{x_{0}} f'(x) h(x) dx$, 
where $h(x)\:= \eulerv{-\weight{{\exec||i}}\epsilon}\pathprob{x+1,\exec||i} .$ 
Because of the previous observation, we see that 
$f'(x) \:=g'_0(x)\prod_{k\in
  S_{1}}\int^{v_{k}}_{\max(x,u_{k})}g'_k(y)dy\prod_{k\in
  S_{2}}\int^x_{-\infty}g'_k(y)dy.$ 
Now, observe that, for each $k\in S_1$,
$\int^{v_{k}}_{\max(x-1,u_{k})}g'_k(y)dy\geq
\int^{v_{k}}_{\max(x,u_{k})}g'_k(y)dy.$
From this, using the reasoning employed in the base case, we see that 
$$\begin{array}{lcl}f'(x-1)&\geq&\displaystyle{ \eulerv{-2(d_{q_{0}}+\sum_{k\in S_{1}\cup
    S_{2}}d_{q_{k}})\epsilon} g_0(x)}\\
    && \hspace*{0.6cm} \displaystyle{\prod_{k\in
  S_{1}}\int^{v_{k}}_{\max(x,u_{k})}g_k(y)dy}\\
  && \hspace*{1.2cm}
  \displaystyle{\prod_{k\in
  S_{2}}\int^x_{-\infty}g_k(y)dy.}\end{array}$$
Now, by taking $f(x)\:=f'(x)h(x)$, using inequality (1) of Lemma
\ref{lem:integralineq} and by taking
$k=0$ in that inequality, we get property (1) for the path $\exec$
 using the same simplification/reasoning used in the base case and by
 observing that 
 $$\begin{array}{lcl}G_{p}(x_0+1)&=&\displaystyle{\int^\infty_{x_{0}+1} g_0(x) \prod_{k\in
     S_1}\int^{v_{k}}_{\max(x,u_{k})}g_k(y)dy}\\
     &&\hspace*{1cm}\displaystyle{\prod_{k\in S_{2}}\int^x_{-\infty}g_k(y)dy
   \pathprob{x,\exec||i} dx}.
   \end{array}$$ 
We can similarly prove the inductive step when the suffix ${\exec||i}$
satisfies the other properties (i.e., 2 through 5) of the Lemma.

Now consider the sub-case when a cycle transition appears before
$t_i.$ Assume  that
the cycle transitions are {\gcycle} transitions. Let $S_1'$ be the
set of $k\in S_1$ such that $t_k$ is a {\gcycle} transition and 
$S_1''=S_1\setminus S_1'.$ Since $\cA$ is  well-formed, using
Proposition \ref{prop:abs}, we see that
$o_0\in \outalph$, and for every $k\in S_1'\cup S_2$, $o_k\in \outalph.$
Let 
$f(x)\:=f'(x)h(x)$ where $f'(x)\:=g'_0(x) \prod_{k\in S_{2}}\int^x_{-\infty}g'_k(y)dy \prod_{k\in S_{1}''}\int_{\max(x,u_{k})}^{v_{k}}g'_k(y)dy$
and  $h(x) \:=\eulerv{-\weight{{\exec||i}}\epsilon}\pathprob{x+1,\exec||i}.$
If $c_i$ is also
$\getest$, then the suffix ${\exec||i}$ can satisfy any of the
conditions of the first two properties of the Lemma; In this
situation, 
observe that, if ${\exec||i}$ satisfies the condition of property (1)
then $h(x)$ is the right handside of the inequality, we get,  by
applying the inductive hypothesis to ${\exec||i}$; If ${\exec||i}$ satisfies the condition of property (2) of the Lemma
then, by applying the inductive hypothesis to ${\exec||i}$, we get 
$\pathprob{x,\exec'||i}\geq \eulerv{-\weight{{\exec||i}}\epsilon}\pathprob{x,\exec||i}$;
since, $\pathprob{x,\exec||i}\geq \pathprob{x+1,\exec||i}$, we see that 
$\pathprob{x,\exec'||i}\geq \eulerv{-\weight{{\exec||i}}\epsilon}\pathprob{x+1,\exec||i}.$
Now, assume that $c_i$ is $\lttest.$ Now, since $\cA$ is 
well-formed, it is easy to see that the condition of property (3) of 
the Lemma cannot be satisfied. Assume that ${\exec||i}$ satisfies the
condition of property (4) of the Lemma. Let $k'$ be the smallest
integer such  that, $i\leq k'\leq n$, and either $k'=n$, or $t_{k'}$ is
an assignment transition and $c_{k'}$ is $\getest.$ Now, we see
that the path starting with $t_1$ and ending with $t_{k'-1}$ is an {\alpath}. 
Using Proposition \ref{prop:abs} and the fact that  $\cA$ is  well-formed, we see that, for all $j, i\leq
j<k'$, such that $c_j$ is $\lttest$, $o_j\in \outalph.$ Now, applying the
induction hypothesis for ${\exec||i}$, using the second part of property (4), we get 
$\pathprob{x,\exec'||i}\geq \eulerv{-\weight{{\exec||i}}\epsilon}\pathprob{x+1,\exec||i}.$
Now, if $c_i$ is $\true$, applying the induction hypothesis and using
property (5), we see that $\pathprob{x,\exec'||i}\geq
\eulerv{-\weight{{\exec||i}}\epsilon}\pathprob{x,\exec||i}$; since $\pathprob{x,\exec||i}$ is
independent of $x$, we see that 
$\pathprob{x,\exec'||i}\geq \eulerv{-\weight{{\exec||i}}\epsilon}\pathprob{x+1,\exec||i}.$
Thus, irrespective of what guard $c_i$ is, we have 
 $\pathprob{x,\exec'||i}\geq h(x).$
Now, we have $\pathprob{x_0,\exec'}\:\geq \int^\infty_{x_{0}}f'(x)h(x)\prod_{k\in
  S'_{1}}\int^\infty_{x}g'_k(z) dz dx.$
Applying the  inequality (1) of
 Lemma \ref{lem:integralineq}, we get
$\pathprob{x_0,\exec'}\:\geq \int^\infty_{x_{0}+1}f'(x-1)h(x-1)\prod_{k\in
  S'_{1}}\int^\infty_{x}g_k(z) dz dx.$
Observe that, for $k\in S_1''$,
$\int_{\max(x-1,u_{k})}^{v_{k}}g'_k(y)dy\geq
\int_{\max(x,u_{k})}^{v_{k}}g'_k(y)dy.$
Using this observation and the reasoning/simplification as in the base
case, we see that property (1) is satisfied by $\exec.$   

Now consider the situation where  the cycle transitions appearing
before $t_i$ are  {\lcycle} transitions. Now, we apply
inequality (2) of Lemma \ref{lem:integralineq} to prove that property
(2) of the Lemma is satisfied by $\exec.$
 Let $S_2'$ be the
set of $k\in S_2$ such that $t_k$ is an {\lcycle} transition and 
$S_2''=S_2\setminus S_2'.$ Since $\cA$ is  well-formed, using
Proposition \ref{prop:abs}, we see that
$o_0\in \outalph$, and for every $k\in S_1\cup S'_2$, $o_k\in \outalph.$
 Now, let $f(x)\:=f'(x)h(x)$
where $f'(x)\:=g'_0(x)\prod_{k\in
  S_{1}}\int_x^{\infty}g'_k(y)dy\prod_{k\in S_{2}''}\int_{u_{k}}^{\min(x,v_{k})}g'_k(z)dz$ and
$h(x) \:= \eulerv{-\weight{{\exec||i}}\epsilon}\pathprob{x-1,{\exec||i}}.$ Now, applying
the induction hypothesis to ${\exec||i}$, we show that 
$$\pathprob{x_0,\exec'}\geq \int^{\infty}_{x_{0}} f'(x)h(x) \prod_{k\in
  S'_{2}}\int^x_{-\infty}g'_k(y)dy dx.$$
Since $\cA$ is well-formed ${\exec||i}$ cannot satisfy the condition of
property (1). Now, consider the case when ${\exec||i}$ satisfies the
condition of property (2). 
Let $k'$ be the smallest
integer such  that, $i\leq k'\leq n$, and either $k'=n$ or $t_{k'}$ is
an assignment transition and $c_{k'}$ is $\lttest.$ Now, we see
that the path starting with $t_i$ and ending with $t_{k'-1}$ is a
{\agpath}. From this observation, using the fact that   
 $\cA$ is  well-formed and using Proposition \ref{prop:abs}, we see that, for all $j, i\leq
j<k'$, such that $c_j$ is $\getest$, $o_j\in \outalph.$ Now, applying the
induction hypothesis for ${\exec||i}$, using the second part of property (2), we get 
$\pathprob{x,\exec'||i}\geq \eulerv{-\weight{{\exec||i}}\epsilon}\pathprob{x-1,{\exec||i}}.$ If
${\exec||i}$ satisfies property (3), then we directly see from the induction hypothesis
$\pathprob{x,\exec'||i}\geq \eulerv{-\weight{{\exec||i}}\epsilon}\pathprob{x-1,{\exec||i}}.$
If ${\exec||i}$ satisfies property(4), we get the above inequality, using
the first part of the induction hypothesis and the observation that 
$\pathprob{x,\exec||i}\geq \pathprob{x-1,{\exec||i}}.$ If ${\exec||i}$ satisfies property (5)
then, we get the above inequality from the induction hypothesis and
the observation that $\pathprob{x,\exec||i}$ is independent of $x.$ 
In all the above cases, it is easy to see,
$$\pathprob{x_0,\exec'}\geq \int^{\infty}_{x_{0}} f'(x)h(x) \prod_{k\in S'_{2}}\int^x_{-\infty}g_k(y-\theta_k)dy.$$
Now, using the inequality (2) of
Lemma  \ref{lem:integralineq}, and observing that, for all $k\in
S_2''$,  $\int_{u_{k}}^{\min(x+1,v_{k})}g'_k(z)dz \geq
\int_{u_{k}}^{\min(x,v_{k})}g'_k(z)dz$, and
using simplifications and reasoning
as in the base cases, we see that property (2) of the Lemma
is satisfied by $\exec.$ 
 
\end{proof}

\fi

%\input{AppendixDecidability}
%\input{ESVT-proof.tex}

\end{document}